\documentclass[11pt]{article}
\usepackage{color}
\usepackage{graphicx}
\usepackage{amsmath}
\usepackage{appendix}

\usepackage{float}
\usepackage{amssymb}
\usepackage{enumitem}
\usepackage{mathrsfs}
\usepackage[numbers,sort&compress]{natbib}
\usepackage{amsthm}

\numberwithin{equation}{section}
\usepackage{pgf}
\usepackage{CJK}

\usepackage{graphicx}
\usepackage{tikz}
\usepackage{stmaryrd}
\usepackage{graphicx}

\usepackage[latin1]{inputenc}
\usepackage[T1]{fontenc}
\usepackage[english]{babel}
\usepackage{listings}
\usepackage{xcolor,mathrsfs,url}
\usepackage{amssymb}
\usepackage{amsmath}
\usepackage{ifthen}
\usepackage{placeins}
\newtheorem{theorem}{Theorem}
\newtheorem{lemma}{Lemma}
\newtheorem{remark}{Remark}
\newtheorem{corollary}{Corollary}
\newtheorem{Proposition}{Proposition}
\newtheorem{RHP}{RH problem}
\newtheorem{painleve}{Painlev\'e RH model}
\newtheorem{parabolic}{PC RH model}
\newtheorem{Assumption}{Assumption}

\newtheorem{dbar}{$\bar{\partial}$-problem}
\usepackage {caption}
\usepackage{float}
\usepackage{subfigure}
\usepackage[
    unicode,
    colorlinks=true,
    linkcolor=blue,
    anchorcolor=blue,
    citecolor=red
]{hyperref}
\DeclareMathOperator*{\res}{Res}

\begin{document}
\baselineskip=15pt

	\title{\Large  Soliton resolution, asymptotic stability  and Painlev\'e transcendents in     the combined  Wadati-Konno-Ichikawa and short-pulse equation }
	\author{\Large Yidan Zhang$^1$,  \  Engui Fan$^{1}$}
	\footnotetext[1]{ \  School of Mathematical Sciences  and Key Laboratory for Nonlinear Science, Fudan University, Shanghai 200433, P.R. China.}
	\footnotetext[2]{ \   Email address:\ Yidan ZHANG: \ 23110840018@m.fudan.edu.cn;  Engui FAN:  faneg@fudan.edu.cn}
	\date{ }
	\maketitle

    	\begin{abstract}
		\baselineskip=16pt
In this paper,  we develop   a Riemann-Hilbert (RH)  approach to the Cauchy problem for the combined  Wadati-Konno-Ichikawa and short-pulse (WKI-SP) equation
\begin{align}
    &u_{xt}+\alpha\left(\frac{u_x}{\sqrt{1+u_x^2}}\right)_{xxx}=\beta\left(u+\frac{1}{6}(u^3)_{xx}\right),\nonumber\\
    &u(x,t=0)=u_0(x),\nonumber
\end{align}
with initial data $u_0(x)$ belongs to a weighted Sobolev space $ H^{2,3}(\mathbb{R})$, and $\alpha, \beta \not=0$ are real constants. The solution of the Cauchy problem is first  expressed in terms of the solution of a RH problem with direct scattering transform based on the Lax pair.
Further through  a series of deformations to the  RH  problem by using the $\bar{\partial}$-generalization of Deift-Zhou steepest descent method, we obtain  the long-time asymptotic approximations to the solution   of the WKI-SP equation under a new scale $(y,t)$ in three kinds of  space-time regions. The first asymptotic  result from  the  space-time  regions   $ \xi:=y/t <-2\sqrt{3\alpha\beta}, \alpha\beta>0$ and $|\xi|<\infty,\alpha\beta<0$ with saddle points on $\mathbb{R}$, is characterized with  solitons and soliton-radiation interaction with residual error $\mathcal{O}(t^{-3/4})$;
The second asymptotic  result from   the region    $ \xi >-2\sqrt{3\alpha\beta}, \alpha\beta>0$ without saddle point on $\mathbb{R}$,  is characterized with modulation-solitons with residual error $\mathcal{O}(t^{-1})$;    The third  asymptotic  result  from  a  transition  region   $\xi \approx -2\sqrt{3\alpha\beta},\alpha\beta>0$  can be expressed in terms of the  solution of the  Painlev\'{e} \uppercase\expandafter{\romannumeral2} equation with error $\mathcal{O}(t^{-1/3-5\mu})$, where $0<\mu<1/30$. This is a new phenomena that the long-time asymptotics for the solution to the  Cauchy problem of the WKI equation and SP equation don't possess.  Our   results above  are  a verification of  the soliton resolution conjecture and asymptotic stability of N-solitons for the WKI-SP equation.
\\[6pt]
    \end{abstract}

\tableofcontents

\section{Introduction}
In this paper, we  consider  the Cauchy problem of the combined Wadati-Konno-Ichikawa and short-pulse (WKI-SP) equation
\begin{align}
    &u_{xt}+\alpha\left(\frac{u_x}{\sqrt{1+u_x^2}}\right)_{xxx}=\beta\left(u+\frac{1}{6}(u^3)_{xx}\right),\label{wkisp-1}\\
    &u(x,t=0)=u_0(x),\label{wkisp-2}
\end{align}
where the initial data $u_0(x)\in H^{2,3}(\mathbb{R})$, and $\alpha, \beta $ are real   constants.
This equation  was  found recently  in \cite{Hu-2024}, where  a novel hodograph transformation is introduced  to convert the   WKI-SP equation  \eqref{wkisp-1} into the modified Korteweg-de Vries(mKdV) and sine-Gordon equation. The WKI-SP equation   \eqref{wkisp-1}  is a compound equation of the real  Wadati-Konno-Ichikawa(WKI) equation ($\beta=0,\;u_x\to u$)
\begin{equation}
    u_t+\left[\frac{u_x}{(1+u^2)^{\frac{3}{2}}}\right]_{xx}=0\label{wki}
\end{equation}
and the short-pulse (SP) equation ($\alpha=0$)
\begin{equation}
    u_{xt}=u+\frac{1}{6}\left(u^3\right)_{xx}.\label{sp}
\end{equation}

The WKI equation \eqref{wki} and another  type complex  WKI equation
\begin{equation}
   i u_t+\left[\frac{u}{\sqrt{1+|u|^2}}\right]_{xx}=0\label{wki2}
\end{equation}
were  proposed by Wadati et al. in 1979 \cite{wki-1,wki-2}. The WKI equation can be used to describe nonlinear transverse oscillations of elastic beams under tension \cite{ichikawa-1981}.  Since then, there are many significant work about the WKI hierarchy.  Shimizu and Wadati first studied the WKI equation \eqref{wki2} by the inverse scattering transform.   Wadati, Konno  and Ichikawa considered a modified  version of
\eqref{wki} and obtained a loop soliton solution \cite{KIW-1981}.   The WKI equation can also be seen from the motion of non-stretching plane curves in $\mathbb{E}^2$ \cite{chou-2002,qu-2005}. Starting from a WKI spectrum problem, the  Lenard  gradient sequence method was used to derive  the WKI hierarchy, which further is non-linearized into an Hamilton system by Bargmann constraint between the
potentials and the eigenfunctions \cite{qiao-1993,qiao-1995}. The Darboux transformation is derived in Zhang et al. \cite{zhang-2017}, thus a $\textup{sl}(2)$ WKI spectral problem was also generalized to a $\textup{so}(3)$ one in studies \cite{ma-2014,ma-2020-1,ma-2020-2}. The direct scattering data problem of the Wadati-Konno-Ichikawa equation  \eqref{wki2} with box-like initial value was solved in \cite{tu-2021}. The long-time asymptotics of the solution of the initial value problem for the potential WKI  equation are obtained by using the nonlinear steepest descent method \cite{CGW}.  Recently, Li,  Tian and Yang  obtained long-time and  the soliton resolution for the WKI equation \eqref{wki2} with both zero boundary conditions and non-zero boundary conditions  \cite{LTY-2022,li-2022}.

The SP equation \eqref{sp}   was proposed by Sch\"afer and Wayne to describe the propagation of ultra-short optical pulses in silica optical fibers \cite{sch-2004}. It turns out that the SP equation made its first appearance in Rabelo's paper in his study  of pseudospherical surfaces \cite{rabelo-1989}. It has been shown that the SP equation \eqref{sp} is related to the sine-Gordon equation through a chain of transformations \cite{sakovich-2005}. The bi-Hamilton structure and the conservation laws were studied by Brunelli \cite{brunelli-2005,brunelli-2006}. Moreover, integrable semi-discrete and full-discrete analogues \cite{feng-2010}, well-posedness of the Cauthy problem \cite{coclite-2015,pelinovsky-2010} and Riemann-Hilbert(RH) approach also have been considered \cite{boutet-2017}.    Feng    proposed  a complex short pulse equation and a coupled complex short equation to describe ultra-short pulse propagation in optical fiber \cite{Feng2}.  Further the inverse scattering transform  is developed  for the complex SP  equation  on the line with zero boundary conditions \cite{Feng3}. Using the method of testing by wave packets, Okamato discovered the unique global existence of small solutions to the equation \eqref{sp} under small initial data \cite{okamoto-2017}. Xu and Fan obtained  the  long-time asymptotic behavior of the solution of the initial value problem  for  both
 SP  equation and complex SP equation without solitons \cite{XF1,XF2}.    Yang and Fan gave the long-time asymptotics for the SP equation with  initial data in the weighted Sobolev space by using $\bar\partial$-steepest descent method  \cite{YF1}.

 This method,   as  a $\bar\partial$-generalization of    the Deift-Zhou steepest descent method \cite{zhou-1993},  was first presented by McLaughlin and Miller to analyze the asymptotics of orthogonal polynomials with non-analytical weights  \cite{mclaughlin-2006,mclaughlin-2008}. Later, Dieng and McLaughin used it to study long-time asymptotics for the defocusing nonlinear Schr\"odinger nonlinear(NLS) and focusing NLS equations under essentially minimal regularity assumptions on finite mass initial data
\cite{dieng-2019}. Cussagna and Jenkins studied the asymptotic stability  of N-soliton solutions for defocusing NLS equation with finite density initial data \cite{cuccagna-2016}. Jenkins et al. proved  soliton resolution conjecture
for the derivative   NLS equation with generic initial data in a weighted Sobolev space \cite{jenkins-2018}. In recent years, the  $\bar\partial$-steepest descent method also has been    successfully applied to obtain long-time asymptotics of  focusing NLS equation and  modified Camassa-Holm(mCH)  equation  \cite{borghese-2018-2,YF2}.

The appearance of transition   regions for integrable systems was first understood in the case of the Korteweg-de Vries(KdV)   equation by Segur and Ablowitz \cite{AblwzP1}, for which the asymptotics is described in terms of Painlev\'e
transcendents.  Later,  Painlev\'e asymptotics   as the connection between different regions was found in  the mKdV   equation   by Deift and Zhou   \cite{zhou-1993}.
Boutet de Monvel, Its, and Shepelsky found the Painlev\'{e}-type asymptotics  of the Camassa-Holm(CH)   equation by the Deift-Zhou steepest descent method \cite{Monvel}.
The connection  between    the  tau-function of the Sine-Gordon  reduction and the Painlev\'e III equation was given   by the RH  approach \cite{Its3}.
 Charlier and   Lenells carefully  considered the Airy and higher order Painlev\'{e} asymptotics  {of} the mKdV equation  \cite{Charlier}.
    Huang and  Zhang  obtained
  Painlev\'{e} asymptotics   for the whole  mKdV hierarchy \cite{Huanglin}. More recently,  the Painlev\'e asymptotics is found appearing in   the defocusing NLS  equation  and the mCH equation with non-zero boundary conditions   \cite{wfp,xyz}.

The purpose of our paper is  to establish the  RH problem   associated with the Cauchy problem for the WKI-SP equation   \eqref{wkisp-1}-\eqref{wkisp-2} with $\alpha, \beta \not=0$ and further  apply the $\bar\partial$-steepest descent method to study its  long-time asymptotics    in different space-time regions, including Painlev\'e asymptotics in a transition region.

\begin{remark}
In this paper we only need to  consider the WKI-SP equation \eqref{wkisp-1}  with  $\alpha>0,\beta>0$ and   $\alpha<0, \beta>0,$ since   by changing variable  $t\to -t$,  these two cases are the same with  the WKI-SP equation \eqref{wkisp-1}  with    $ \alpha <0, \beta <0$  and  $ \alpha > 0, \beta<0$, respectively.
\end{remark}

  Compared with the asymptotic results obtained
  for WKI equation \eqref{wki2} in  \cite{li-2022} and short pulse equation  \eqref{sp} in \cite{YF1},  our paper has the  following highlights   need to be mentioned:
\begin{itemize}
    \item  Considering that the Lax pair \eqref{lax0} of the WKI-SP equation  \eqref{wkisp-1} has two singularities at $k=0$ and $k=\infty$, we not only need to study the behavior of the solutions of spectral problem \eqref{lax0} as spectral parameter $k=0$, but also as spectral parameter $k=\infty$. Moreover, we reconstruct the solution of the WKI-SP equation
 with the asymptotics of a RH problem  as  $k\to0$, introducing a new scale $y$.

 \item As we need to consider  the asymptotics of $k\to0$ for the $\bar\partial$-problem $M^{(3)}(k)$, which may encounter the singularity $k=0$, to overcome this difficulty and reconstruct the solution form the $k^{-1}$ term, we construct the extension functions in a different way in Proposition \ref{p11}, which makes sure that $|\bar\partial R_{\ell}|\lesssim|k|$ near $k=0$. Also, for the estimates of $M^{(3)}(k)$, we consider when near $k=0$ and away from $k=0$ respectively. For this purpose, we establish the scattering map from initial data $u_0(x)\in H^{2,3}(\mathbb{R})$ to reflection coefficient $r(k)\in H^3(\mathbb{R})\cap H^{1,1}(\mathbb{R})$.
 \item
In  the cases of  the Cauchy problem for the  short pulse equation \eqref{sp}  and WKI equation \eqref{wki2}, there is no transition regions or Painlev\'e asymptotics  \cite{li-2022,YF1}, however we find a new phenomena that a transition region  $y/t\approx -2\sqrt{3\alpha\beta},\alpha\beta>0$ appears between different asymptotic regions of the solutions to the  Cauchy problem of the WKI-SP equation \eqref{wkisp-1}-\eqref{wkisp-2} there exists.  The long-time asymptotics in the transition region can     be expressed in terms of the  solution of the  Painlev\'{e} \uppercase\expandafter{\romannumeral2} equation with error $\mathcal{O}(t^{-1/3-5\mu})$.

  \item For the region without saddle point on $\mathbb{R}$, we also need to make sure $|\bar\partial R_{\ell}|\lesssim|k|$ near $k=0$, which means we can't open the jump line at $0$. So we choose to open the jump line at $\pm1$.

 \item  For  the case of  defocusing mKdV equation, where the reflection coefficient $ r(0) $  is real and $-1<r(0)<1$   \cite{zhou-1993} or $ r(0) $  is purely complex  but  $|r(0)|<1$ \cite{Charlier}, the corresponding  Painlev\'e RH model leads  to    a  global real solutions of  the Painlev\'{e} \uppercase\expandafter{\romannumeral2}  equation.  However, for our  WKI-SP case,   we cannot ensure that  the reflection coefficient $r(\pm k_0)$ is   real as well as   $|r(\pm k_0)|<1$.   Following the idea due to Boutet de Monvel, Its, and Shepelsky   \cite{Monvel}, we make   a   transformation to reduce the RH model to      a new Painlev\'e RH model    associated with  the Painlev\'{e} \uppercase\expandafter{\romannumeral2} equation  with  a  global pure imaginary  solution \cite{DZ2}.

\end{itemize}

\subsection{Main results}

By denoting $\xi=\frac{y}{t}$ with $y$ defined by \eqref{xtoy}, we divide the new time-space $(y,t)$ region into three kinds of regions depending on the values of parameters $\alpha,\beta,\xi$. See Figure \ref{spacetime}.  And we calculate the solution of transition region in detail, namely:
 $$  {\mathcal{P}:=\left\{ (y,t) \in \mathbb{R}\times\mathbb{R}^+: 0<\left|\frac{y}{t}+2\sqrt{3\alpha\beta}\right|t^{2/3}\leqslant C\right\}, }  $$
where $C>0$ is a constant.
We list our main results in this paper as follows.

\begin{figure}[!ht]
	\centering
    \subfigure[$\alpha,\beta>0$]{\begin{tikzpicture}[scale=1]
            \draw[blue!15, fill=blue!15] (-4.5,0)--(4.5,0)--(4.5,2.1)--(-4.5, 2.1)--(-4.5,0);

           \draw[yellow!20, fill=yellow!20] (0,0 )--(-4.5,0)--(-4.5,1.8)--(0,0);
            \draw[green!20, fill=green!20] (0,0 )--(-4.5,2.1)--(4.5,2.1)--(4.5,0);

		\draw [-> ](-5,0)--(5,0);
		\draw [-> ](0,0)--(0,2.8);
		\node    at (0.1,-0.3)  {$0$};
		\node    at (5.26,0)  { $y$};
	
		 \node  [below]  at (-4.4,2.7) {\scriptsize $\xi=-2\sqrt{3\alpha\beta}$};

  		\draw [red](0,0)--(-4.5,2);
\node    at (0,3.2)  { $t$};
        \node [] at (-2.8,0.5) {\scriptsize \uppercase\expandafter{\romannumeral1}. $4$ saddle points };
  \node [] at (-2.8,0.2) {\scriptsize region };
    \node [] at (1.8,0.8) {\scriptsize \uppercase\expandafter{\romannumeral3}. $0$ saddle point region };

        \node  []  at (-5.2,1.7) {\scriptsize  \uppercase\expandafter{\romannumeral2}. Transition}; \node  []  at (-5.2,1.4) {\scriptsize region};
		\end{tikzpicture}}

	\subfigure[$\alpha<0, \beta>0$]{\begin{tikzpicture}[scale=1]
    \hspace{0.5cm}
            \draw[yellow!15, fill=yellow!15] (-4.5,0)--(4.5,0)--(4.5,2.1)--(-4.5, 2.1)--(-4.5,0);

		\draw [-> ](-5,0)--(5,0);
		\draw [-> ](0,0)--(0,2.8);
		\node    at (0.1,-0.3)  {$0$};
		\node    at (5.26,0)  { $y$};
		\node    at (0,3.2)  { $t$};
		
     \node [] at (0,1.1) {\scriptsize \uppercase\expandafter{\romannumeral4}. $2$ saddle points region };

		\end{tikzpicture}}	
	
	\caption{\footnotesize {The space-time regions of $(y,t)-$plane, depending on the values of $\alpha,\beta,\xi$. For $\alpha, \beta>0$,  the yellow region $\xi<-2\sqrt{3\alpha\beta}$ denotes that  there are $4$ saddle points on $\mathbb{R}$, the green region $\xi>-2\sqrt{3\alpha\beta}$ denotes  there is no saddle point on $\mathbb{R}$, and the blue  region, $\xi\approx -2\sqrt{3\alpha\beta}$,  is  the transition  region. For $\alpha<0, \beta>0$, there are $2$ saddle points on $\mathbb{R}$ for  $ |\xi|<\infty$.}  }
	\label{spacetime}
\end{figure}
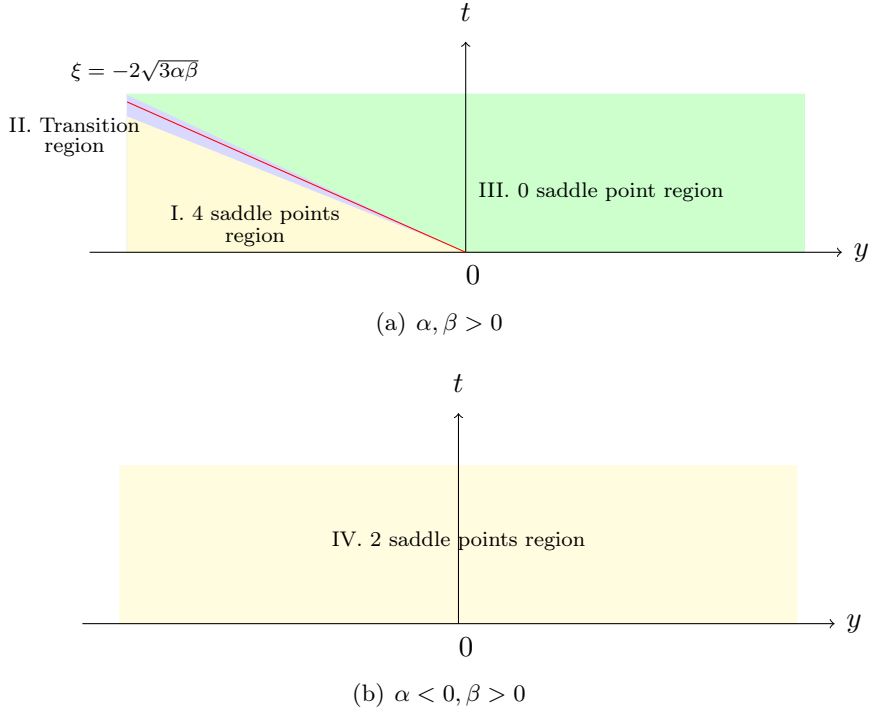

\FloatBarrier

\begin{theorem}\label{thm1}
    Let $u(x,t)$ be the solution  for  the Cauchy  problem \eqref{wkisp-1}-\eqref{wkisp-2} associated with the initial data
  $u_0(x) \in H^{2,3}(\mathbb{R})$, and $\sigma_d=\{(z_n, c_n)\}_{n=1}^N$ be the reflectionless discrete data. Then  as   $ t\to+\infty$,  we obtain the following asymptotic expansions:

\noindent\textbf{\textup{ \uppercase\expandafter{\romannumeral1}}}. In the regions $\alpha, \beta>0,\xi<-2\sqrt{3\alpha\beta}$ or $\alpha<0, \beta>0$,
\begin{align*}
    u(x,t)&=u(y(x,t),t)=u_{sol}(y(x,t), t;\sigma_d)-T_0^{2}it^{-\frac{1}{2}}f_{12}+\mathcal{O}(t^{-\frac{3}{4}}),\\
    y(x,t)&=x-c_+(x,t;\sigma_d)+T_1^{-1}+it^{-\frac{1}{2}}f_{11}+\mathcal{O}(t^{-\frac{3}{4}}),
\end{align*}
where
\begin{align*}
    f_{11} = \left[ M^{(out)}(0)^{-1}\widehat{E}_1 M^{(out)}(0) \right]_{11},\quad f_{12} = \left[ M^{(out)}(0)^{-1} \widehat{E}_1 M^{(out)}(0) \right]_{12},
\end{align*}
with
\begin{align*}
   & \widehat E_1=\sum_{j=1}^\Lambda \frac{i  }{\left[2  \eta(k_j) \theta^{\prime\prime}(k_j)\right]^{\frac{1}{2}} k_j^2}M^{(out)}(k_j)A_j^{mat}M^{(out)}(k_j)^{-1},\\ &T_0=\prod\limits_{n\in\Delta^-}\frac{\bar{z}_n}{z_n}=\mathrm{exp}\left[-2i\sum_{n\in\Delta^-}\mathrm{arg}(z_n)\right],\quad T_1=\int_I\frac{\nu(s)}{s^2}\mathrm{d}s-\sum_{n\in\Delta^-}\frac{2\mathrm{Im}(z_n)}{|z_n|^2},
\end{align*}
where $\Lambda=4$, for $\alpha,\beta>0$ and
$\Lambda=2$, for $\alpha<0,\beta>0$.

\noindent\textbf{\textup{ \uppercase\expandafter{\romannumeral2}}}. In the region $\alpha, \beta>0,\xi>2\sqrt{3\alpha\beta}$,
\begin{align*}
    u(x,t)&=u(y(x,t),t)=u_{sol}(y(x,t), t;\sigma_d)+\mathcal{O}(t^{-1}),\\
    y(x,t)&=x-c_+(x,t;\sigma_d)+iT_1^{-1}+\mathcal{O}(t^{-1}),
\end{align*}
where $T_1=-\displaystyle\sum_{n\in\Delta^-}\frac{2\mathrm{Im}(z_n)}{|z_n|^2}$.

\noindent\textbf{\textup{ \uppercase\expandafter{\romannumeral3}}}. In the region $\alpha, \beta>0,(y,t)\in\mathcal{P}$,
\begin{align*}
    u(x,t)&=u(y(z,t),t)=u_{sol}(y(x,t), t;\sigma_d)-iT_0^{2}\tau^{-\frac{1}{3}}\widehat P_{12}+\mathcal{O}(t^{-\frac{1}{3}-5\mu}),\\
    y(x,t)&=x+iT_1^{-1}+i\tau^{-\frac{1}{3}}\widehat P_{11}+\mathcal{O}(t^{-\frac{1}{3}-5\mu}),
\end{align*}
where $\mu$ is a constant with $0<\mu<1/30$ and
\begin{align*}
    \widehat P_{11}=\left[ M^{(out)}(0)^{-1}\widehat N^{(err)}_1 M^{(out)}(0) \right]_{11},\quad \widehat P_{12}=\left[ M^{(out)}(0)^{-1}\widehat N^{(err)}_1 M^{(out)}(0) \right]_{12},
\end{align*}
with
\begin{align*}
&\widehat N^{(err)}_0=\frac{1}{k_0} \left(M^{(out)}(k_0)N^{(\infty,k_0)}_1(s)M^{(out)}(k_0)^{-1}
        -\overline{M^{(out)}(k_0)}N^{(\infty,-k_0)}_1(s)\overline{M^{(out)}(k_0)^{-1}}\right),\\
        &\widehat N^{(err)}_1=\frac{1}{k_0^2}\left(M^{(out)}(k_0)N^{(\infty,k_0)}_1(s)M^{(out)}(k_0)^{-1}+\overline{M^{(out)}(k_0)}N^{(\infty,-k_0)}_1(s)\overline{M^{(out)}(k_0)^{-1}}\right).
\end{align*}
In the above formula,
\begin{align*}
&N_1^{(\infty,k_0)} ( s)
=   \frac{{i}}{2} \begin{pmatrix}
-\int_s^\infty P^2(z)\mathrm{d}z &- e^{i\varphi_0 }P(s) \\  e^{-i\varphi_0 }P(s) &   \int_s^\infty P^2(z)\mathrm{d}z
 \end{pmatrix},\\
 &N_1^{(\infty,-k_0)} ( s)
= \frac{{i}}{2} \begin{pmatrix}
  \int_s^\infty P^2(z)\mathrm{d}z &- e^{-i\varphi_0 }P(s) \\  e^{i\varphi_0 }P(s) &-\int_s^\infty P^2(z)\mathrm{d}z
 \end{pmatrix}, \\
&\varphi_0(s,t)=2\theta(k_0,\xi=-2\sqrt{3\alpha\beta})t+2k_0s\tau^{\frac{1}{3}}+\arg r(k_0)-4\sum_{n\in\Delta^-}\arg(k_0-z_n),\\
&\tau=12\alpha t,\quad s=\frac{\xi+2\sqrt{3\alpha\beta}}{12\alpha}\tau^{\frac{2}{3}},\quad k_0=\left(\frac{\beta}{48\alpha}\right)^{1/4},\\
&T_0=\prod\limits_{n\in\Delta^-}\frac{\bar{z}_n}{z_n}=\mathrm{exp}\left[-2i\sum_{n\in\Delta^-}\mathrm{arg}(z_n)\right],\quad T_1=-\sum_{n\in\Delta^-}\frac{2\mathrm{Im}(z_n)}{|z_n|^2},
\end{align*}
with $P(s)$ be a real solution of the following Painlev\'{e} \textup{\uppercase\expandafter{\romannumeral2} }equation
\begin{align*}
	P_{ss} =-2P^3 +sP, \quad s \in \mathbb{R}.
\end{align*}

\end{theorem}

\subsection{Outline of this paper}
We arrange our paper as follows. In Section \ref{sec2}, we  start from the Lax pair of WKI-SP equation \eqref{wkisp-1} for the spectral analysis for initial data
  $u_0(x) \in H^{2,3}(\mathbb{R})$ in Subsection \ref{subs2.1}. By the map between initial data and the reflection coefficient, we prove that the reflection coefficient is in a  weighted Sobolev space $r(k)\in H^3(\mathbb{R})\cap H^{1,1}(\mathbb{R})$ in Subsection \ref{subs2.2}. By introducing a new scale $y$, we set up the basic RH problem and give a classification of asymptotic regions depending on parameters $\alpha,\beta,\xi$. In Section \ref{sec3}, we deal with the long-time asymptotics in  the region \uppercase\expandafter{\romannumeral1} and \uppercase\expandafter{\romannumeral4},in which there will exist  saddle points on $\mathbb{R}$. By a series of deformations, the original RH problem is transformed into a hybrid $\bar\partial$-RH problem in Subsection \ref{subs3.1} which can be decomposed into a pure RH problem and a $\bar{\partial}$-problem.The pure RH can be solved   with two  RH models for discrete spectrum and the jump line respectively in Subsection \ref{subs3.2} and Subsection \ref{subs3.3}. While  the $\bar{\partial}$-problem is analyzed in Subsection \ref{subs3.4}. In section \ref{sec5}, we deal with the region \uppercase\expandafter{\romannumeral3}, which has no saddle point on $\mathbb{R}$. We open the jump line at $\pm1$ and get a hybrid $\bar{\partial}$-RH problem in Subsection \ref{sec5.1}, then we operate the analysis on the pure RH problem and pure $\bar{\partial}$-problem in Subsection \ref{sec5.2} and Subsection \ref{sec5.3} respectively. In Section \ref{sec4}, we deal with the transition region \uppercase\expandafter{\romannumeral2}. We first modify the basic RH problem and deform it into a  hybrid $\bar\partial$-RH problem in Subsection \ref{sec4.1},
   which can be solved   by decomposing it into a pure RH problem in Subsection \ref{sec4.3}  and a pure $\bar\partial$-problem in Subsection \ref{sec4.4}. The  RH problem  for the pure RH problem can be  constructed by the outer discrete spectrum model and a solvable Painlev\'e model  via the local paramatrix near the saddle points, and the residual error   comes  from a small normed RH problem.

\subsection{Some notations}
Here we present some notations used through out this paper.
\begin{itemize}
    \item In this paper, $\sigma_1,\sigma_2,\sigma_3$ denote the Pauli matrices
    \begin{equation*}
       \sigma_1= \left(\begin{array}{cc}
            0 & 1 \\
            1 & 0
        \end{array}\right),\quad \sigma_2= \left(\begin{array}{cc}
            0 & -i \\
            i & 0
        \end{array}\right),\quad \sigma_3= \left(\begin{array}{cc}
            1 & 0 \\
            0 & -1
        \end{array}\right).
    \end{equation*}

    \item A weighted space $L^{p,s}(\mathbb{R})$ is defined by
    \begin{equation*}
        L^{p,s}(\mathbb{R})=\left\{f(x)\in L^{p}(\mathbb{R}):\langle x \rangle ^sf(x)\in L^{p}(\mathbb{R})\right\},
    \end{equation*}
    with the norm $\|f\|_{ L^{p,s}(\mathbb{R})}=\|\langle x \rangle ^sf(x)\|_{ L^{p}(\mathbb{R})}$.
    \item A Sobolev space is defined by
    \begin{equation*}
        W^{m,p}=\{f(x)\in L^{p}(\mathbb{R}):\partial^jf(x)\in L^{p}(\mathbb{R}),\,j=1,\dots,m\},
    \end{equation*}
     with the norm $\|f\|_{ W^{m,p}(\mathbb{R})}=\displaystyle\sum_{j=0}^m\|\partial^j f(x)\|_{L^{p}(\mathbb{R})}$. Usually, we are used to expressing $H^m(\mathbb{R})=W^{m,2}(\mathbb{R})$.
    \item A weighted Sobolev space is defined by
    \begin{equation*}
        H^{m,s}(\mathbb{R})=\{f(x)\in L^{2}(\mathbb{R}):\langle x \rangle ^s\partial^j f(x)\in L^{2}(\mathbb{R}),\,j=1,\dots,m\}= L^{2,s}(\mathbb{R})\cap H^m(\mathbb{R}).
    \end{equation*}
    In this paper, we define the initial data $u_0(x)\in H^{2,3}(\mathbb{R})$.
    \item In this paper, we frequently use $a\lesssim b,a \gtrsim b$ to denote $a\leqslant Cb,a\geqslant C'b$ for constants $C,C'>0$.
\end{itemize}

\section{Inverse scattering transform and RH problem}\label{sec2}
\subsection{Spectral analysis}\label{subs2.1}

The WKI-SP equation \eqref{wkisp-1} admits the following Lax pair \cite{Hu-2024}:
\begin{equation}\label{lax0}
  \Phi_x=U\Phi,\quad \Phi_t=V\Phi,
\end{equation}
where
\begin{equation}
    U=\left(\begin{array}{cc}
        ik & iku_x \\
        iku_x & -ik
    \end{array}\right),\quad V=\left(\begin{array}{cc}
        A  & B  \\
        C  & -A
    \end{array}\right),
\end{equation}
with
\begin{eqnarray*}
    &&A =\frac{4\alpha}{\sqrt{m}}ik^3+\frac{\beta ik}{2}u^2-\frac{\beta i}{4k},\\
    &&B =2\alpha k^2\left(\frac{u_x}{\sqrt{m}}\right)_x-\frac{\beta u}{2}+\frac{1}{4ik}\left[4\alpha k^2\left(\frac{u_x}{\sqrt{m}}\right)_{xx}-\beta u_x\right]+u_xA,\\
     &&C =-2\alpha k^2\left(\frac{u_x}{\sqrt{m}}\right)_x+\frac{\beta u}{2}+\frac{1}{4ik}\left[4\alpha k^2\left(\frac{u_x}{\sqrt{m}}\right)_{xx}-\beta u_x\right]+u_xA,
\end{eqnarray*}
and $m=1+u_x^2$. From the symmetry of $U(x,t;k)$, we can find that $\Phi(x,t;k)$ holds the symmetries that
\begin{equation}
    \Phi(k)=\sigma_2\Phi(-k)\sigma_2=\sigma_2\overline{\Phi(\bar{k})}\sigma_2.
\end{equation}

The Lax pair \eqref{lax0} for the WKI-SP equation has singularities at $k=0,k=\infty$, so the asymptotic behaviors of their eigenfunctions should be controlled. Following the idea due to Boutet de Monvel \cite{boutet-2017}, we need to analyze these singularities respectively. First, we start from $k=0$.

\noindent\textbf{When $\mathbf{k=0}.$}  We rewrite  the Lax pair \eqref{lax0} as
\begin{eqnarray}
   &&\Phi_x-ik\sigma_3\Phi=U_0\Phi\label{lax1-1},\\
   &&\Phi_t-ik\left(4\alpha k^2-\frac{\beta}{4k^2}\right)\sigma_3\Phi=V_0\Phi,\label{lax1-2}
\end{eqnarray}
where
\begin{align*}
    U_0&=iku_x\sigma_1,\\
    V_0&=\frac{\beta}{2}u^2U_0+4\alpha ik^3\left(\frac{1}{\sqrt{m}}-1\right)\sigma_3\\
    &+\left[2\alpha ik^2\left(\frac{u_x}{\sqrt{m}}\right)_x-\frac{\beta i}{2}u\right]\sigma_2+\left[4\alpha ik^3\frac{u_x}{\sqrt{m}}-\alpha ik\left(\frac{u_x}{\sqrt{m}}\right)_{xx}\right]\sigma_1.
\end{align*}
Take the transformation
\begin{equation}
    \mu^0=\Phi e^{-ik\left[x+t\left(4\alpha k^2-\frac{\beta}{4k^2}\right)\right]\sigma_3},
\end{equation}
then
$$\mu^0 \to I,\;x\to\pm\infty,$$ and the Lax pair \eqref{lax1-1}-\eqref{lax1-2} becomes
\begin{eqnarray}
    &&\mu^0_x-ik\left[\sigma_3,\mu^0\right]=U_0\mu^0,\label{lax1-3}\\
    &&\mu^0_t-ik\left(4\alpha k^2-\frac{\beta}{4k^2}\right)\left[\sigma_3,\mu^0\right]=V_0\mu^0,\label{lax1-4}
\end{eqnarray}
which  can be written as
\begin{equation}
    \mathrm{d}\left(e^{-ik[x+(4\alpha k^2-\frac{\beta}{4k^2})t]\hat{\sigma}_3}\mu^0\right)=W^0(x,t;k),\label{oneform1}
\end{equation}
where $W^0(x,t;k)$ is the closed one-form defined by
\begin{equation}
    W^0(x,t;k)=e^{-ik[x+(4\alpha k^2-\frac{\beta}{4k^2})t]\hat{\sigma}_3}\left(U_0\mathrm{d}x+V_0\mathrm{d}t\right)\mu^0.
\end{equation}
We obtain  two eigenfunctions $\mu_{\pm}^0 $ from  \eqref{oneform1} by the Volterra integral equations
\begin{equation}
\mu_{\pm}^0(x,t;k)=I+\int_{\pm\infty}^xe^{ik(x-y)\hat{\sigma}_3}\left[U_0(y,t;k)\mu_{\pm}^0(y,t;k)\right]\mathrm{d}y,
\end{equation}
by  which we can show that
\begin{Proposition}\label{p1}
    From the definition of $\mu^0_{\pm}(k)$, with $u_0(x)\in H^{2,3}(\mathbb{R})$, we find that they hold the following analytic properties
    \begin{enumerate}
    [label=\textnormal{(\arabic*)}]
        \item $\left[\mu_{+}^0(k)\right]_1$ and $\left[\mu_{-}^0(k)\right]_2$ are analytical in $\mathbb{C}^+$,
        \item $\left[\mu_{+}^0(k)\right]_2$ and $\left[\mu_{-}^0(k)\right]_1$ are analytical in $\mathbb{C}^-$,
    \end{enumerate}
    where $\left[\mu^0_{\pm}(k)\right]_i$ denotes the i-th column of $\mu^0_{\pm}(k)$.
\end{Proposition}

When $k\to0$, from Lax pair \eqref{lax1-3}-\eqref{lax1-4}, $\mu^0(k)$ has the following asymptotic expansion
\begin{equation}
    \mu^0(k)=I+iu\sigma_1k+\mathcal{O}(k^2),\quad k\to0.\label{mu0}
\end{equation}

\noindent\textbf{When} $\mathbf{k=\infty}.$  In order to control the asymptotic behavior of the Lax pair when $k\to\infty$, by  introducing  a matrix function
\begin{equation}
    Q(x,t)=\sqrt{\frac{\sqrt{m}+1}{2\sqrt{m}}}\left(\begin{array}{cc}
        1 &\frac{u_x}{\sqrt{m}+1}  \\
        -\frac{u_x}{\sqrt{m}+1} &1
    \end{array}\right),
\end{equation}
and taking  the transformation $\Psi=Q\Phi$, we obtain a new Lax pair:
\begin{align}
    &\Psi_x-ik\sqrt{m}\sigma_3\Psi=U_1\Psi,\label{lax2-1}\\
    &\Psi_t-ik\left[\frac{\beta}{2}u^2\sqrt{m}+\alpha\left(\frac{1}{2}\left(\frac{u_x}{\sqrt{m}}\right)_x^2-\frac{u_x}{\sqrt{m}}\left(\frac{u_x}{\sqrt{m}}\right)_{xx}\right)-\frac{\beta}{4k^2}+4\alpha k^2\right]\sigma_3\Psi=V_1\Psi,\label{lax2-2}
\end{align}
where
\begin{align*}
    U_1&=\frac{iu_{xx}}{2m}\sigma_2,\\
    V_1&=-\frac{\beta i}{4k}\left(\frac{1}{\sqrt{m}}-1\right)\sigma_3+\left[\frac{\beta iu^2u_{xx}}{4m}+2\alpha ik^2\left(\frac{u_x}{\sqrt{m}}\right)_x-\frac{\beta i}{2}u\right]\sigma_2\\
    &-\frac{\alpha ik}{2}\left(\frac{u_x}{\sqrt{m}}\right)_x^2\sigma_3+\left[\frac{\beta iu_x}{4k\sqrt{m}}-\frac{\alpha ik}{\sqrt{m}}\left(\frac{u_x}{\sqrt{m}}\right)_{xx}\right]\sigma_1.
\end{align*}
Define
\begin{equation}
    p(x,t;k)=x-\int_x^\infty\left(\sqrt{m(s,t)}-1\right)\mathrm{d}s-\frac{\beta t}{4k^2}+4\alpha k^2t.\label{p}
\end{equation}
As we can rewrite the WKI-SP equation \eqref{wkisp-1} into the conservation law form:
\begin{equation}
    (\sqrt{m})_t=\left[\frac{1}{2}\beta u^2\sqrt{m}+\alpha\left(\frac{1}{2}\left(\frac{u_x}{\sqrt{m}}\right)_x^2-\frac{u_x}{\sqrt{m}}\left(\frac{u_x}{\sqrt{m}}\right)_{xx}\right)\right]_x ,
\end{equation}
 then function $p(x,t;k)$ defined in \eqref{p} satisfies the compatibility condition $p_{xt}=p_{tx}$, which implies that
\begin{align*}
    p_x&=\sqrt{m},\\
    p_t&=\frac{1}{2}\beta u^2\sqrt{m}+\alpha\left(\frac{1}{2}\left(\frac{u_x}{\sqrt{m}}\right)_x^2-\frac{u_x}{\sqrt{m}}\left(\frac{u_x}{\sqrt{m}}\right)_{xx}\right)-\frac{\beta }{4k^2}+4\alpha k^2.
\end{align*}

Take the transformation
\begin{equation}
    \Psi(x,t;k)=Q^{-1}(x,t;k)\mu (x,t;k)e^{ikp(x,t;k)\sigma_3},\label{defPhi}
\end{equation}
we obtain the following Lax pair:
\begin{align}
    &\mu_x-ikp_x\left[\sigma_3,\mu\right]=U_1\mu,\label{lax2-3}\\
    &\mu_t-ikp_t\left[\sigma_3,\mu\right]=V_1\mu,\label{lax2-4}
\end{align}
with $\mu\to I,x\to\pm\infty$. The Lax pair \eqref{lax2-3}-\eqref{lax2-4} can be written into a total differential form
\begin{equation}
    \mathrm{d}\left(e^{-ikp\hat{\sigma}_3}\mu\right)=e^{-ikp\hat{\sigma}_3}\left(U_1\mathrm{d}x+V_1\mathrm{d}t\right)\mu,
\end{equation}
which leads to two Volterra type integrals
\begin{equation}
    \mu_{\pm}(x,t;k)=I+\int_{\pm\infty}^xe^{ik(p(x)-p(y))\hat{\sigma}_3}\left[U_1(y,t;k)\mu_{\pm}(y,t;k)\right]\mathrm{d}y.\label{mupm}
\end{equation}
Denote $\mu_{\pm}(k)=\left(\left[\mu_{\pm}(k)\right]_1,\left[\mu_{\pm}(k)\right]_2\right)$, we can obtain the following proposition.
\begin{Proposition}\label{p2}
    Let the initial data $u_0(x)\in H^{2,3}(\mathbb{R})$, then we have
    \begin{enumerate}
    [label=\textnormal{(\arabic*)}]
        \item $\left[\mu_{+}(k)\right]_1$ and $\left[\mu_{-}(k)\right]_2$ are analytical in $\mathbb{C}^+$,
 $\left[\mu_{+}(k)\right]_2$ and $\left[\mu_{-}(k)\right]_1$ are analytical in $\mathbb{C}^-$,
 \item $ \mu_{\pm}(k)=\sigma_2\mu_{\pm}(-k)\sigma_2=\sigma_2\overline{\mu_{\pm}(\bar{k})}\sigma_2$.
    \end{enumerate}
\end{Proposition}

 As $\mu_+$ and $\mu_-$ are two fundamental matrix solutions of the Lax pair \eqref{lax2-3}-\eqref{lax2-4}, which means there exists a matrix $S(k)$, such that
 \begin{equation}
     \mu_-(x,t;k)=\mu_+(x,t;k)e^{ikp\hat{\sigma}_3}S(k),\label{rmu}
 \end{equation}
where, by  the symmetry of $\mu_{\pm}(k)$,  $S(k)$ can be written  as
 \begin{equation*}
     S(k)=\left(\begin{array}{cc}
         \overline{a(\bar{k})} & b(k) \\
         -\overline{b(\bar{k})} & a(k)
     \end{array}\right),\quad k\in\mathbb{C},
 \end{equation*}
and $a(k)=\overline{a(-\bar{k})}$.

Moreover, the equation \eqref{rmu} implies that
 \begin{align}
     a(k)&=\mathrm{det}\left(\left[\mu_{+}(k)\right]_1,\left[\mu_{-}(k)\right]_2\right),\label{ak}\\
b(k)&=e^{-2ikp}\mathrm{det}\left(\left[\mu_{-}(k)\right]_2,\left[\mu_{+}(k)\right]_2\right),\label{bk}
 \end{align}
 which means $a(k)$ is analytical in $\mathbb{C}^+$. Introduce the reflection coefficient
 \begin{equation}
     r(k)=\frac{b(k)}{a(k)}.
 \end{equation}

To construct the RH problem $M(k)$ (see RH problem \ref{rhp1}), we need to use the eigenfunctions $\mu_{\pm}$. While to reconstruct the solution $u(x,t)$, we need the asymptotic behavior of $\mu_{\pm}$ as $k\to0$. For this purpose, we need to relate $\mu_{\pm}$ to $\mu_{\pm}^0$.

 \begin{Proposition}\label{p3}
     The functions $\mu_{\pm}(x,t;k)$ and $\mu_{\pm}^0(x,t;k)$ can be related as:
     \begin{align}
          \mu_{+}(x,t;k)&=Q(x,t)\mu^0_+(x,t;k)e^{ik\int^{+\infty}_{x}\left(\sqrt{m(s,t)}-1\right)\mathrm{d}s\sigma_3},\\
          \mu_{-}(x,t;k)&=Q(x,t)\mu^0_-(x,t;k)e^{-ik\int_{-\infty}^{x}\left(\sqrt{m(s,t)}-1\right)\mathrm{d}s\sigma_3}.
     \end{align}
 \end{Proposition}
 \begin{proof}
     As $\mu^0_{\pm}$ and $\mu_{\pm}$ are derived from the same Lax pair \eqref{lax0}, then there exists constant matrices $C_{\pm}(k)$ satisfying
     \begin{equation}
         \mu_{\pm}(x,t;k)=Q(x,t)\mu^0_{\pm}(x,t;k)e^{-ik[x+(4\alpha k^2-\frac{\beta}{4k^2})t]\sigma_3}C_{\pm}(k)e^{-ikp\sigma_3}.
     \end{equation}
     Take $x\to\pm\infty$ respectively, we can obtain
     \begin{equation}
         C_+=I,\quad C_-=e^{ikc\sigma_3},
     \end{equation}
     where $c=\int^{+\infty}_{-\infty}(\sqrt{m(s,t)}-1)\mathrm{d}s$.
 \end{proof}
From Proposition \ref{p3} and expansion \eqref{mu0}, $a(k)$ has the following asymptotic expansion as $k\to0$
 \begin{equation}
     a(k)=1+ick+\mathcal{O}(k^2).\label{expanak}
 \end{equation}

\subsection{Reflection coefficient}\label{subs2.2}

In this part, we discuss the relationship between the initial data $u_0(x) $ and the reflection coefficient $r(k)$.  For this purpose, we first prove the following three lemmas.

Denote $\mu_{\pm}(x,k)=\left(\mu_{jk}^{\pm}(x,k)\right)$ as the solutions of \eqref{mupm} for $t=0$, and further define a vector function
\begin{equation}
    \textbf{n}^{\pm}(x,k)=(n^{\pm}_{11}(x,k),n^{\pm}_{21}(x,k))^T=(\mu_{11}^{\pm}(x,k)-1,\mu_{21}^{\pm}(x,k))^T.\label{defn}
\end{equation}
By \eqref{mupm} and \eqref{defn}, we have
\begin{equation}
\textbf{n}(x,k)=\textbf{n}_0(x,k)+T\textbf{n}(x,k),\label{n}
\end{equation}
where $T$ is an integral operator defined by
\begin{equation}
T\textbf{f}(x,k)=\int_x^{+\infty}K(x,y,k)\textbf{f}(y,k)\mathrm{d}y,\label{defT}
\end{equation}
with the kernel
\begin{equation}
K(x,y,k)=\left(\begin{array}{cc}0&-\frac{u_{yy}}{2m}\\\frac{u_{yy}}{2m}e^{2ik\left[h(x)-h(y)\right]}&0\end{array}\right),\label{K}
\end{equation}
and
\begin{equation}
\textbf{n}_0(x,k)=T\textbf{e}_1=\left(\begin{array}{c}0\\ \int_x^{+\infty}\frac{u_{yy}}{2m}e^{2ik\left[h(x)-h(y)\right]}\mathrm{d}y\end{array}\right).
\end{equation}
Here the function $h(x)$ is defined as

\begin{equation*}
h(x)=\int_x^{\infty}\sqrt{m(s,0)}\mathrm{d}s,
\end{equation*}
and thus
\begin{equation*}
h(x)-h(y)=\int_x^y\sqrt{m(s,0)}\mathrm{d}s.
\end{equation*}

Taking the partial derivatives of $k$ for \eqref{n}, we get
\begin{eqnarray}
&&(\textbf{n})_k=\textbf{n}_1+T(\textbf{n})_k,\quad  \textbf{n}_1=(\textbf{n}_0)_k+(T)_k\textbf{n},\label{nk}\\
&&(\textbf{n})_{kk}=\textbf{n}_2+T(\textbf{n})_{kk},\quad \textbf{n}_2=(\textbf{n}_0)_{kk}+(T)_{kk}\textbf{n}+2(T)_k(\textbf{n})_k,\label{nkk}\\
&&(\textbf{n})_{kkk}=\textbf{n}_3+T(\textbf{n})_{kkk},\ \textbf{n}_3=(\textbf{n}_0)_{kkk}+(T)_{kkk}\textbf{n}+3(T)_{kk}(\textbf{n})_k+3(T)_k(\textbf{n})_{kk},\qquad\label{nkkk}
\end{eqnarray}
To find the solutions of the differential equations\eqref{n}, \eqref{nk}, \eqref{nkk} and \eqref{nkkk}, we need several lemmas as follows:

\begin{lemma}\label{lem1}
For $u_0(x) \in H^{2,3}(\mathbb{R})$, the following estimates hold:
\begin{equation}\label{n0est}
\|\textbf{n}_0\|_{C^0(\mathbb{R}^+,L^2(\mathbb{R}))}\lesssim \| u_{xx}\|_{L^{2}},\quad \| \textbf{n}_0\|_{L^2(\mathbb{R}^+\times\mathbb{R})}\lesssim \| u_{xx}\|_{L^{2,\frac{1}{2}}};
\end{equation}
\begin{equation}\label{n0exp}
\begin{aligned}
&\| (\textbf{n}_0)_k\|_{C^0(\mathbb{R}^+,L^2(\mathbb{R}))}\lesssim \| u_{xx}\|_{L^{2,1}}+\| u\|_{H^{1}}\| u_{xx}\|_{L^{2,\frac{1}{2}}},\\ &\| (\textbf{n}_0)_k\|_{L^2(\mathbb{R}^+\times\mathbb{R})}\lesssim \| u_{xx}\|_{L^{2,\frac{3}{2}}}+\| u\|_{H^{1}}\| u_{xx}\|_{L^{2,1}};\\
\end{aligned}
\end{equation}
\begin{equation}
\begin{aligned}
&\| (\textbf{n}_0)_{kk}\|_{C^0(\mathbb{R}^+,L^2(\mathbb{R}))}\lesssim \| u_{xx}\|_{L^{2,2}}+\| u\|_{H^{1}}\| u_{xx}\|_{L^{2,\frac{3}{2}}}+\| u\| ^2_{H^{1}}\| u_{xx}\|_{L^{2,1}},\\
&\| (\textbf{n}_0)_{kk}\|_{L^2(\mathbb{R}^+\times\mathbb{R})}\lesssim \| u_{xx}\|_{L^{2,\frac{5}{2}}}+\| u\|_{H^{1}}\| u_{xx}\|_{L^{2,2}}+\| u\| ^2_{H^{1}}\| u_{xx}\|_{L^{2,\frac{3}{2}}};
\end{aligned}
\end{equation}
\begin{equation}
\begin{aligned}
\|(\textbf{n}_0)_{kkk}\|_{C^0(\mathbb{R}^+,L^2(\mathbb{R}))}\lesssim\| u_{xx}\|_{L^{2,3}}+\| u\|_{H^{1}}\| u_{xx}\|_{L^{2,\frac{5}{2}}}+\| u\|^2_{H^{1}}\|u_{xx}\|_{L^{2,2}}+\| u\| ^3_{H^{1}}\| u_{xx}\|_{L^{2,\frac{3}{2}}}.\\
\end{aligned}
\end{equation}
\end{lemma}
\begin{proof}
We take the proof of \eqref{n0exp} for example, and the rest can be proved similarly.

Take the derivative of $\textbf{n}_0(x,k)$ on $k$, we get
\begin{equation*}
(\textbf{n}_0)_k(x,k)=\left(\begin{array}{c}0\\ 2i\left[h(x)-h(y)\right]\int_x^{+\infty}\frac{u_{yy}}{2m}e^{2ik\left[h(x)-h(y)\right]}\mathrm{d}y\end{array}\right).
\end{equation*}
Considering that for $y>x$, by H\"older equality we can obtain
\begin{equation*}
h(x)-h(y)=\int_x^y\sqrt{u^2_{s}+1}\mathrm{d}s\leqslant(y-x)+(y-x)^{1/2}\| u\|_{H^{1}},
\end{equation*}
we deduce that for any function $\varphi(k)\in L^2(\mathbb{R})$ satisfying $\| \varphi\|_{L^2}=1$,
\begin{equation*}
\begin{aligned}
&\| (\textbf{n}_0)_k\|_{L^2(\mathbb{R})}=\sup_{\varphi\in L^2(\mathbb{R})}\int_0^{\infty}2i\left[h(x)-h(y)\right]\int_x^{+\infty}\frac{u_{yy}}{2m}e^{2ik\left[h(x)-h(y)\right]}\varphi(k)\mathrm{d}y\mathrm{d}k\\
\lesssim&\sup_{\varphi\in L^2(\mathbb{R})}\left(\int_x^{+\infty}\frac{(y-x)u_{yy}}{2m}\widehat{\varphi}(h(y)-h(x))\mathrm{d}y
+\| u\|_{H^{1}}\int_x^{+\infty}\frac{(y-x)^{1/2}u_{yy}}{2m}\widehat{\varphi}(h(y)-h(x))\mathrm{d}y\right)\\
\lesssim&\left(\int_x^{+\infty}|yu_{yy}|^2\mathrm{d}y\right)^{1/2}
+\| u\|_{H^{1}}\left(\int_x^{+\infty}|y^{\frac{1}{2}}u_{yy}|^2\mathrm{d}y\right)^{1/2},
\end{aligned}
\end{equation*}
where the first inequality comes from the definition of Fourier transform, the second comes from H\"older equality and Plancherel's identity.
Therefore,
\begin{equation*}
\| (\textbf{n}_0)_k\|_{C^0(\mathbb{R}^+,L^2(\mathbb{R}))}=\sup_{x\geqslant0}\| (\textbf{n}_0)_k\|_{L^2(\mathbb{R})}\lesssim\| u_{xx}\|_{L^{2,1}}+\| u\|_{H^{1}}\| u_{xx}\|_{L^{2,\frac{1}{2}}},
\end{equation*}
and
\begin{equation*}
\begin{aligned}
\| (\textbf{n}_0)_k\|_{L^2(\mathbb{R}^+\times\mathbb{R})}&\lesssim\left(\int_0^{+\infty}\int_x^{+\infty}|yu_{yy}|^2\mathrm{d}y\mathrm{d}x
+\| u\|^2_{H^{1}}\int_0^{+\infty}\int_x^{+\infty}|y^{\frac{1}{2}}u_{yy}|^2\mathrm{d}y\mathrm{d}x\right)^{1/2}\\
&\lesssim\left(\int_0^{+\infty}\int_0^y|yu_{yy}|^2\mathrm{d}x\mathrm{d}y\right)^{1/2}+\| u\|_{H^{1}}\left(\int_0^{+\infty}\int_0^y|y^{\frac{1}{2}}u_{yy}|^2\mathrm{d}x\mathrm{d}y\right)^{1/2}\\
&\lesssim\| u_{xx}\|_{L^{2,\frac{3}{2}}}+\| u\|_{H^{1}}\| u_{xx}\|_{L^{2,1}}.
\end{aligned}
\end{equation*}

\end{proof}

Next, we deal with the operators $(T)_k,\ (T)_{kk}$ and $(T)_{kkk}$, which have the integral kernel $(K)_k,\ (K)_{kk}$ and $(K)_{kkk}$ respectively, where
\begin{equation}
(K)_{k}(x,y,k)=\left(\begin{array}{cc}0&0\\2i\left[h(x)-h(y)\right]\frac{u_{yy}}{2m}e^{2ik\left[h(x)-h(y)\right]}&0\end{array}\right).
\end{equation}
$(K)_{kk}$ and $(K)_{kkk}$  have the same form with $2i\left[h(x)-h(y)\right]$ replaced by $\left[2i(h(x)-h(y))\right]^2$,\ and $\left[2i(h(x)-h(y))\right]^3$. These operators admit following estimates:
\begin{lemma}\label{tk}
For $u_0(x) \in H^{2,3}(\mathbb{R})$, the following operator bounds hold uniformly, and the operators are Lipschitz continuous of $u_0(x)$.
\begin{equation*}
\begin{aligned} &\| (T)_k\|_{L^2(\mathbb{R}^+\times\mathbb{R})\to C^0(\mathbb{R}^+,L^2(\mathbb{R}))}\lesssim\| u_{xx}\|_{L^{2,1}}+\| u\|_{H^{1}}\| u_{xx}\|_{L^{2,\frac{1}{2}}},\\
&\| (T)_k\|_{L^2(\mathbb{R}^+\times\mathbb{R})\to L^2(\mathbb{R}^+\times\mathbb{R})}\lesssim\| u_{xx}\|_{L^{2,\frac{3}{2}}}+\| u\|_{H^{1}}\| u_{xx}\|_{L^{2,1}};
\end{aligned}
\end{equation*}
\begin{equation*}
\begin{aligned} &\| (T)_{kk}\|_{L^2(\mathbb{R}^+\times\mathbb{R})\to C^0(\mathbb{R}^+,L^2(\mathbb{R}))}\lesssim\| u_{xx}\|_{L^{2,2}}+\| u\|_{H^{1}}\| u_{xx}\|_{L^{2,\frac{3}{2}}}+\| u\| ^2_{H^{1}}\| u_{xx}\|_{L^{2,1}},\\
&\| (T)_{kk}\|_{L^2(\mathbb{R}^+\times\mathbb{R})\to L^2(\mathbb{R}^+\times\mathbb{R})}\lesssim\| u_{xx}\|_{L^{2,\frac{5}{2}}}+\| u\|_{H^{1}}\| u_{xx}\|_{L^{2,2}}+\| u\| ^2_{H^{1}}\| u_{xx}\|_{L^{2,\frac{3}{2}}};
\end{aligned}
\end{equation*}
\begin{align*}
\| (T)_{kkk}\|_{L^2(\mathbb{R}^+\times\mathbb{R})\to C^0(\mathbb{R}^+,L^2(\mathbb{R}))}\lesssim\| u_{xx}\|_{L^{2,3}}+\| u\|_{H^{1}}\| u_{xx}\|_{L^{2,\frac{5}{2}}}+\| u\|^2_{H^{1}}\|u_{xx}\|_{L^{2,2}}+\| u\| ^3_{H^{1}}\| u_{xx}\|_{L^{2,\frac{3}{2}}}.
\end{align*}

\end{lemma}

To solve the equations \eqref{n}, \eqref{nk}, \eqref{nkk} and \eqref{nkkk}, we finally discuss the existence of the operator $(I-T)^{-1}$. Denote $f^*(x)=\displaystyle\sup_{y\geqslant x}\| f(y,\cdot)\|_{L^2(\mathbb{R})}$, then by \eqref{K}, we find $K(x,y,k)\leqslant g(y)$ and
\begin{equation}\label{Tf*}
(Tf)^*(x)\leqslant\int_x^{\infty}g(y)f^*(y)\mathrm{d}y,
\end{equation}
where
\begin{equation*}
g(y)=\frac{u_{yy}}{m}.
\end{equation*}

Therefore, the resolvent $(I-T)^{-1}$ exists with following lemma:
\begin{lemma}\label{revol}
For each $k\in\mathbb{R}$ and $u_0(x) \in H^{2,3}(\mathbb{R})$, $(I-T)^{-1}$ exists as a bounded operator from $C^0(\mathbb{R}^+)$ to itself. What's more, $\hat{L}:=(I-T)^{-1}-I$ is an integral operator with continuous integral kernel $L(x,y,k)$ satisfying
\begin{equation}\label{gn}
|L(x,y,k)|\leqslant \exp(\| g\|_{L^1})g(y).
\end{equation}
\end{lemma}
\begin{proof}
By \eqref{defT}, it's obvious that $T$ is a Volterra operator, and together with \eqref{Tf*}, we can deduce that $(I-T)^{-1}$ exists unique as a bounded operator on $C^0(\mathbb{R}^+)$. For the operator $\hat{L}$, the integral kernel $L(x,y,k)$ is given by
\begin{equation*}
L(x,y,k)=\left\{\begin{array}{ll}
\sum_{n=1}^{\infty}K_n(x,y,k),&x\leqslant y,\\
0,&x>y,
\end{array}\right.
\end{equation*}
where

\begin{equation*}
K_n(x,y,k)=\int_{x\leqslant y_1\leqslant\cdots\leqslant y_{n-1}}K(x,y_1,k)K(y_1,y_2,k)\cdots K(y_{n-1},y,k)\mathrm{d}y_{n-1}\cdots \mathrm{d}y_1.
\end{equation*}
By the estimate $|K(x,y,k)|\leqslant g(y)$, we get
\begin{equation*}
|K_n(x,y,k)|\leqslant\frac{1}{(n-1)!}\left(\int_x^{\infty}g(s)\mathrm{d}s\right)^{n-1}g(y),
\end{equation*}
and then \eqref{gn} follows.
\end{proof}

By \eqref{Tf*}, we find that $T$ is a bounded operator as $T:L^2\to C^0,\ T:C^0\to L^2,$ and $T:L^2\to T^2$. Therefore, by the formula
\begin{equation*}
\hat{L}=(I-T)^{-1}-I=T+T(I-T)^{-1}T,
\end{equation*}
we deduce that $\hat{L}$ is a bounded operator as
$\hat{L}:C^0(\mathbb{R}^+,L^2(\mathbb{R}))\to C^0(\mathbb{R}^+,L^2(\mathbb{R}))$ and $\hat{L}:L^2(\mathbb{R}^+\times\mathbb{R})\to L^2(\mathbb{R}^+\times\mathbb{R})$.

 Based on above results, we now prove the following two  propositions.

\begin{Proposition}\label{p5}
The maps
\begin{equation*}
u_0(x)\to n_{11}^{\pm}(0,k),\quad u_0(x)\to n_{21}^{\pm}(0,k)
\end{equation*}
are Lipschitz continuous from $H^{2,3}(\mathbb{R})$ to $H^{3}(\mathbb{R})$.
\end{Proposition}

\begin{proof} By \eqref{n}, we find
\begin{equation}\label{n2}
  \textbf{n}(x,k)=((I-T)^{-1}-I)\textbf{n}_0(x,k)+\textbf{n}_0(x,k).
\end{equation}
By \eqref{n0est} in Lemma \ref{lem1}, $\textbf{n}_0(x,k)\in C^0(\mathbb{R}^+,L^2(\mathbb{R}))\cap L^2(\mathbb{R}^+\times\mathbb{R})$,
and then Lemma \ref{revol} guarantees that there exists unique solution $\textbf{n}(x,k)$ of \eqref{n2} with $\textbf{n}(x,k)\in C^0(\mathbb{R}^+,L^2(\mathbb{R}))\cap L^2(\mathbb{R}^+\times\mathbb{R})$.
Similarly, together with Lemma \ref{tk}, we have
\begin{equation*}
\begin{aligned}
&\textbf{n}_k(x,k)\in C^0(\mathbb{R}^+,L^2(\mathbb{R}))\cap L^2(\mathbb{R}^+\times\mathbb{R}),\\
&\textbf{n}_{kk}(x,k)\in C^0(\mathbb{R}^+,L^2(\mathbb{R}))\cap L^2(\mathbb{R}^+\times\mathbb{R}),\\
&\textbf{n}_{kkk}(x,k)\in C^0(\mathbb{R}^+,L^2(\mathbb{R})).
\end{aligned}
\end{equation*}
Taking $x=0$ in all above, we get $\textbf{n}(0,k)\in H^3(\mathbb{R})$.
\end{proof}

  As $a(k),\ b(k)$ are independent with $x$ and $t$, combined with the symmetry of $\mu_\pm$ in Proposition \ref{p2}, taking $x=t=0$, we have
\begin{align}
&a(k)=\mu_{11}^+(0,k)\overline{\mu_{11}^-(0,k)}+\mu_{21}^+(0,k)\overline{\mu_{21}^-(0,k)},\nonumber\\
&e^{-2ikc_0}{b(k)}=-\overline{\mu_{11}^+(0,k)}\overline{\mu_{21}^-(0,k)}+\overline{\mu_{21}^+(0,k)}\overline{\mu_{11}^-(0,k)},\nonumber
\end{align}
where $c_0=\int_0^{\infty}(\sqrt{m(s,0)}-1)\mathrm{d}s$ is real. This implies
\begin{equation}
    \|b(k)\|_{L^2(\mathbb{R})}=\| e^{-2ikc_0}b(k)\|_{L^2(\mathbb{R})}.
\end{equation}

From \eqref{defn},   $a(k)$ and $ b(k)$ can be  represented by
\begin{align}
&a(k)-1=n_{11}^+(0,k)\overline{n_{11}^-(0,k)}+n_{21}^+(0,k)\overline{n_{21}^-(0,k)}+n_{11}^+(0,k)+\overline{n_{11}^-(0,k)},\label{ak1}\qquad\quad\\
&e^{-2ikc_0}b(k)=\overline{n_{11}^-(0,k)n^+_{21}(0,k)}-\overline{n_{21}^-(0,k)n^+_{11}(0,k)}+\overline{n^+_{21}(0,k)}-\overline{n_{21}^-(0,k)}.\label{bk1}
\end{align}

 Based on the results in Proposition \ref{p5}, we can prove the scattering map from $u_0(x)$ to $r(k)$ as follows.

\begin{Proposition}\label{p4}
Suppose the initial data $u_0(x)\in H^{2,3}(\mathbb{R})$ , then reflection coefficient   $ r(k) \in H^3(\mathbb{R})\cap H^{1,1}(\mathbb{R})$, moreover the map $u_0(x) \to r(k)$ is Lipschitz continuous.
\end{Proposition}

\begin{proof}  As $\textbf{n}^{\pm}(0,k)\in H^3(\mathbb{R})$, by \eqref{ak1} and \eqref{bk1}, it's obvious that $a(k)$ is bounded and $a'(k),\ a''(k),\ a'''(k)\in L^2(\mathbb{R}),\ b(k)\in H^3(\mathbb{R})$. Thus $r(k)\in H^3(\mathbb{R})$.

Moreover, we need to prove $r(k)\in H^{1,1}(\mathbb{R})$, which equals to prove that $k{b(k)},\ k{b'(k)}\in L^2(\mathbb{R})$. Based on \eqref{mupm}, we find
\begin{equation*}
\begin{aligned}
k\overline{n^{\pm}_{21}(0,k)}&=-k\int_{\pm\infty}^0e^{2ik\int_y^0\sqrt{m(s,0)}\mathrm{d}s}\frac{u_{yy}}{2m}\mathrm{d}y-
k\int_{\pm\infty}^0e^{2ik\int_y^0\sqrt{m(s,0)}\mathrm{d}s}\frac{u_{yy}}{2m}\overline{n_{11}^{\pm}(y,k)}\mathrm{d}y\\
&=\int_{\pm\infty}^0\frac{1}{4i}\frac{u_{yy}}{m^{3/2}}\mathrm{d}e^{2ik\int_y^0\sqrt{m(s,0)}\mathrm{d}s}+\int_{\pm\infty}^0\frac{1}{4i}\frac{u_{yy}}{m^{3/2}}\overline{n_{11}^{\pm}(y,k)}\mathrm{d}e^{2ik\int_y^0\sqrt{m(s,0)}\mathrm{d}s}\\
&=\frac{1}{4i}\frac{u_{xx}(0)}{m^{3/2}(0)}+I_1^{\pm}+I_2^{\pm},
\end{aligned}
\end{equation*}
where
\begin{align*}
I_1^{\pm}&=\frac{1}{4i}\frac{u_{xx}(0)}{m^{3/2}(0)}\overline{n_{11}^{\pm}(0,k)},\\
I_2^{\pm}&=-\int_{\pm\infty}^0\frac{1}{4i}\left[\frac{u_{yy}}{m^{3/2}}\left(1+\overline{n_{11}^{\pm}(y,k)}\right)\right]_ye^{2ik\int_y^0\sqrt{m(s,0)}\mathrm{d}s}\mathrm{d}y,
\end{align*}
belong to $L^2(\mathbb{R})$.
Therefore, by \eqref{bk1} , we have
\begin{equation*}
\begin{aligned}
e^{-2ikc_0}kb(k)=&\frac{1}{4i}\frac{u_{xx}(0)}{m^{3/2}(0)}\left[\overline{n^-_{11}(0,k)}-\overline{n^+_{11}(0,k)}\right]+(I_1^{+}+I_2^{+})\overline{n^-_{11}(0,k)}\\
&-(I_1^{-}+I_2^{-})\overline{n^+_{11}(0,k)}+(I_1^{+}+I_2^{+})-(I_1^{-}+I_2^{-}).
\end{aligned}
\end{equation*}
Thus we conclude that $kb(k)\in L^2(\mathbb{R})$, and the proof of $kb'(k) \in L^2(\mathbb{R})$ is similar.
\end{proof}

What's more, we give a remark as a supplement of Proposition \ref{p4}. It plays an important role in solving the singularity at $k=0$ in following sections.
\begin{remark}\label{r1}
If  $r(k)\in H^3(\mathbb{R})$, then  $r(k)\in C^{2}(\mathbb{R})$  by the Sobolev  embedding theorem.
\end{remark}

It is known that $a(k)$ may have zeros on $\mathbb{R}$, which is excluded from our analysis. To clarify the aim of our paper, we give the following assumption.

 \begin{Assumption}\label{a1}
     The initial data $u_0(x)\in H^{2,3}(\mathbb{R})$, and we suppose the scattering data satisfy the following assumptions:
     \begin{itemize}
         \item $a(k)$ has no zero point on $\mathbb{R}$,
         \item $a(k)$ has finite number of simple points.
     \end{itemize}
 \end{Assumption}
 We assume that $a(k)$ has $N$ simple zeros $z_n\in\mathbb{C}^+,n=1,2,\dots,N$, then by symmetry, $a(\bar{k})$ has $N$ simple zeros $\bar{z}_n\in\mathbb{C}^-,n=1,2,\dots,N$. Define $\mathcal{Z}:=\{z_n\}_{n=1}^N,\overline{\mathcal{Z}}:=\{\bar{z}_n\}_{n=1}^N$ then the discrete spectrum is $\mathcal{Z}\cup\overline{\mathcal{Z}}$.  Denote  $\mathcal{N}= \{1,2,\cdots,N\}$.

\subsection{Set-up of a basic RH problem}\label{subs2.3}
\quad  We introduce a new  scale
\begin{equation}
y: =x-\int_x^{+\infty}(\sqrt{m(s,t)}-1)\mathrm{d}s,\label{xtoy}
\end{equation}
and write  $p(x,t;k)$  in the form
\begin{equation}
 p(x,t;k)=t\theta(k,\xi),
\end{equation}
where
\begin{equation}
\theta(k,\xi)=k\xi+4\alpha k^3-\frac{\beta}{4k},\quad \xi=\frac{y}{t}. \label{thetak}
\end{equation}

Define a matrix function
\begin{equation}
M(k):=M(y,t,k)=\left\{\begin{array}{lr}
\left(\left[\mu_+\right]_1\quad \frac{\left[\mu_-\right]_2}{a\left(k\right)}\right),\quad \mathrm{Im}k>0,\\
\left(\frac{\left[\mu_-\right]_1}{\overline{a(\bar{k})}} \quad \left[\mu_+\right]_2 \right),\quad \mathrm{Im}k<0,\end{array}
\right.
\end{equation}
which solves the following RH problem
\begin{RHP}\label{rhp1}
Find a $2\times 2$ matrix-valued function $M(k)$ satisfying \begin{itemize}
             \item Analyticity: $M(k)$ is meromorphic in $\mathbb{C}\setminus\mathbb{R};$
             \item Symmetry: $M(k)=\sigma_2 \overline{ M(\bar{k})}\sigma_2=\sigma_2 M(-k)\sigma_2;$
             \item Jump condition: $M(k)$ has continuous boundary values $M_{\pm}(k)$ on $\mathbb{R}$ and
\begin{equation}
M_+(k)=M_-(k)V(k),\quad k\in \mathbb{R},
\end{equation}
where
\begin{equation}
V(k)=e^{it\theta(k)\hat{\sigma}_3}\left(\begin{array}{cc}
1&r(k)\\\bar{r}(k)&1+|r(k)|^2\end{array}\right);
\end{equation}
\item Asymptotic behaviors:
\begin{eqnarray}
&&M(k)=I+\mathcal{O}(k^{-1}),\quad k\to \infty;\nonumber\\
&&M(k)=Q\left[I+(ic_+\sigma_3+iu\sigma_1)k+\mathcal{O}(k^2)\right],\quad k\to 0,
\end{eqnarray}
where
{\begin{equation}
c_+=\int_x^{+\infty}\left(\sqrt{m(s,t)}-1\right)\mathrm{d}s;
\end{equation}}
\item
Residue condition: $M(k)$ has simple poles at each $z_n\in\mathcal{N}$ with
\begin{eqnarray}
&&\res_{k=z_n}M(k)=\lim_{k\to z_n}M\left(\begin{array}{cc}
0&c_ne^{2it\theta(z_n)}\\0&0\end{array}\right),\\
&&\res_{k=\bar{z}_n}M(k)=\lim_{k\to\bar{z}_n}M\left(\begin{array}{cc}
0&0\\-\bar{c}_ne^{-2it\theta(\bar{z}_n)}&0\end{array}\right),
\end{eqnarray}
where $c_n=\frac{b(z_n)}{a'(z_n)},\quad n=1,2,\cdots,N.$
           \end{itemize}
\end{RHP}
The reconstruction formula of  $u(x,t)=u(y(x,t),t)$ is given by
\begin{equation}
u(x,t)=u(y(x,t),t)=\lim_{k\to0}\frac{\left[M(y,t;0)^{-1}M(y,t;k)\right]_{12}}{ik},
\end{equation}
where
\begin{equation}
    y(x,t)=x-c_+(x,t)=x-\lim_{k\to0}\frac{\left[M(y,t;0)^{-1}M(y,t;k)\right]_{11}-1}{ik}.
\end{equation}
\subsection{\texorpdfstring{Classification of asymptotic regions by parameters $\alpha,\beta,\xi$}{Classification of asymptotic regions by parameters alpha, beta, xi}}\label{subs2.4}
In this section, we present the signature tables for $e^{2it\theta(k)}$ and the distribution of saddle points for $\theta(k) $ on $\mathbb{R}$. By calculation,
\begin{align}
    \theta^{'}(k)&=\xi+12\alpha k^2+\frac{\beta}{4k^2},\nonumber\\ \mathrm{Im}\theta(k)&=\mathrm{Im}k\left[\xi+12\alpha|k|^2-16\alpha (\mathrm{Im}k)^2+\frac{\beta}{4|k|^2} \right].\label{Imthata}
\end{align}

We can divide the problem into four cases by values of the parameter $\alpha,\beta,\xi$. From $\theta'(k)=0$, let $w=k^2$, we have
    \begin{equation}
        48\alpha w^2+4\xi w+\beta=0.\label{eq1}
    \end{equation}
    It can be calculated that the quadratic equation \eqref{eq1} has two roots
    \begin{equation*}
        w_1=\frac{-\xi+\sqrt{\xi^2-12\alpha\beta}}{24\alpha},\quad w_2=\frac{-\xi-\sqrt{\xi^2-12\alpha\beta}}{24\alpha},
    \end{equation*}
    from which we can obtain the $4$ roots for the equation $ \theta^{'}(k)=0$ on the complex plane $\mathbb{C}$
    \begin{equation}\label{eq42}
    \begin{aligned}
         k_1&=\sqrt{\frac{-\xi+\sqrt{\xi^2-12\alpha\beta}}{24\alpha}},\quad k_4=-\sqrt{\frac{-\xi+\sqrt{\xi^2-12\alpha\beta}}{24\alpha}},\\
         k_2&=\sqrt{\frac{-\xi-\sqrt{\xi^2-12\alpha\beta}}{24\alpha}},\quad k_3=-\sqrt{\frac{-\xi-\sqrt{\xi^2-12\alpha\beta}}{24\alpha}}.
    \end{aligned}
    \end{equation}

    Based on the number of roots on the real line, which is associated with the parameter $\alpha,\beta,\xi$, we can divide this problem into the following four cases.

    \begin{itemize}
        \item \textbf{\textup{Case \uppercase\expandafter{\romannumeral1}.}}  When  $\alpha, \beta >0,\; \xi< - 2\sqrt{2\alpha\beta}$,
        there are four saddle points $k_j,\; j=1,2,3,4$, located on the jump line $\mathbb{R}$ with $k_4=-k_1,\;k_3=-k_2$.
        \item \textbf{\textup{Case \uppercase\expandafter{\romannumeral2}.}} When  $\alpha, \beta >0,\; \xi= - 2\sqrt{2\alpha\beta}$,
        there are two saddle points $\pm k_0 $ located on the jump line $\mathbb{R}$ .
         \item \textbf{\textup{Case \uppercase\expandafter{\romannumeral3}.}} When  $\alpha, \beta >0,\; \xi> - 2\sqrt{2\alpha\beta}$,
        there is no saddle point located on the jump line, which means the saddle points are non-real complex numbers.

    \end{itemize}

\begin{figure}[!h]
	\centering \subfigure[Four saddle points on $\mathbb{R}$\qquad\quad]{ \includegraphics[width=0.25\linewidth]{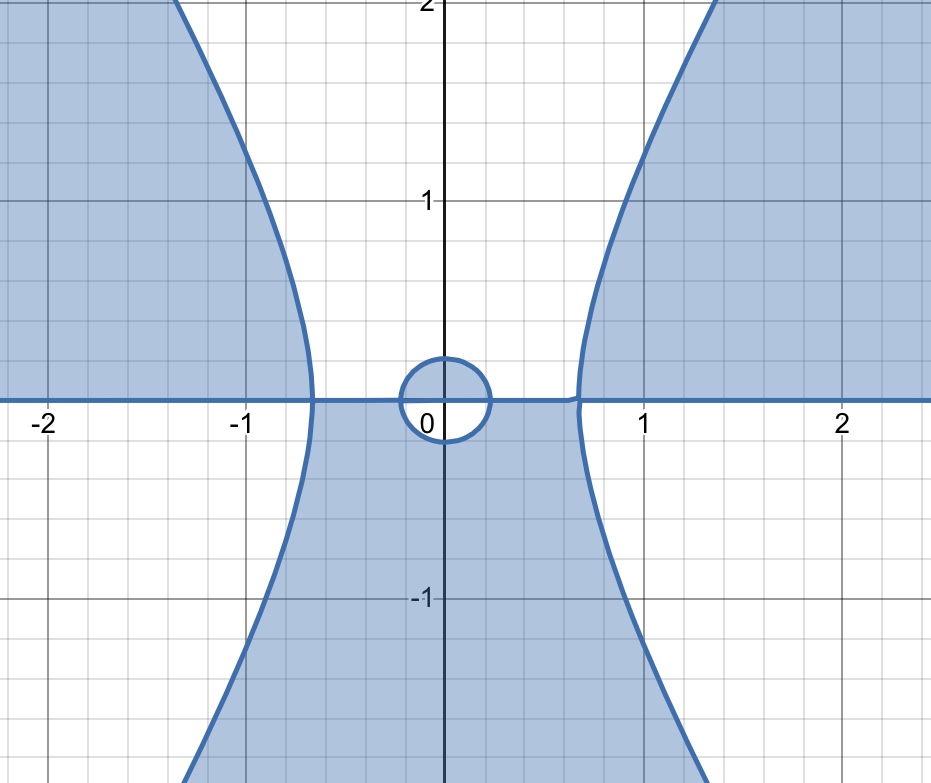}\hspace{1cm}
	\label{fig:desmos-graph-1.1}} \subfigure[Two saddle points on $\mathbb{R}$\qquad\quad]{\includegraphics[width=0.25\linewidth]{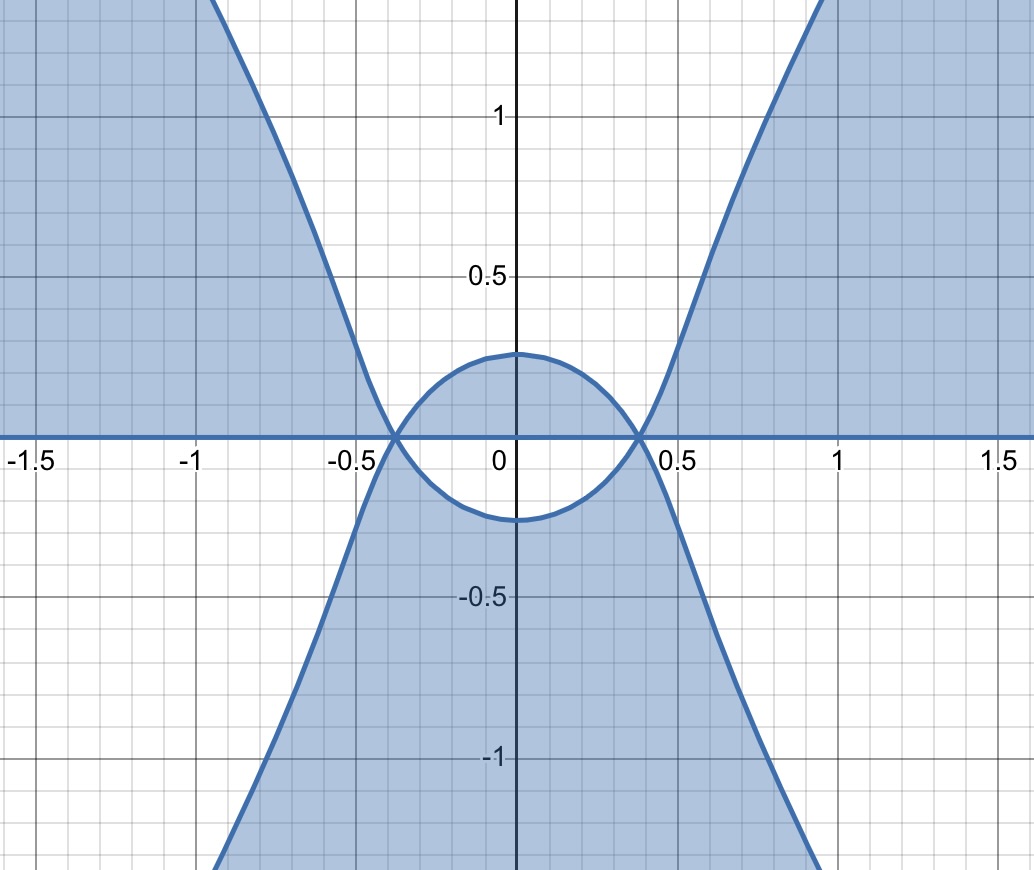}\hspace{1cm}
	\label{fig:desmos-graph-1.2}}	\subfigure[No  saddle point on $\mathbb{R}$\qquad\quad]{\includegraphics[width=0.25\linewidth]{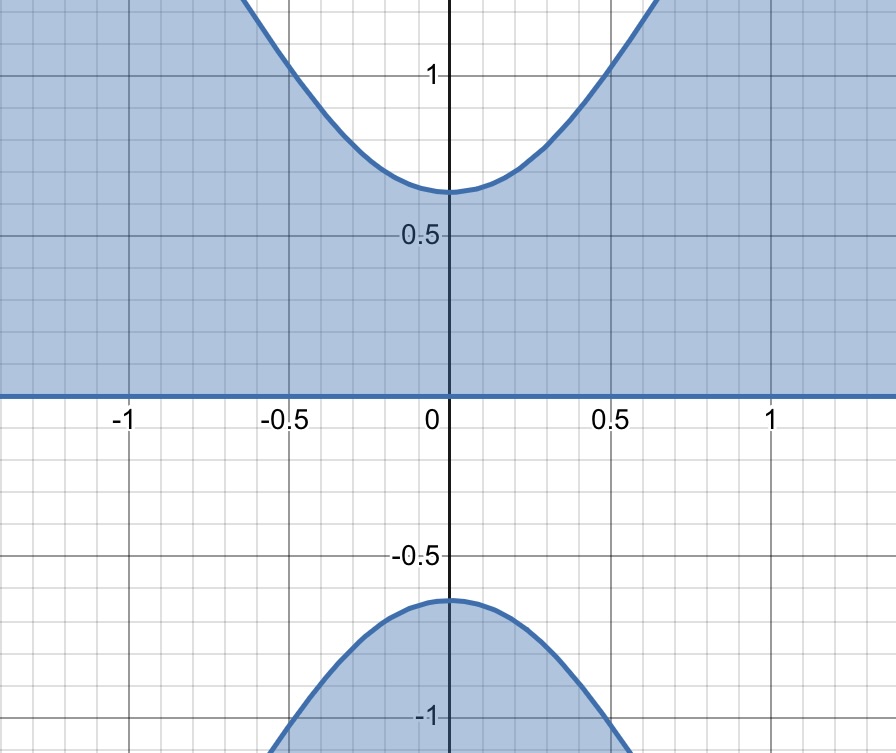}\hspace{1cm}
	\label{fig:desmos-graph-1.3}}

	\caption{\footnotesize The classification  of $\text{sign Im}\theta$ for cases I-III.  In the blue regions, $\text{Im}\theta>0$, which implies that $|e^{2it\theta}|\to 0$ as $t\to\infty$.
While  in the white regions,   $\text{Im}\theta<0$, which means  $|e^{-2it\theta}|\to 0$ as $t\to\infty$.   The blue curves  $\text{Im}\theta=0$ are the dividing lines between the decay and growth regions.  }
	\label{figtheta-1}
\end{figure}
\FloatBarrier

    \begin{itemize}
         \item \textbf{\textup{Case \uppercase\expandafter{\romannumeral4}.}}  When $\alpha<0, \beta>0$,
        there are two saddle points $k_j,\; j=1,2$ located on the jump line $\mathbb{R}$ with $k_2=-k_1$.
    \end{itemize}

\begin{figure}[!h]
	\centering \includegraphics[width=0.25\linewidth]{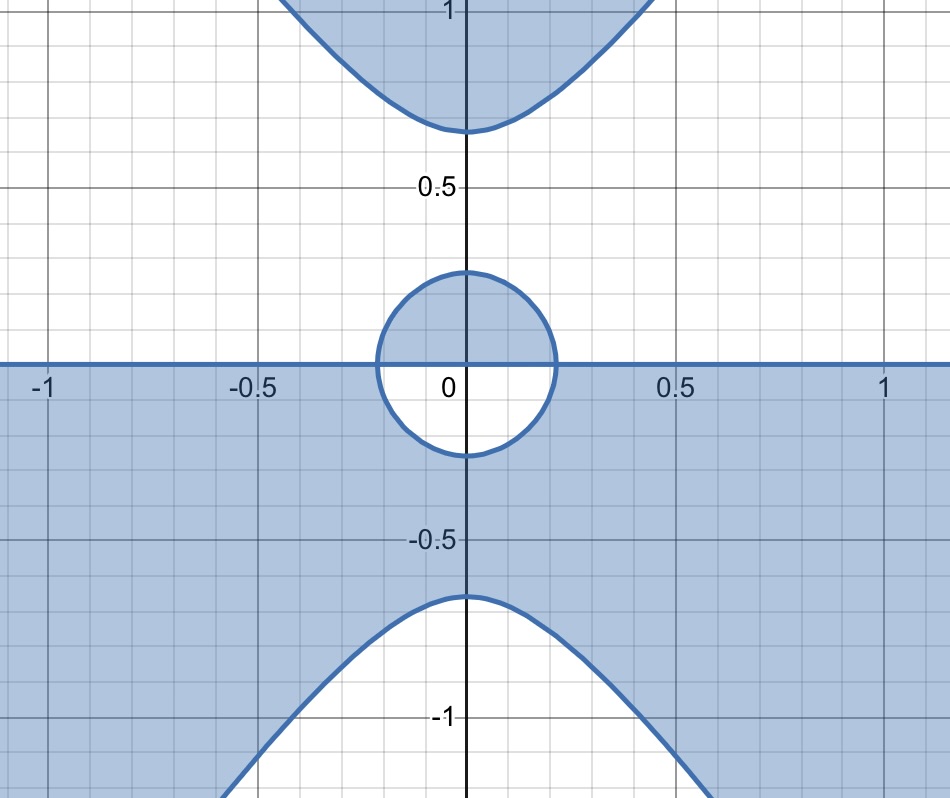}\hspace{1cm}

	\caption{\footnotesize The classification  of $\text{sign Im}\theta$ for Case  IV.  }
	\label{figtheta-2}
\end{figure}
\FloatBarrier


\section{Long-time asymptotics in regions with saddle points}\label{sec3}

As we shown in Subsection \ref{subs2.4}, for Case \uppercase\expandafter{\romannumeral1} ($\alpha>0,\beta>0,\xi<-2\sqrt{3\alpha\beta}$) and Case \uppercase\expandafter{\romannumeral4} ($\alpha<0,\beta>0$), there exist four and two saddle points on the real axis respectively, which is denoted as $k_1>k_2>k_3>k_4$ and $k_1>k_2$.

\subsection{\texorpdfstring{Jump matrix factorizations and hybrid $\bar{\partial}$-RH problem}{Jump matrix factorizations and hybrid dbar-RH problem}}\label{subs3.1}

We denote
\begin{equation}
    I:=I(\alpha,\beta,\xi)=\begin{cases}
        (k_4,k_3)\cup(k_2,k_1),\quad &\alpha>0,\beta>0,\xi<-2\sqrt{3\alpha\beta},\\
        (-\infty,k_2)\cup(k_1,+\infty),\quad &\alpha<0,\beta>0.
    \end{cases}
\end{equation}
For brevity, we denote
\begin{align}
   &\Lambda:=\Lambda(\alpha,\beta,\xi)= \begin{cases}
    4,\quad &\alpha>0,\beta>0,\xi<-2\sqrt{3\alpha\beta},\\2,\quad &\alpha<0,\beta>0.
\end{cases}\\
&\eta:=\eta(\alpha,\beta,\xi,k_j)=\begin{cases}
    (-1)^{j+1},\quad &\alpha>0,\beta>0,\xi<-2\sqrt{3\alpha\beta},\\(-1)^{j},\quad &\alpha<0,\beta>0.
\end{cases}
\end{align}

We can decompose jump matrix $V(k)$ into the upper and lower triangular matrices
\begin{equation} \label{234}
V(k)=\left\{\begin{array}{ll}
\left(\begin{array}{cc}1&\frac{r}{1+|r|^2}e^{2it\theta}\\ 0&1\end{array}\right)\left(1+|r|^2\right)^{-\sigma_3}\left(\begin{array}{cc}1&0\\ \frac{\bar{r}}{1+|r|^2}e^{-2it\theta}&1\end{array}\right)
 &k\in I,\\
\left(\begin{array}{cc}1&0\\
\bar{r}e^{-2it\theta}&1\end{array}\right) \left(\begin{array}{cc}1&re^{2it\theta}\\ 0&1\end{array}\right)&k \in \mathbb{R}\setminus I.
\end{array}\right.
\end{equation}

In order to eliminate the diagonal matrix in \eqref{234}, we introduce the following scalar RH problem:
\begin{RHP}\label{rhp2}
    Find a scalar function $\delta(k)$ satisfying the following properties:
\begin{itemize}
  \item Analyticity: $\delta(k)$ is analytical in $\mathbb{C}\setminus\mathbb{R};$
  \item Jump condition: $\delta(k)$ has continuous boundary values $\delta_{\pm}$ and
      \begin{equation*}
\left\{\begin{array}{ll}
\delta_+(k)=\delta_-(k)(1+|r|^2), &k\in I;\\
\delta_+(k)=\delta_-(k), &k\in\mathbb{R}\setminus I;\end{array}\right.
\end{equation*}
  \item Asymptotic behavior:
  \begin{equation*}
  \delta(k)\to 1,\quad k\to \infty.
  \end{equation*}
\end{itemize}
\end{RHP}
By the Plemelj formula, the unique solution for RH problem can be calculated as
\begin{equation*}
\delta(k)=\mathrm{exp}\left[i\int_I\frac{\nu(s)}{s-k}\mathrm{d}s\right],
\end{equation*}
where
\begin{equation*} \nu(s)=-\frac{1}{2\pi}\log(1+|r(s)|^2).
\end{equation*}
Further, we classify $\mathcal{Z}$ with the sign of $\theta(k)$,
\begin{equation}\label{eq43}
    \Delta^-=\{n\in\mathcal{N}:\mathrm{Im}\theta(z_n)<0\},\;\Delta^+=\{n\in\mathcal{N}:\mathrm{Im}\theta(z_n)>0\}.
\end{equation}
Define function
\begin{equation}
    T(k)=\prod\limits_{n\in\Delta^-}\frac{k-\bar{z}_n}{k-z_n}\delta(k).\label{T1}
\end{equation}
In the above formulas, we choose the principle branch of power and logarithm functions.
\begin{Proposition}\label{p8}
    The function we defined above has the following properties:
    \begin{enumerate}[label=\textnormal{(\arabic*)}]
      \item $T(k)$ is meromorphic in $\mathbb{C} \setminus I$. And for each $n \in \Delta^-, T(k)$ has a simple pole at $z_n$ and a simple zero at $\bar{z}_n$;
       \item For $k\in \mathbb{C}\backslash I,\;T(k)\overline{T(\bar{k})}=1$;
       \item For $k \in I$, denote the boundary values of $T(k)$ as $T_{ \pm}(k)$ with $k$ approaching the real axis from above and below respectively, which satisfy:
$$
T_{+}(k)=T_{-}(k)\left(1+|r(k)|^2\right),  \quad k \in I;
$$

\item As $|k|\to+\infty,\;|\mathrm{arg}k|\leqslant c<\pi$,
$$T(k)=1+\frac{i}{k}\left[2\sum_{n\in\Delta^-}\mathrm{Im}z_n-\int_{I}\nu(s)\mathrm{d}s\right]+\mathcal{O}(k^{-2});$$

\item $T(k)$ is continuous at $k=0$, and
\begin{equation}
    T(k)=T_0(1+iT_1k)+\mathcal{O}(k^2),
\end{equation}
where $$T_0=\prod\limits_{n\in\Delta^-}\frac{\bar{z}_n}{z_n}=\mathrm{exp}\left[-2i\sum_{n\in\Delta^-}\mathrm{arg}(z_n)\right],\quad T_1=\int_I\frac{\nu(s)}{s^2}\mathrm{d}s-\sum_{n\in\Delta^-}\frac{2\mathrm{Im}(z_n)}{|z_n|^2};$$

\item As $k \to k_j$ along any ray $k_j+e^{i \phi} \mathbb{R}^{+}$with $|\phi|<\pi$,

\begin{equation}
    \left|T(k, k_j)-T_0(k_j,k_j)\left(k-k_j\right)^{\eta(k_j) i \nu\left(k_j\right)}\right| \lesssim\|r\|_{H^1(\mathbb{R})}\left|k-k_j\right|^{\frac{1}{2}},\label{eq12}
\end{equation}
where $T_0(k,k_j)$ is the complex function
\begin{equation}
  T_0(k,k_j)=\prod_{n \in \Delta^-} \frac{k-\bar{z}_n}{k-z_n} e^{i \beta\left(k,k_j\right)}
\end{equation}
for $j=1, \cdots, \Lambda$. In the above formula,
\begin{equation}
    \beta(k, k_j)=-\eta(k_j)\nu(k_j)\mathrm{log}\left(k-k_j+\eta(k_j)\right)+\int_I\frac{\nu(s)-\chi_j(s)\nu(k_j)}{s-k}\mathrm{d}s,
\end{equation}
where $\chi_j(s)$ are the characteristic functions of the interval $I\cap(k_j-1,k_j+1)$.
    \end{enumerate}
  \end{Proposition}
 \begin{proof}
     ($1$)-($3$) can be proved by the definition of $T(k)$.We only proof ($4$),($5$) and($6$).For ($4$), we make the asymptotic expansion as $|k|\to+\infty$,
     \begin{equation*}
       \prod_{n \in \Delta^-} \frac{k-\bar{z}_n}{k-z_n}=1+\frac{2i}{k}\sum_{n\in\Delta^-}\mathrm{Im}(z_n)+\mathcal{O}(k^{-2}) ,\quad \delta(k)=1-\frac{i}{k}\int_I\nu(s)\mathrm{d}s+\mathcal{O}(k^{-2}),
     \end{equation*}
which solves ($4$).
For $k\to0$,
\begin{align*}
 \prod_{n \in \Delta^-} \frac{k-\bar{z}_n}{k-z_n}=\prod_{n \in \Delta^-}\left[\frac{\bar{z}_n}{z_n}-\frac{z_n-\bar{z}_n}{z_n^2}k+\mathcal{O}(k^2)\right],\quad
\delta(k)=1+ik\int_I\frac{\nu(s)}{s^2}\mathrm{d}s+\mathcal{O}(k^2).
\end{align*}
By simple calculation, we can obtain ($5$).
The key to proof ($6$) is the following estimation on $\beta(k,k_j)$ and $\nu(k)$:
\begin{equation}
   |\nu(k)|\lesssim|r(k)|,\quad|\beta(k,k_j)-\beta(k_j,k_j)|\lesssim\|r\|_{H^1(\mathbb{R})}\left|k-k_j\right|^{\frac{1}{2}}.
\end{equation}
Detailed proof can be found in \cite{borghese-2018-2}.

 \end{proof}

Next we use function $T(k)$ to define a new  transformation.
\begin{equation}
    M^{(1)}(y,t;k)=M(y,t;k)T(k)^{\sigma_3},\label{eq2}
\end{equation}
$M^{(1)}(y,t;k)$ is the solution to the following RH problem.

\begin{RHP}
    Find a $2\times 2$ matrix-valued function $M^{(1)}(k)$ with the following properties:
\begin{itemize}
  \item Analyticity: $M^{(1)}(k)$ is analytical in $\mathbb{C}\setminus\mathbb{R};$

  \item Jump condition: $M^{(1)}(k)$ has continuous boundary values $M^{(1)}_{\pm}(k)$ on $\mathbb{R}$ and
 \begin{equation*}
M^{(1)}_+(k)=M^{(1)}_-(k)V^{(1)}(k),
\end{equation*}
where
\begin{equation}
V^{(1)}(k)=\begin{array}{ll}
\left(\begin{array}{cc}1&0\\ \bar{\rho}(k)T_-^2(k)e^{-2it\theta}&1\end{array}\right)
\left(\begin{array}{cc}1&\rho(k)T_+^{-2}(k)e^{2it\theta}\\ 0&1\end{array}\right), &k\in \mathbb{R},
\end{array} \label{3.18}
\end{equation}
with the reflection coefficient is defined as
\begin{equation}
\rho(k) =\left\{\begin{array}{ll}
r(k),& k\in 	\mathbb{R}\setminus I,\\
-\dfrac{r(k) }{1+|r(k)|^2},&k\in I;
\end{array}\right.\label{3.19}
\end{equation}
The orientation of the jump line $\mathbb{R}$ is shown in the Figure \ref{fig0} below, which brings convenience to the unification of jump matrix.

  \item Asymptotic behavior:
  $\ M^{(1)}(k)=I+\mathcal{O}(k^{-1}),\ as \ k\to \infty;$
\item Residue condition: $M^{(1)}(k)$ has simple poles at each $n\in\mathcal{N}$ with the following residue condition
\begin{align}
 &\res_{k=z_n} M^{(1)}(k)=\lim _{k \rightarrow z_n} M^{(1)}(k)\left(\begin{array}{cc}
0 & c_nT^{-2}\left(z_n\right) e^{2 i t \theta(z_n)} \\
0 & 0
\end{array}\right), &n \in \Delta^+; \label{res1}\\
 &\res_{k=\bar{z}_n} M^{(1)}(k)=\lim _{k \rightarrow \bar{z}_n} M^{(1)}(k)\left(\begin{array}{cc}
0 & 0 \\
-\bar{c}_n T^2\left(z_n\right) e^{-2 it\theta(\bar{z}_n)} & 0
\end{array}\right), &n \in \Delta^+; \\
&\res_{k=z_n} M^{(1)}(k)=\lim _{k \rightarrow z_n} M^{(1)}(k)\left(\begin{array}{cc}
0 & 0 \\
c_n\left[\left(\frac{1}{T}\right)^{\prime}\left(z_n\right)\right]^{-2} e^{-2 i t \theta(z_n)} & 0
\end{array}\right), &n \in \Delta^-; \\
&\res_{k=\bar{z}_n} M^{(1)}(k)=\lim_{k \rightarrow \bar{z}_n} M^{(1)}(k)\left(\begin{array}{cc}
0 & -\bar{c}_n \left[T^{\prime}\left(\bar{z}_n\right)\right]^{-2} e^{2 i t \theta(\bar{z}_n)} \\
0 & 0
\end{array}\right), &n \in\Delta^-.  \label{res4}
\end{align}
\end{itemize}
\end{RHP}

\begin{figure}[h]
    \centering
   \subfigure[Case \uppercase\expandafter{\romannumeral1} ]{\begin{tikzpicture}[scale=0.6]

    \draw[red,thick] (-4.5,0) -- (-4,0);
    \draw[red,thick,] (-4,0) -- (-3,0);
     \draw[blue,  thick,] (-3,0) -- (-2,0);
    \draw[blue,  thick] (-2,0) -- (-1,0);
    \draw[red, thick] (-1,0) -- (0,0);
    \draw[red, thick,] (0,0) -- (1,0);
    \draw[blue, thick,] (1,0) -- (2,0);
    \draw[blue, thick] (2,0) -- (3,0);
 \draw[red,thick,->] (3,0) -- (4.5,0);
    \draw[red,thick,] (4,0) -- (4.5,0);

\coordinate (A) at (1,0);
\fill (A) circle (1.7pt) node[below]at (1,0){$k_2$};
\coordinate (B)  at (3,0);
\fill (B) circle (1.7pt) node[below]at(3,0){$k_1$};
\coordinate (C)  at (-1,0);
\fill (C) circle (1.7pt) node[below]at(-1,0){$k_3$};
\coordinate (D)  at (-3,0);
\fill (D) circle (1.7pt) node[below]at(-3,0){$k_4$};

    \node[right] at (4.7,0) {$\mathbb{R}$};

\end{tikzpicture}}

\subfigure[Case \uppercase\expandafter{\romannumeral4}]{\begin{tikzpicture}[scale=0.6]

    \draw[red,thick ] (-2,0) -- (0,0);
    \draw[red,thick] (2,0) -- (0,0);
     \draw[blue, thick] (-3,0) -- (-4,0);
    \draw[blue,  thick ] (-2,0) -- (-3,0);
    \draw[blue,  thick ] (2,0) -- (3,0);
     \draw[blue,  thick,->] (3,0) -- (4,0);
\coordinate (A) at (2,0);
\fill (A) circle (1.7pt) node[below]at (2,0){$k_1$};
\coordinate (B)  at (-2,0);
\fill (B) circle (1.7pt) node[below]at(-2,0){$k_2$};
    \node[right] at (4.7,0) {$\mathbb{R}$};

\end{tikzpicture}}
    \caption{The classification of jump contour $\mathbb{R}$  for  $M^{(1)}$  with  Case I and Case IV: The red line corresponds  to the first  decomposition of \eqref{3.18}-\eqref{3.19}; The blue  line corresponds  to the second  decomposition of  \eqref{3.18}-\eqref{3.19}.  }
    \label{fig0}
\end{figure}
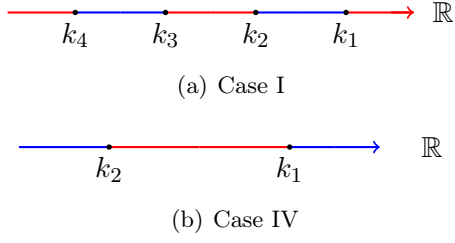

\FloatBarrier

\subsubsection{Deformation of the RH problem}

In this part, we make a continuous extension of $V^{(1)}(k)$ on $\mathbb{R}$ to open the jump line, which transforms the RH problem \ref{rhp2} into a hybrid RH problem. We opened the contour $\mathbb{R}$  in the vicinity with  deformation contours  $\Sigma_1 $  and   $\Sigma_2$ as shown in Figure \ref{Fig4}, with   $\Omega_{1,2}$ denote the regions enclosed by $\Sigma_{1,2}$ and the real line $\mathbb{R}$ respectively.
So, there is no spectrum point in the open regions $\Omega_1 $ and $\Omega_2$. Take $\phi$ as a small enough angle satisfying the following conditions:
\begin{enumerate}
    \item each $\Omega_j$ doesn't intersect with the  critical line $\{k\in\mathbb{C}:\mathrm{Im}\theta(k)=0\}$;\label{phi-1}
    \item each $\Omega_j$ is away from the $N$ solitons;
    \item $0<\sin\phi<\frac{\sqrt{3\alpha}}{2}$.\label{phi-3}
\end{enumerate}

First we give some estimates for imaginary part of the phase function $\theta(k)$ in different regions. We consider $\mathrm{Im}\theta(k)$ near $k=0$ and $k=k_j$ respectively.
Give small enough $\rho_0>0$ which satisfies $\rho_0<|k_2|$, and define
\begin{align}
B_{\rho_0}&=\{k\in\mathbb{C}:|k|<\rho_0\},\\
   \Omega &=\Omega_1\cup\Omega_2,\quad \Sigma^{(2)} =\Sigma_1\cup\Sigma_2.\label{eq13}
\end{align}
\begin{lemma}\label{p9}\textup{(}near $k=0$\textup{)}
    For a fixed small angle $\phi$ which satisfies \ref{phi-1}-\ref{phi-3}, the imaginary part of phase function $\theta(k)$ defined by \eqref{Imthata} has the following estimations for $k=le^{i\phi}$:
    \begin{align}
        &\mathrm{Im}\theta(k)\geqslant l|\sin(\phi)|\left[\xi+(12\alpha-16\alpha \sin^2\phi)\rho_0^2+\frac{\beta}{4\rho_0^2}\right],\quad k\in\Omega_{1}\cap B_{\rho_0},\label{eq3}\\
        &\mathrm{Im}\theta(k)\leqslant -l|\sin(\phi)|\left[\xi+(12\alpha-16\alpha \sin^2\phi)\rho_0^2+\frac{\beta}{4\rho_0^2}\right],\quad k\in\Omega_{2}\cap B_{\rho_0}.
    \end{align}
\end{lemma}
\begin{proof}
    For convenience, we only prove the proposition for $k\in\Omega_{1}$ of case \uppercase\expandafter{\romannumeral1}. To begin with the definition of $\theta(k)$, by $k=le^{i\phi}$, we obtain \begin{equation*}
        \mathrm{Im}\theta(k)=l\sin\phi\left[\xi+(12\alpha-16\alpha \sin^2\phi)l^2+\frac{\beta}{4l^2}\right].
    \end{equation*}
    As small enough $\phi$ satisfies \ref{phi-3}, we denote
    \begin{equation*}
        F(s)=as+\frac{b}{s}+\xi,
    \end{equation*}
    where $s=l^2,$ and
    \begin{align*}
        &a=-16\alpha\sin^2\phi+12\alpha>0,\ \
      b=\frac{\beta}{4}>0.
    \end{align*}
There are two zero points of $F(s)$ for $s>0$,
\begin{equation*}
    s_{\pm}=\frac{-\xi\pm\sqrt{\xi^2+\beta(16\alpha\sin^2\phi-12\alpha)}}{2(-16\alpha\sin^2\phi+12\alpha)},
\end{equation*}
which comes from the non-negativity of the formula inside the square roots. Obviously, $F(s)$ decreases in the interval $(0,s_-)$. As long as $\rho_0<s_-$, we can obtain \eqref{eq3}.

\end{proof}

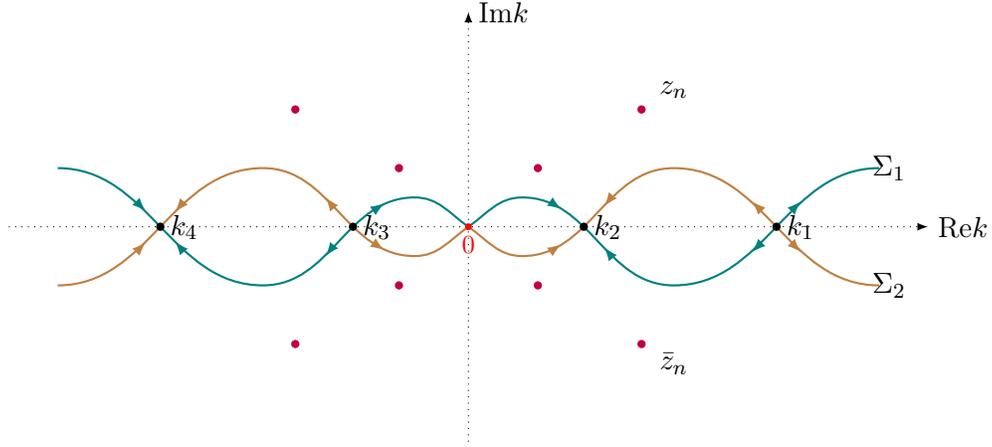
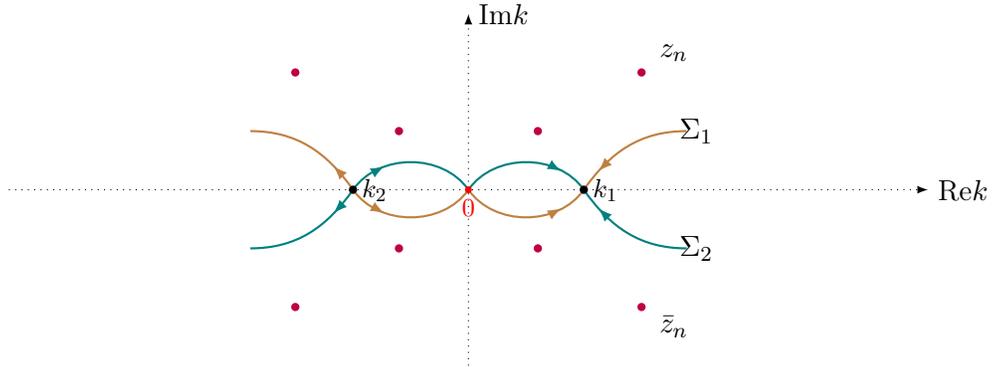
\begin{figure}[H]
\begin{center}
	\subfigure[  The opened contour $\Sigma  $ for the asymptotic region  with  Case \uppercase\expandafter{\romannumeral1}, which corresponds to   the   Figure  \ref{fig:desmos-graph-1.1}.  There
are four  saddle points on $ \mathbb{R}$.  ]{
\begin{tikzpicture}[scale=1.3]
\draw[-latex,dotted](-4.7,0)--(4.7,0)node[right]{ \textcolor{black}{Re$k$}};
\draw[-latex,dotted ](0,-2.2)--(0,2.2)node[right]{\textcolor{black}{Im$k$}};

\draw [brown,thick] (-4.2,-0.6) to [out=0,in=180] (-2.1,0.6)
to [out=0,in=180] (-0.55,-0.3) to [out=0,in=-145] (0,0)  ;

\draw [brown,thick] (0,0) to [out=-35,in=180] (0.55,-0.3) to[out=0,in=180] (2.1,0.6)  to  [out=0,in=180] (4.2,-0.6);

\draw [teal,thick](-4.2,0.6) to [out=0,in=180] (-2.1,-0.6) to [out=0,in=180] (-0.55,0.3)to  [out=0,in=145] (0,0);

\draw [teal,thick](0,0) to [out=35,in=180](0.55,0.3)to [out=0,in=180] (2.1,-0.6)  to [out=0,in=180] (4.2,0.6);

\draw[-latex,teal,thick](-3.35,0.21)--(-3.3,0.16);
\draw[-latex,brown,thick](-3.35,-0.21)--(-3.3,-0.16);
\draw[-latex,teal,thick ](-2.95,-0.21)--(-3,-0.15);
\draw[-latex,brown,thick](-2.95,0.21)--(-3,0.15);

\draw[-latex,teal,thick](3.35,0.21)--(3.4,0.26);
\draw[-latex,brown,thick](3.35,-0.21)--(3.4,-0.26);
\draw[-latex,teal,thick ](2.95,-0.21)--(2.9,-0.26);
\draw[-latex,brown,thick](2.95,0.21)--(2.9,0.26);

\draw[-latex,brown,thick](-1.39,0.21)--(-1.46,0.3);
\draw[-latex,teal,thick](-1.39,-0.21)--(-1.46,-0.3);
\draw[-latex,brown,thick](-0.95,-0.18)--(-0.86,-0.24);
\draw[-latex,teal,thick](-0.95,0.18)--(-0.86,0.24);

\draw[-latex,brown,thick](1.46,0.3)--(1.39,0.21);
\draw[-latex,teal,thick](1.46,-0.3)--(1.39,-0.21);
\draw[-latex,brown,thick](0.86,-0.24)--(0.95,-0.18);
\draw[-latex,teal,thick](0.86,0.24)--(0.95,0.18);

\coordinate (a) at (0.71,0.6);
\fill[purple] (a) circle (1.2pt);
\coordinate (a) at (-0.71,0.6);
\fill[purple] (a) circle (1.2pt);
\coordinate (a) at (0.71,-0.6);
\fill[purple] (a) circle (1.2pt);
\coordinate (a) at (-0.71,-0.6);
\fill[purple] (a) circle (1.2pt);

\coordinate (a) at (1.77,1.2);
\fill[purple] (a) circle (1.2pt);

\coordinate (a) at (1.77,-1.2);
\fill[purple] (a) circle (1.2pt);
\coordinate (a) at (-1.77,1.2);
\fill[purple] (a) circle (1.2pt);

\coordinate (a) at (-1.77,-1.2);
\fill[purple] (a) circle (1.2pt);

\node at (2.1,1.4) {$z_n$};
\node at (2.1,-1.4) {$\bar{z}_n$};
\node  at (4.3,0.6) {$\Sigma_1$};
\node  at (4.3,-0.6) {$\Sigma_2$};

\coordinate (A) at (1.18,0);
\fill (A) circle (1.2pt) node[right]{$k_2$};
\coordinate (B)  at (3.15,0);
\fill (B) circle (1.2pt) node[right]{$k_1$};
\coordinate (C)  at (-1.18,0);
\fill (C) circle (1.2pt) node[right]{$k_3$};
\coordinate (D)  at (-3.15,0);
\fill (D) circle (1.2pt) node[right]{$k_4$};
\coordinate (I) at (0,0);
		\fill[red] (I) circle (1pt) node[below,scale=0.9] {$0$};
\end{tikzpicture}
		\label{case1}}

\subfigure[ The opened contour $\Sigma  $ for the asymptotic region  with   case \uppercase\expandafter{\romannumeral4}, which corresponds to   the   Figure  \ref{figtheta-2}.  There
are two  saddle points on $ \mathbb{R}$. ]{\begin{tikzpicture}[scale=1.3]
\draw[-latex,dotted](-4.7,0)--(4.7,0)node[right]{ \textcolor{black}{Re$k$}};
\draw[-latex,dotted ](0,-1.8)--(0,1.8)node[right]{\textcolor{black}{Im$k$}};
\coordinate (a) at (1.77,1.2);
\fill[purple] (a) circle (1.2pt);

\coordinate (a) at (1.77,-1.2);
\fill[purple] (a) circle (1.2pt);
\coordinate (a) at (-1.77,1.2);
\fill[purple] (a) circle (1.2pt);

\coordinate (a) at (-1.77,-1.2);
\fill[purple] (a) circle (1.2pt);
\coordinate (a) at (0.71,0.6);
\fill[purple] (a) circle (1.2pt);
\coordinate (a) at (-0.71,0.6);
\fill[purple] (a) circle (1.2pt);
\coordinate (a) at (0.71,-0.6);
\fill[purple] (a) circle (1.2pt);
\coordinate (a) at (-0.71,-0.6);
\fill[purple] (a) circle (1.2pt);
\node at (2.1,1.4) {$z_n$};
\node at (2.1,-1.4) {$\bar{z}_n$};

\draw [teal,thick] (-2.23,-0.6) to [out=0,in=-125] (-1.18,0);
\draw [brown,thick] (-1.18,0) to [out=-55,in=-125] (0,0)  ;

\draw [brown,thick] (0,0) to [out=-55,in=-125] (1.18,0);
\draw [teal,thick] (1.18,0)  to  [out=-55,in=180] (2.23,-0.6);

\draw [brown,thick](-2.23,0.6) to [out=0,in=125] (-1.18,0);
\draw [teal,thick] (-1.18,0)to  [out=55,in=125] (0,0);

\draw [teal,thick](0,0) to [out=55,in=125](1.18,0);
\draw [brown,thick]  (1.18,0)   to [out=55,in=180] (2.23,0.6);

\draw[-latex,brown,thick](-1.33,0.19)--(-1.38,0.24);
\draw[-latex,teal,thick](-1.33,-0.19)--(-1.38,-0.24);

\draw[-latex,brown,thick](1.38,0.24)--(1.33,0.19);
\draw[-latex,teal,thick](1.38,-0.24)--(1.33,-0.19);

\draw[-latex,brown,thick](-0.95,-0.19)--(-0.86,-0.24);
\draw[-latex,teal,thick](-0.95,0.19)--(-0.86,0.24);

\draw[-latex,brown,thick](0.86,-0.24)--(0.95,-0.19);
\draw[-latex,teal,thick](0.86,0.24)--(0.95,0.19);

\node  at (2.33,0.6) {$\Sigma_2$};
\node  at (2.33,-0.6) {$\Sigma_1$};

\coordinate (A) at (1.18,0);
\fill (A) circle (1.2pt) node[right,scale=0.9]{$k_1$};

\coordinate (C)  at (-1.18,0);
\fill (C) circle (1.2pt) node[right,scale=0.9]{$k_2$};

\coordinate (I) at (0,0);
		\fill[red] (I) circle (1pt) node[below,scale=0.9] {$0$};
\end{tikzpicture}}\label{case2}

	\caption{\footnotesize
Opening the real axis $\mathbb{R}$ at saddle points  $k_j, \ j=1,\cdots, \Lambda$ with sufficient small angle $\phi$.
 The opened contours   $\Sigma_1$ and   $\Sigma_2$  decay  in   blue region  and   white  region
    in Figure  \ref{fig:desmos-graph-1.1}-Figure \ref{figtheta-2}, respectively.  The discrete spectrum
    on $\mathbb{C}$  denoted by ($\textcolor{purple}{\bullet} $).    }
	\label{Fig4}
\end{center}
\end{figure}
\FloatBarrier

\begin{corollary}\label{c1}
     $\mathrm{Im}\theta(k)$ defined by \eqref{Imthata} has the following estimates:
\begin{align*}
        &\mathrm{Im}\theta(k)\gtrsim \left|\mathrm{Im}k\right|,\quad k\in\Omega_{1}\cap B_{\rho_0},\\
        &\mathrm{Im}\theta(k)\lesssim -|\mathrm{Im}k|,\quad k\in\Omega_{2}\cap B_{\rho_0}.
    \end{align*}

\end{corollary}

\begin{lemma}\label{p10}\textup{(}near saddle points $k_j$\textup{)}
     $\mathrm{Im}\theta(k)$ defined by \eqref{Imthata} has the following estimates:
    \begin{align*}
        \mathrm{Im}\theta(k)&\gtrsim |\mathrm{Im}(k)|\left|\mathrm{Re}k-k_j\right|,\quad k\in\Omega_1,\quad j=1,\dots,\Lambda,\\
\mathrm{Im}\theta(k)&\lesssim -|\mathrm{Im}(k)|\left|\mathrm{Re}k-k_j\right|,\quad k\in\Omega_2,\quad j=1,\dots,\Lambda.
    \end{align*}
\end{lemma}
\begin{proof}
    The proof is similar with Lemma \ref{p9}.
\end{proof}

\begin{Proposition}\label{p11}
There exist  the functions $R_{\ell}(k)$: $\bar{\Omega}_{\ell}\to \mathbb{C}$, $\ell=1,2$  with the boundary values
	\begin{align} &R_{1}(k)=\left\{\begin{array}{ll}
	\rho(k)T_+(k)^{-2}, &k\in \mathbb{R},\\[4pt]
 \rho(k_j)T_0(k_j)^{-2} (k-k_j)^{-2\eta(k_j)i\nu(k_j)},  &k\in \Sigma_{1},\\
	\end{array}\right.\label{eq4} \\[5pt]
	&R_{2}(k)=\left\{\begin{array}{ll} \bar{\rho}(k)T_-(k)^{2},
 &k\in  \mathbb{R},\\[4pt] \bar{\rho}(k_j)T_0(k_j)^{2}(k-k_j)^{2\eta(k_j)i\nu(k_j)}, &k\in \Sigma_{2},
	\end{array} \right.
	\end{align}	
where $j=1,\cdots, \Lambda$.
The functions  $R_{\ell}(k), \ell=1,2$ admit the following estimates:
	\begin{align}
    &|R_{\ell}(k)|\lesssim 1+\left[1+\mathrm{Re}^2(k)\right]^{-\frac{1}{2}},\quad\;for\ k \in\Omega,\label{eq5}\\
    &|\bar{\partial}R_{\ell}(k)|\lesssim\chi(\mathrm{Re}k)+|r^\prime(\mathrm{Re}k)|+|k-k_j|^{-\frac{1}{2}}, \quad   \;for\  k\in \Omega,\;j=2,3 \;of\;case\;\textup{\uppercase\expandafter{\romannumeral1}},\label{eq6}\\
     &|\bar{\partial}R_{\ell}(k)|\lesssim \chi(\mathrm{Re}k)+|r'(\mathrm{Re}k)|+|k-k_j|^{-\frac{1}{2}}, \quad \; for\   k\in \Omega, \;j=1,2\;of\;case\;\textup{\uppercase\expandafter{\romannumeral4}},\label{eq7}\\
	&|\bar{\partial}R_{\ell}(k)|\lesssim|r'(\mathrm{Re}k)|+|k-k_j|^{-\frac{1}{2}}, \quad  for\   k\in \Omega, \;j=1,4\;of\;case\;\textup{\uppercase\expandafter{\romannumeral1}},\label{eq8}\\
&|\bar{\partial}R_{\ell}(k)|\lesssim|k| \quad as\ k\to0,\;for\ k\in \Omega,\label{eq9}\\
&\bar{\partial}R_{\ell}(k)=0,\quad \;for\ k\in\mathbb{C}\setminus\Omega,\nonumber
	\end{align}
where $\chi \in C^{\infty}_0(\mathbb{R},[0,1])$ is a fixed cut-off function with support near 0.
\end{Proposition}

\begin{proof}
To give the estimates for $|\bar{\partial}R_{\ell}(k)|$, here we consider region $\Omega_{1}$ of case \uppercase\expandafter{\romannumeral1} as an example for the situation near the origin and the saddle points respectively.

For $k\in\Omega_{1}\cap\{k\in\mathbb{C}:k_3<\mathrm{Re}k<0\}$ , we denote $k=k_3+l e^{i\varphi},\varphi \in\left[0,\phi\right],\kappa_0=\frac{\pi}{2\phi}$. Under the $(l,\phi)$ coordinate, the $\bar{\partial}$-derivative can be represented as
\begin{equation}
    \bar{\partial}=\frac{1}{2}e^{i\varphi}(\partial_l+il^{-1}\partial_\varphi).\label{eq10}
\end{equation}
There are many ways to construct $R_\ell$ for $k\in\Omega$, here we use the following method to ensure good property around $0$. First, we introduce a cut-off function  $\chi_0(x)\in\ C^{\infty}_0([0,1]),$
\begin{equation}
\chi_0(x)=\left\{\begin{array}{ll}
1,&|x|\leqslant \min\{1,|k_3|\}/8,\\
0,&|x|\geqslant \min\{1,|k_3|\}/4.
\end{array}\right.
\end{equation}

Define the function $R_{1}$ in this region as
\begin{equation*}
R_{1}=R_{1,1}+R_{1,2},
\end{equation*}
where
\begin{equation}
\begin{aligned}
R_{1,1}=&[1-\chi_0(\mathrm{Re}k)]r(\mathrm{Re}k)T_+^{-2}cos(\kappa_0\varphi)+\tilde{g}_{1}[1-cos(\kappa_0\varphi)],\\
R_{1,2}=&\tilde{f}_1(k)cos(\kappa_0\varphi)+\frac{i}{\kappa_0}l e^{-i\varphi}sin(\kappa_0\varphi)\chi_0(\varphi)\tilde{f}_1^\prime(k),
\end{aligned}
\end{equation}
and
\begin{eqnarray}
&&\tilde{g}_{1}(k)=r(k_3)T_0^{-2}(k_3)(k-k_3)^{-2i\nu(k_3)},\nonumber\\
&&\tilde{f}_1(k)=\chi_0(\mathrm{Re}k)r(\mathrm{Re}k)T_+^{-2}(k).\nonumber
\end{eqnarray}
See Figure \ref{Fig6}.
Here the function $R_{1,2}$ is used to implement the estimate near $k=0$, which can be shown in the diagram below.

\begin{figure}
    \centering
      \begin{tikzpicture}[scale=1]
\draw[-latex,thick](-6.5,0)--(1.5,0)node[right]{ \textcolor{black}{Re$k$}};


\fill[yellow!20] (-0.01,0.01).. controls(-0.49,1.07) and (-1.15,1.09)..(-1.5,1.2)--(-1.5,0.01);
\draw [teal, thick](-5,0) to [out=55,in=110](0,0);
\draw[teal, thick] (0,0) to [out=90,in=180] (1.03,0.7);
\draw [brown, thick](-6.5,0.8) to [out=0,in=90](-5,0);

   \coordinate (I) at (0,0);
		\fill[black] (I) circle (1pt) node[below,scale=0.9] {$0$};
        \coordinate (A) at (-5,0);
		\fill[black] (A) circle (1pt) node[below,scale=0.9] {$k_3$};
        \coordinate (B) at (-1.5,0);
		\fill[black] (B) circle (1pt) node[below,scale=0.9] {$\frac{|k_3|}{4}$};
        \node at (-2.8,0.35) {$R_{1,1}$};
         \node at (-0.85,0.35) {$R_{1,2}$};
          \draw[dotted,very thick](-1.5,0)--(-1.5,1.2);
\end{tikzpicture}
\caption{\footnotesize
The construction of the extension function $R_1$ in $\Omega_1$ near $k=0$.    }
	\label{Fig6}
\end{figure}
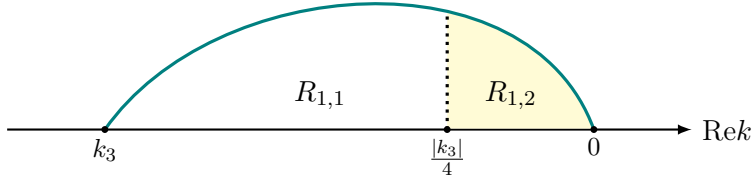

And the values of $R_1$ on $\mathbb{R}$ and $\Sigma_1$ are consistent with \eqref{eq4}. From $r(k)\in H^{1,1}(\mathbb{R})$ we can get $|r(k)|\lesssim\left[1+(\mathrm{Re}k)^2\right]^{-\frac{1}{2}}$, together with
\begin{equation*}
    |(k-k_3)^{-2i\nu(k_3)}|\lesssim e^{\pi\nu(k_3)}=\sqrt{1+|r(k_3)|^2},       \end{equation*}
we can prove \eqref{eq5}.

To prove \eqref{eq6}, We first deal with $R_{1,1}$, by \eqref{eq10}, we have
\begin{equation}
\begin{aligned}\label{eq11}
\bar{\partial}R_{1,1}=&-\frac{1}{2}\chi_0'(\mathrm{Re}k)r(\mathrm{Re}k)T_+^{-2}cos(\kappa_0\varphi)+\frac{1}{2}[1-\chi_0(\mathrm{Re}k)]r'(\mathrm{Re}k)T_+^{-2}cos(\kappa_0\varphi)\\
&-\frac{\kappa_0i}{2}l^{-1}e^{i\varphi}[1-\chi_0(\mathrm{Re}k)]r(\mathrm{Re}k)T_+^{-2}sin(\kappa_0\varphi)+\frac{\kappa_0i}{2}l^{-1}e^{i\varphi}\tilde{g}_{1}sin(\kappa_0\varphi),
\end{aligned}
\end{equation}
where $r(\mathrm{Re}k)$ is bounded on the support of $\chi_0^{\prime}(\mathrm{Re}k)$, thus \eqref{eq11} is estimated as
\begin{equation}
\left|\bar{\partial}R_{1,1}\right|\lesssim\chi(\mathrm{Re}k)+\left|r'(\mathrm{Re}k)\right|+l^{-1}\left|\tilde{g}_{1}-r(\mathrm{Re}k)T_+^{-2}\right|.
\end{equation}
The last item on the right is rewritten as
\begin{align*}
    l^{-1}\left|\tilde g_1-r(\mathrm{Re}k)T_+^{-2}\right|&=l^{-1}\left|r(k_3)T_0^{-2}(k_3)(k-k_3)^{-2i\nu(k_3)}-r(\mathrm{Re}k)T_+^{-2}\right|\\
    &\leqslant l^{-1}\left|\left[r(\mathrm{Re}k)-r(k_3)\right]T_+^{-2}+r(k_3)\left[T_+^{-2}-T_0^{-2}(k_3)(k-k_3)^{-2i\nu(k_3)}\right]\right|,
\end{align*}
from $|r(\mathrm{Re}k)-r(k_3)|\lesssim|k-k_3|^{\frac{1}{2}}$ and \eqref{eq12}, we finally come to
\begin{equation}
    l^{-1}|\tilde{g}_{1}-r(\mathrm{Re}k)T_+^{-2}|\lesssim|k-k_3|^{-\frac{1}{2}}.
\end{equation}
For $R_{1,2}$, we have
\begin{align}
    \bar{\partial }R_{1,2}&= \frac{1}{2}\tilde{f}_1^\prime(k)\cos(\kappa_0\varphi)\left[1-\chi_0(\mathrm{Re}k)\right]
    -\frac{\kappa_0i}{2}l^{-1}e^{i\varphi}\tilde{f}_1(k)\sin(\kappa_0\varphi)\\
    &+\left[\frac{i}{\kappa_0}\chi_0(\varphi)-\frac{1}{2\kappa_0}\chi_0^\prime(\varphi)\right]\tilde{f}_1^\prime(k)\sin(\kappa_0\varphi)+\frac{i}{2\kappa_0}le^{-i\varphi}\chi_0(\varphi)\tilde{f}_1^{\prime\prime}(k)\sin(\kappa_0\varphi).
\end{align}
Obviously, each item of the right is bounded in the support of $\chi_0(\mathrm{Re}k)$, so
\begin{equation}
|\bar{\partial}R_{1,2}|\lesssim\chi(\mathrm{Re}k).
\end{equation}
Summering the results we obtain for $\bar{\partial }R_{1,1}$ and  $\bar{\partial }R_{1,2}$, we can obtain \eqref{eq6}.
As $k\to 0$, we have $\mathrm{Re}k\to 0, l\to |k_3|$ and within a small neighborhood of $0$, $\chi_0(\mathrm{Re}k)\equiv1,\ \chi_0'(\mathrm{Re}k)\equiv0$,
thus
\begin{equation}
\begin{aligned}
\left|\bar{\partial}R_{1,2}\right|&\lesssim|\tilde{f}'(k)+\tilde{f}_1'(k)
+\tilde{f}''_1(k)||sin(\kappa_0\varphi)|\lesssim|k|,
\end{aligned}
\end{equation}
the last equality comes from Remark \ref{r1}, which implies that $r(k), r^{\prime}(k), r^{\prime\prime}(k)$ are all bounded near $k=0$. Together with \eqref{eq11}, we can obtain \eqref{eq9}.

    For $k\in\Omega_{1}\cap\{k\in\mathbb{C}:\mathrm{Re}k>k_1\}$, where $k=k_1+le^{i\varphi}$, we obtain
    \begin{align*}
        R_1(k)&=r(k_1)T_0(k_1)^{-2} (k-k_1)^{-2i\nu(k_1)}\left[1-\cos(\kappa_0\varphi)\right]\\
        &+r(\mathrm{Re}k)T_+(k)^{-2}\cos(\kappa_0\varphi),
    \end{align*}
then
\begin{align*}
\bar{\partial}R_1(k)&=\left[r(\mathrm{Re}k) T_+(k)^{-2}-r\left(k_1\right) T_0(k_1)^{-2}\left(k-k_1\right)^{-2 i \nu\left(k_1\right)}\right] \bar{\partial} \cos \left(\kappa_0\varphi\right) \nonumber\\
& +\frac{1}{2} T_+(k)^{-2} r^{\prime}(\mathrm{Re}k) \cos \left(\kappa_0\varphi\right) ,
\end{align*}
we can obtain \eqref{eq8} immediately by the same method we used when $k\in\Omega_{1}\cap\{k\in\mathbb{C}:k_3<\mathrm{Re}k<0\}$.

\end{proof}

Define  a new  function
\begin{equation}
R^{(2)}(k)=\left\{\begin{array}{lll}
\left(\begin{array}{cc}
1 & -R_{1}(k)e^{2it\theta}\\
0 & 1
\end{array}\right), & k\in \Omega_{1},\\
\\
\left(\begin{array}{cc}
1 & 0\\
R_{2}(k)e^{-2it\theta} & 1
\end{array}\right),  &k\in \Omega_{2},\\
\\
I,  &elsewhere,\\
\end{array}\right.
\end{equation}
where   the functions $R_{\ell}(k)$, $\ell=1,2$ are given by  Proposition \ref{p11}.

Make  a transformation
\begin{equation} M^{(2)}(k):=M^{(2)}(y,t;k)=M^{(1)}(k)R^{(2)}(k),\label{eq17}
\end{equation}
then $M^{(2)}(k)$ is a hybrid RH problem which can be detailed as follows:

\begin{RHP}\label{rhp4}
    Find a $2\times 2$ matrix-valued function $M^{(2)}(k)$ with the following properties:
\begin{itemize}
  \item Analyticity: $M^{(2)}(k)$ is continuous in $\mathbb{C}$, sectionally continuous for first-order partial derivatives in $\mathbb{C}\setminus(\Sigma^{(2)}\cup\mathcal{Z}\cup\bar{\mathcal{Z})}$  and analytical in $\mathbb{C}\setminus(\Omega_1\cup\Omega_2)$, where $\Sigma^{(2)}$ is defined in \eqref{eq13};
  \item Jump condition: $M^{(2)}(k)$ has continuous boundary values $M^{(2)}_{\pm}(k)$ on $\Sigma^{(2)}$ and
 \begin{equation*}
M^{(2)}_+(k)=M^{(2)}_-(k)V^{(2)}(k),
\end{equation*}
where
\begin{equation}\label{eq14}
V^{(2)}(k)=\left\{\begin{array}{lll}
\left(\begin{array}{cc}
1 & R_{1}(k)e^{2it\theta}\\
0 & 1
\end{array}\right), & k\in \Sigma_{1},\\
\\
\left(\begin{array}{cc}
1 & 0\\
R_{2}(k)e^{-2it\theta} & 1
\end{array}\right),  &k\in \Sigma_{2}.\\

\end{array}\right.
\end{equation}

  \item Asymptotic behavior:
  $\ M^{(2)}(k)=I+\mathcal{O}(k^{-1}),\ as \ k\to \infty;$
  \item $\bar{\partial}$-Derivative: For $k\in\mathbb{C}$, we have the $\bar{\partial}$-Derivative equation
\begin{equation}
    \bar{\partial}M^{(2)}(k)=M^{(2)}(k)\bar{\partial}R^{(2)}(k),
\end{equation}
where
  \begin{equation}
\bar{\partial}R^{(2)}(k)=\left\{\begin{array}{lll}
\left(\begin{array}{cc}
0 & -\bar{\partial}R_{1}(k)e^{2it\theta}\\
0 & 0
\end{array}\right), & k\in \Omega_{1};\\
\\
\left(\begin{array}{cc}
0 & 0\\
\bar{\partial}R_{2}(k)e^{-2it\theta} & 0
\end{array}\right),  &k\in \Omega_{2};\\
\\
0,  &elsewhere;\\
\end{array}\right.
\end{equation}
\item Residue condition: $M^{(2)}(k)$ has simple poles at each $z_n\in\mathcal{Z}\cup\bar{\mathcal{Z}}$, which  has the same residue condition with $M^{(1)}(k)$ in \eqref{res1}-\eqref{res4}.

\end{itemize}
\end{RHP}

To solve RH problem \ref{rhp4}, we need to decompose it into a pure RH problem by introducing $M^{(2)}_{RHP}$ which has the property of $\bar{\partial}R^{(2)}(k)=0$ on $\mathbb{C}\setminus(\Sigma^{(2)}\cup\mathcal{Z}\cup\bar{\mathcal{Z}})$ and a pure $\bar{\partial}$-RH problem $M^{(3)}(y,t;k)$ with $\bar{\partial}R^{(2)}(k)\neq0$. This process can be shown by the following structure
\begin{equation}
    M^{(2)} = M^{(3)} M^{(2)}_{RHP}=\begin{cases}
    \bar{\partial} R^{(2)} \equiv 0 \rightarrow M^{(2)}_{RHP}, \\
    \bar{\partial} R^{(2)} \neq 0 \rightarrow M^{(3)} = M^{(2)} \left(M^{(2)}_{RHP}\right)^{-1}.
\end{cases}
\end{equation}
For the first step, we establish an RH problem for $M^{(2)}_{RHP}(k)$:

\begin{RHP}\label{rhp5}
    Find a $2\times 2$ matrix-valued function $M^{(2)}_{RHP}(k)$ with the following properties:
\begin{itemize}
  \item Analyticity: $M^{(2)}_{RHP}(k)$ is analytic in $\mathbb{C}\setminus(\Sigma^{(2)}\cup\mathcal{Z}\cup\bar{\mathcal{Z}});$
  \item Jump condition: $M^{(2)}_{RHP}(k)$ has continuous boundary values $M^{(2)}_{RHP\pm}(k)$ on $\Sigma^{(2)}$ and
 \begin{equation*}
M^{(2)}_{RHP+}(k)=M^{(2)}_{RHP-}(k)V^{(2)}(k);
\end{equation*}
\item Symmetry:
$M^{(2)}_{RHP}(k)=\sigma_2 \overline{ M^{(2)}_{RHP}(\bar{k})}\sigma_2=\sigma_2 M^{(2)}_{RHP}(-k)\sigma_2;$

  \item Asymptotic behavior:
  $\ M^{(2)}_{RHP}(k)=I+\mathcal{O}(k^{-1}),\ as \ k\to \infty;$

\item Residue condition: $M^{(2)}_{RHP}(k)$ has simple poles at each $z_n\in\mathcal{Z}\cup\bar{\mathcal{Z}}$ with residue condition \eqref{res1}-\eqref{res4}.
\end{itemize}
\end{RHP}

Define $U(\xi)$ as the union set of the neighborhood of the saddle point $k_j$ for $j=1,\dots,\Lambda$.

$$
U_{\varrho}=\bigcup_{j=1,\dots,\Lambda} U_{\varrho}\left(k_j\right), \text { with } U_{\varrho}\left(k_j\right)=\left\{k:\left|k-k_j\right|<\varrho\right\},
$$

where

$$
\varrho<\frac{1}{2} \min \left\{\min _{\lambda,\mu\in\mathcal{Z}\cup\bar{\mathcal{Z}}}\left|\lambda-\mu\right|, \quad \min _{j=1,\dots,\Lambda}\left|k_j\right|\right\}.
$$

We solve the RHP problem for $M^{(2)}_{RHP}(k)$ by dividing the complex plane into two parts: regions near saddle points and away from saddle points. In the neighborhood of the saddle points, contribution to the solution mainly comes from the jump line, denoted as $M^{(pc)}(k)$, which is considered in Subsection \ref{subs3.3}. While away from the saddle points, contribution mainly comes from spectrum points, denoted as $M^{(out)}(k)$, which is considered in Subsection \ref{subs3.2}. And we denote $E(k)$ as an error matrix. The next two subsections is constructed based this idea:
\begin{equation}\label{eq28}
    M_{R H P}^{(2)}(k)= \begin{cases}E(k) M^{(o u t)}(k), & k \in \mathbb{C} \backslash U_\varrho ,\\ E(k) M^{(o u t)}(k) M^{(pc)}\left(k\right), & k \in U_{\varrho}.\end{cases}
\end{equation}
First we give some estimates on the jump matrix $V^{(2)}(k)$ away from the saddle points $k_j,j=1,\dots,\Lambda.$
\begin{Proposition}\label{p12}
    For $1 \leqslant p \leqslant+\infty$, there exists a constant $h=h(p)>0$, so that the jump matrix $V^{(2)}$ defined in \eqref{eq14} admits the following estimation as $t \rightarrow+\infty$

$$
\left\|V^{(2)}-I\right\|_{L^p\left(\Sigma^{(2)}\backslash U_{\varrho}\right)}=\mathcal{O}\left(e^{-h t}\right).
$$

\end{Proposition}
\begin{proof}
    for $k\in\Sigma_1\setminus U_\varrho$, we have
\begin{align*}
\|V^{(2)}-I\|_{L^p\left(\Sigma_1\backslash U_{\varrho}\right)}&=\|R_1(k)e^{2it\theta}\|_{L^p\left(\Sigma_1\backslash U_{\varrho}\right)}\\
&\leqslant \|R_1(k)\|_{L^\infty\left(\Sigma_1\backslash U_{\varrho}\right)}\|e^{2it\theta}\|_{L^p\left(\Sigma_1\backslash U_{\varrho}\right)}\\
&\lesssim \|e^{2it\theta}\|_{L^p\left(\Sigma_1\backslash U_{\varrho}\right)}.
\end{align*}
We also denote $k=k_j+u+vi=k_j+le^{i\varphi},j=1,\dots,\Lambda$ for $l>\varrho$. Recall the Lemma \ref{p10} about the estimates on $\mathrm{Im}\theta(k)$, we have
\begin{align*}
    \|e^{2it\theta}\|^p_{L^p\left(\Sigma_1\backslash U_{\varrho}\right)}&\lesssim\int_{\Sigma_1\backslash U_{\varrho}}e^{-2tpuv}\mathrm{d}k\\
    &\lesssim\int_\varrho^{+\infty}e^{-tpl}\mathrm{d}l\\
    &\lesssim t^{-1}e^{-p\varrho}.
\end{align*}
 When $k\in\Sigma_2\setminus U_\varrho$, the proposition can be proved in the same way.
\end{proof}
\subsection{\texorpdfstring{Soliton solutions for $M^{(out)}(k)$}{Soliton solutions for M(out)(k)}}\label{subs3.2}
In this part, we consider the model $M^{(out)}(k)$ which has the same residue conditions with $M^{(2)}_{RHP}(k)$ but has no jump conditions on the complex plane. We can prove that it has the property of soliton decomposition. The out model $M^{(out)}(k)$ satisfies the following RH problem.

\begin{RHP}\label{rhp6}
 Find a matrix-valued function $M^{(out)}(k) = M^{(out)}(y, t;k)$ with the following properties:

\begin{itemize}
    \item Analyticity: $M^{(out)}(k)$ is analytical in $\mathbb{C} \setminus (\mathcal{Z} \cup \overline{\mathcal{Z}})$;
    \item Symmetry: $M^{(out)}(\overline{k}) = \overline{M^{(out)}(-k)} = \sigma_2 \overline{M^{(out)}(k) }\sigma_2$;
    \item Asymptotic behaviors: $M^{(out)}(k) \sim I + \mathcal{O}(k^{-1}), \quad k \to \infty; $
    \item Residue conditions: $M^{(out)}(k)$ has simple poles at each point in $\mathcal{Z} \cup \overline{\mathcal{Z}}$ satisfying the same residue relations  with $M_{RHP}^{(2)}(k)$.
\end{itemize}
\end{RHP}

Then we show the reflection-less case($r(k)=0$) for RH problem \ref{rhp4}, for which the jump matrix becomes $V^{(2)}(k)=I$.

\begin{RHP}\label{rhp7}
    Given discrete data $\sigma_d = \{(z_n, c_n)\}_{n=1}^N$. Find a matrix-valued function
$M(k|\sigma_d) = M( y, t;k|\sigma_d)$
with following properties:

\begin{itemize}
    \item Analyticity: $M(k|\sigma_d)$ is analytical in $\mathbb{C} \setminus (\mathcal{Z} \cup \overline{\mathcal{Z}})$;
    \item Symmetry: $\overline{M(\overline{k}|\sigma_d) }= M(-k|\sigma_d) = \sigma_2 M(k|\sigma_d) \sigma_2$;
    \item Asymptotic behaviors:
    $ M(k|\sigma_d) \sim I + \mathcal{O}(k^{-1}), \quad k \to \infty$;
    \item Residue conditions: $M(k|\sigma_d)$ has simple poles at each point in $\mathcal{Z} \cup \overline{\mathcal{Z}}$ satisfying
    \begin{equation*}
        \res_{k=z_n} M(k|\sigma_d) = \lim_{k \to z_n} M(k|\sigma_d) \tau_n,
    \end{equation*}
    \begin{equation*}
         \res_{k=\overline{z}_n} M(k|\sigma_d) = \lim_{k \to \overline{z}_n} M(k|\sigma_d) \widehat{\tau}_n,
    \end{equation*}

    where $\tau_n$ is a nilpotent matrix satisfying
    \begin{equation}
        \tau_n =
    \begin{pmatrix}
    0 & \gamma_n \\
    0 & 0
    \end{pmatrix},
    \quad \widehat{\tau}_n = \sigma_2 \overline{\tau}_n \sigma_2, \quad \gamma_n = c_n e^{2it\theta(z_n)}.
    \end{equation}

\end{itemize}

\end{RHP}

\begin{Proposition}\label{p13}
    The RH problem \ref{rhp7} admits a unique solution in the following form
    \begin{equation*}
M(k|\sigma_d) = I + \sum_{n=1}^N
\begin{pmatrix}
\frac{\varsigma_n}{k - \bar{z}_n} & \frac{-\bar{\iota}_n}{k - z_n} \\
\frac{\iota_n}{k - \bar{z}_n} & \frac{\bar{\varsigma}_n}{k - z_n}
\end{pmatrix} ,
    \end{equation*}
where $\varsigma_h=\varsigma_h(y,t)$ and $\iota_h=\iota_h(y,t)$ satisfies linearly dependent equations:
\begin{align*}
\varsigma_h+\sum_{n=1}^N\frac{\gamma_h\bar{\iota}_n}{z_h-\bar{z}_n}&=0,\\
\iota_h-\sum_{n=1}^N\frac{\gamma_h\bar{\varsigma}_n}{z_h-\bar{z}_n}&=\gamma_h,
\end{align*}
For $h=1,\cdots,N$ respectively.
    And the solution satisfies
    \begin{equation*}
\|M(k|\sigma_d)^{-1}\|_{L^\infty(\mathbb{C} \setminus (\mathcal{Z} \cup \overline{\mathcal{Z}}))} \lesssim 1.
    \end{equation*}
\end{Proposition}
\begin{proof}
    The uniqueness of solution for $M(k|\sigma_d)$ follows from the Liouville theorem.
\end{proof}

\begin{corollary}\label{c2}
    If $u_{sol}(y, t) = u_{sol}(y, t; \sigma_d)$ denotes the $N$-soliton solution for the WKI-SP equation \eqref{wkisp-1} with reflection-less discrete data $\sigma_d$, we obtain the reconstruction formula as follows:
\begin{equation}
u_{sol}(x, t;\sigma_d) = u_{sol}(y(x, t), t;\sigma_d) = \lim_{k \to 0} \frac{\left[M^{-1}( y, t;0|\sigma_d) M( y, t;k|\sigma_d)\right]_{12}}{i k},
\end{equation}
where
\begin{equation}
    y(x,t)=x-c_+(x,t; \sigma_d),
\end{equation}
with
\begin{equation}
    c_+(x, t; \sigma_d)= \lim_{k \to 0} \frac{\left[M^{-1}( y, t;0|\sigma_d) M( y, t;k|\sigma_d)\right]_{11} - 1}{i k}.
\end{equation}

\end{corollary}

Denote the following trace formula
\begin{equation*}
    \omega(k) = \prod_{n=1}^N \frac{k - z_n}{k - \overline{z}_n},
\end{equation*}
whose poles can be separated into two parts . Take the subset $\Delta^-$ of $\mathcal{N}$ and let
\begin{equation*}
\omega_{\Delta^-}(k) = \prod_{n \in {\Delta^-}} \frac{k - z_n}{k - \overline{z}_n}.
\end{equation*}
We make a renormalization transformation
\begin{equation}\label{eq16}
M^{\Delta^-}(k|\sigma_d^{\Delta^-}) = M^{\Delta^-}( y, t;k|\sigma_d^{\Delta^-}) = M( y, t;k|\sigma_d) \omega_{\Delta^-}(k)^{-\sigma_3},
\end{equation}

where the scattering data is given by
\begin{equation}\label{eq15}
    \sigma_d^{\Delta^-} = \{(z_n, \tilde c_n)\}_{n=1}^N, \quad \tilde c_n = \begin{cases}
     c_n \omega^2_{\Delta^-}(z_n),&n\notin\Delta^-\\
        c_n^{-1} \omega_{\Delta^-}^{\prime}(z_n)^{-2},&n\in\Delta^-

    \end{cases}\;,
\end{equation}

then the $M^{\Delta^-}(k|\sigma_d^{\Delta^-})$ satisfies the following RH problem:
\begin{RHP}\label{rhp8}
    Given discrete data $\sigma_d^{\Delta^-}$ in \eqref{eq15}, find a matrix-valued function $M^{\Delta^-}(k|\sigma_d^{\Delta^-})$ with the following properties:

\begin{itemize}
    \item Analyticity: $M^{\Delta^-}(k|\sigma_d^{\Delta^-})$ is analytical in $\mathbb{C} \setminus (\mathcal{Z} \cup \overline{\mathcal{Z}})$;
    \item Symmetry: $M^{\Delta^-}(k|\sigma_d^{\Delta^-}) = \sigma_2\overline{M^{\Delta^-}(\overline{k}|\sigma_d^{\Delta^-})}  \sigma_2 =  \sigma_2M^{\Delta^-}(-\overline{k}|\sigma_d^{\Delta^-}) \sigma_2$;
    \item Asymptotic behaviors:
    \begin{equation*}
         M^{\Delta^-}(k|\sigma_d^{\Delta^-}) \sim I + \mathcal{O}(k^{-1}), \quad k \to \infty;
    \end{equation*}
    \item Residue conditions: $M^{\Delta^-}(k|\sigma_d^{\Delta^-})$ has simple poles at each point in $\mathcal{Z} \cup \overline{\mathcal{Z}}$ satisfying
    \begin{equation*}
          \res_{k = z_n} M^{\Delta^-}(k|\sigma_d^{\Delta^-}) = \lim_{k \to z_n} M^{\Delta^-}(k|\sigma_d^{\Delta^-}) \tau_n^{\Delta^-},
    \end{equation*}
    \begin{equation*}
        \res_{k = \overline{z}_n} M^{\Delta^-}(k|\sigma_d^{\Delta^-}) = \lim_{k \to \overline{z}_n} M^{\Delta^-}(k|\sigma_d^{\Delta^-}) \widehat{\tau}_n^{\Delta^-},
    \end{equation*}
    where $\tau_n^{\Delta^-}$ is a nilpotent matrix satisfying
    \begin{equation}\label{eq34}
    \tau_n^{\Delta^-} =
    \begin{cases}
        \begin{pmatrix}
            0 & \gamma_n \omega_{\Delta^-}^2(z_n) \\
            0 & 0
        \end{pmatrix}, & n \notin \Delta^-, \\
        \begin{pmatrix}
            0 & 0 \\
            \gamma_n^{-1} \omega'_{\Delta^-}(z_n)^{-2} & 0
        \end{pmatrix}, & n \in \Delta^-,
    \end{cases}
    \quad \widehat{\tau}_n^{\Delta^-} = \sigma_2 \overline{\tau}_n^{\Delta^-} \sigma_2^{-1}.
    \end{equation}
\end{itemize}

\end{RHP}

Since the uniqueness of $M(y, t;k|\sigma_d)$ by Proposition \ref{p13} and the transformation \eqref{eq16}, we obtain the existence and uniqueness of the solution for the RH problem \ref{rhp8}.
It can be observed from the residue conditions that the reflectional part of the $M^{(out)}(k)$ comes from $\delta(k)$. Then by replacing the scattering data $\sigma_d^{\Delta^-}$ with the following $\sigma_d^{(out)}$
\begin{equation}\label{eq18}
    \sigma_d^{(out)} = \{(z_n, \hat c_n)\}_{n=1}^N, \quad \hat c_n = \begin{cases}
     c_n \omega^2_{\Delta^-}(z_n)\delta^{-2}(z_n),&n\notin\Delta^-\\
        c_n^{-1} \omega_{\Delta^-}^{\prime}(z_n)^{-2}\delta^{2}(z_n),&n\in\Delta^-

    \end{cases}\;,
\end{equation}
 we can obtain
\begin{Proposition}\label{p14}
    There exists a unique solution for the RH Problem \ref{rhp6} and $M^{(out)}( y, t;k)$ can be obtained by the following transformation
\begin{equation}
    M^{(out)}( y, t;k) = M^{(out)}(k|\sigma_d^{(out)})=M^{\Delta^-}(k|\sigma_d^{\Delta^-})\delta(k)^{-\sigma_3},
\end{equation}
where scattering data $\sigma_d^{(out)}$ is given by \eqref{eq18}. Moreover, the $N$-soliton solution of WKI-SP encoded by RH problem \ref{rhp6} can be reconstructed by
\begin{equation}
    u_{sol}(x, t;\sigma_d^{(out)})=u_{sol}(x, t;\sigma_d).
\end{equation}
\end{Proposition}

\subsection{Localized RH problem near saddle points}\label{subs3.3}

\subsubsection{\texorpdfstring{A local solvable RH model $M^{(pc)}(k)$}{A local solvable RH model M(pc)(k)}}

Now we turn to the localized RH problem near saddle points $k_j,j=1,\dots,\Lambda$. Define the jump contour near the saddle points as follows, which can be shown in Figure \ref{jump-pc} intuitively,
\begin{align*}
    \Sigma^{(pc,k_j)}&=\Sigma\cap U_{\varrho}(k_j),\quad j=1,\dots,\Lambda,\\
\Sigma^{(pc)}&=\bigcup_{j=1}^\Lambda\Sigma^{(pc,k_j)}.
\end{align*}
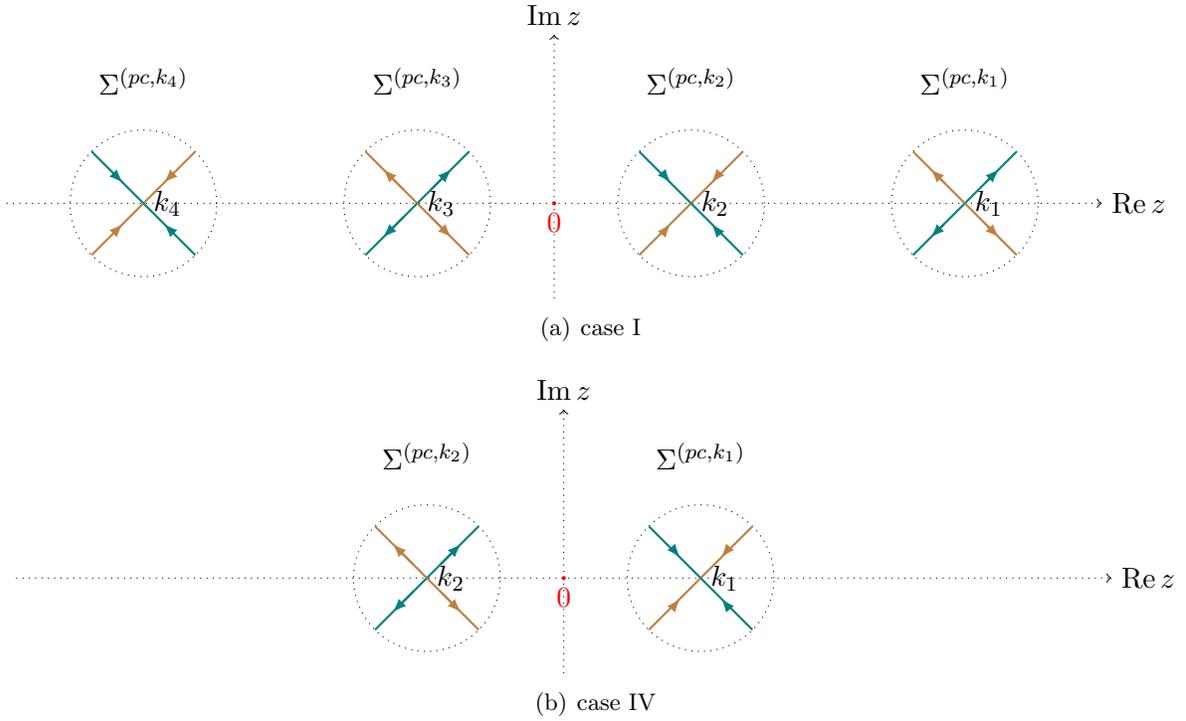
\begin{figure}[h]
    \centering
    \subfigure[case \uppercase\expandafter{\romannumeral1}]
    {\begin{tikzpicture}[scale=0.65]
        \newcommand{\radius}{1.5}
        \newcommand{\gap}{2.8} 

        \draw[dotted] (-3*\gap, 0) circle (\radius);
        \draw[dotted] (-\gap, 0) circle (\radius);
        \draw[dotted] (\gap, 0) circle (\radius);
        \draw[dotted] (3*\gap, 0) circle (\radius);

        \draw[dotted,->] (-4*\gap,0) -- (4*\gap,0) node[right] {$\mathrm{Re}\,z$};
        \draw[dotted,->] (0,-1.3*\radius) -- (0,2.3*\radius) node[above] {$\mathrm{Im}\,z$};

        \draw[brown,thick] (-3*\gap + 1.06, 1.06) -- (-3*\gap, 0);

        \draw[teal,thick] (-3*\gap - 1.06, 1.06) -- (-3*\gap, 0);
        \draw[brown,thick] (-3*\gap, 0) -- (-3*\gap - 1.06, -1.06);
        \draw[teal,thick] (-3*\gap, 0) -- (-3*\gap + 1.06, -1.06);
        \draw[teal,thick] (-\gap + 1.06, 1.06) -- (-\gap, 0);
        \draw[brown,thick] (-\gap - 1.06, 1.06) -- (-\gap, 0);
        \draw[teal,thick] (-\gap, 0) -- (-\gap - 1.06, -1.06);
        \draw[brown,thick] (-\gap, 0) -- (-\gap + 1.06, -1.06);

        \draw[brown,thick] (\gap + 1.06, 1.06) -- (\gap, 0);
        \draw[teal,thick] (\gap - 1.06, 1.06) -- (\gap, 0);
        \draw[brown,thick] (\gap, 0) -- (\gap - 1.06, -1.06);
        \draw[teal,thick] (\gap, 0) -- (\gap + 1.06, -1.06);

        \draw[teal,thick] (3*\gap + 1.06, 1.06) -- (3*\gap, 0);
        \draw[brown,thick](3*\gap - 1.06, 1.06) -- (3*\gap, 0);
        \draw[teal,thick] (3*\gap, 0) -- (3*\gap - 1.06, -1.06);
        \draw[brown,thick] (3*\gap, 0) -- (3*\gap + 1.06, -1.06);

        \node[right] at (-3*\gap, 0) {$k_4$};
        \node[right] at (-\gap, 0) {$k_3$};
        \node[right] at (\gap, 0) {$k_2$};
        \node[right] at (3*\gap, 0) {$k_1$};

\draw[-latex,teal,thick](-3.2*\gap,0.2*\gap)--(-3.15*\gap,0.15*\gap);
\draw[-latex,brown,thick](-3.2*\gap,-0.2*\gap)--(-3.15*\gap,-0.15*\gap);
\draw[-latex,teal,thick ](-2.8*\gap,-0.2*\gap)--(-2.85*\gap,-0.15*\gap);
\draw[-latex,brown,thick](-2.8*\gap,0.2*\gap)--(-2.85*\gap,0.15*\gap);

\draw[-latex,brown,thick](-1.1*\gap,0.1*\gap)--(-1.25*\gap,0.25*\gap);
\draw[-latex,teal,thick](-1.1*\gap,-0.1*\gap)--(-1.25*\gap,-0.25*\gap);
\draw[-latex,brown,thick](-0.9*\gap,-0.1*\gap)--(-0.75*\gap,-0.25*\gap);
\draw[-latex,teal,thick](-0.9*\gap,0.1*\gap)--(-0.75*\gap,0.25*\gap);

\draw[-latex,teal,thick](0.8*\gap,0.21*\gap)--(0.85*\gap,0.15*\gap);
\draw[-latex,brown,thick](0.8*\gap,-0.21*\gap)--(0.85*\gap,-0.15*\gap);
\draw[-latex,teal,thick](1.2*\gap,-0.21*\gap)--(1.15*\gap,-0.15*\gap);
\draw[-latex,brown,thick](1.2*\gap,0.21*\gap)--(1.15*\gap,0.15*\gap);

\draw[-latex,brown,thick](2.85*\gap,0.15*\gap)--(2.75*\gap,0.25*\gap);
\draw[-latex,teal,thick](2.85*\gap,-0.15*\gap)--(2.75*\gap,-0.25*\gap);
\draw[-latex,brown,thick](3.15*\gap,-0.15*\gap)--(3.25*\gap,-0.25*\gap);
\draw[-latex,teal,thick](3.15*\gap,0.15*\gap)--(3.25*\gap,0.25*\gap);
 \node[scale=1] at (3*\gap,2.5) {$\Sigma^{(pc,k_1)}$};
  \node[scale=1] at (1*\gap,2.5) {$\Sigma^{(pc,k_2)}$};
   \node[scale=1] at (-1*\gap,2.5) {$\Sigma^{(pc,k_3)}$};
    \node[scale=1] at (-3*\gap,2.5) {$\Sigma^{(pc,k_4)}$};
    \coordinate (I) at (0,0);
		\fill[red] (I) circle (1.2pt) node[below] {$0$};
    \end{tikzpicture}}
    \subfigure[case \uppercase\expandafter{\romannumeral4}]{    \begin{tikzpicture}[scale=0.65]
        \newcommand{\radius}{1.5}
        \newcommand{\gap}{2.8} 


        \draw[dotted] (-\gap, 0) circle (\radius);
        \draw[dotted] (\gap, 0) circle (\radius);

        \draw[dotted,->] (-4*\gap,0) -- (4*\gap,0) node[right] {$\mathrm{Re}\,z$};
        \draw[dotted,->] (0,-1.3*\radius) -- (0,2.3*\radius) node[above] {$\mathrm{Im}\,z$};


        \draw[teal,thick] (-\gap + 1.06, 1.06) -- (-\gap, 0);
        \draw[brown,thick] (-\gap - 1.06, 1.06) -- (-\gap, 0);
        \draw[teal,thick] (-\gap, 0) -- (-\gap - 1.06, -1.06);
        \draw[brown,thick] (-\gap, 0) -- (-\gap + 1.06, -1.06);

        \draw[brown,thick] (\gap + 1.06, 1.06) -- (\gap, 0);
        \draw[teal,thick] (\gap - 1.06, 1.06) -- (\gap, 0);
        \draw[brown,thick] (\gap, 0) -- (\gap - 1.06, -1.06);
        \draw[teal,thick] (\gap, 0) -- (\gap + 1.06, -1.06);



        \node[right] at (-\gap, 0) {$k_2$};
        \node[right] at (\gap, 0) {$k_1$};


\draw[-latex,brown,thick](-1.1*\gap,0.1*\gap)--(-1.25*\gap,0.25*\gap);
\draw[-latex,teal,thick](-1.1*\gap,-0.1*\gap)--(-1.25*\gap,-0.25*\gap);
\draw[-latex,brown,thick](-0.9*\gap,-0.1*\gap)--(-0.75*\gap,-0.25*\gap);
\draw[-latex,teal,thick](-0.9*\gap,0.1*\gap)--(-0.75*\gap,0.25*\gap);

\draw[-latex,teal,thick](0.8*\gap,0.21*\gap)--(0.85*\gap,0.15*\gap);
\draw[-latex,brown,thick](0.8*\gap,-0.21*\gap)--(0.85*\gap,-0.15*\gap);
\draw[-latex,teal,thick](1.2*\gap,-0.21*\gap)--(1.15*\gap,-0.15*\gap);
\draw[-latex,brown,thick](1.2*\gap,0.21*\gap)--(1.15*\gap,0.15*\gap);

  \node[scale=1] at (1*\gap,2.5) {$\Sigma^{(pc,k_1)}$};
   \node[scale=1] at (-1*\gap,2.5) {$\Sigma^{(pc,k_2)}$};

    \coordinate (I) at (0,0);
		\fill[red] (I) circle (1.2pt) node[below] {$0$};
    \end{tikzpicture}}
    \caption{Jump contour $\Sigma^{(pc)}$ of $M^{(pc, k_j)}(k)$, $j=1,\dots,\Lambda$.}
    \label{jump-pc}
\end{figure}

\FloatBarrier

Next we give the localized RH problem for each saddle point $k_j,j=1,\dots,\Lambda$ respectively.

\begin{RHP}
       Find a $2\times 2$ matrix-valued function $M^{(pc,k_j)}(y,t;k)$ with the following properties:
       \begin{itemize}
  \item Analyticity: $M^{(pc,k_j)}(y,t;k)$ is meromorphic in $\mathbb{C}\setminus \Sigma^{(pc,k_j)}$ ;
  \item Jump condition: $M^{(pc,k_j)}(y,t;k)$ has continuous boundary values $M^{(pc,k_j)}_{\pm}(k)$ on $\Sigma^{(pc,k_j)}$ and
 \begin{equation*}
M^{(pc,k_j)}_+(k)=M^{(pc,k_j)}_-(y,t;k)V^{(pc,k_j)}(k),
\end{equation*}
where
\begin{equation*}
V^{(pc,k_j)}(k)=\left\{\begin{array}{lll}
\left(\begin{array}{cc}
1 & \rho(k_j)T_0(k_j)^{-2} (k-k_j)^{-2\eta(k_j)i\nu(k_j)}e^{2it\theta}\\
0 & 1
\end{array}\right), & k\in \Sigma_{1};\\
\\
\left(\begin{array}{cc}
1 & 0\\
\bar{\rho}(k_j)T_0(k_j)^{2}(k-k_j)^{2\eta(k_j)i\nu(k_j)}e^{-2it\theta} & 1
\end{array}\right),  &k\in \Sigma_{2};\\

\end{array}\right.
\end{equation*}
  \item Asymptotic behavior:
  $\ M^{(pc,k_j)}(y,t;k)=I+\mathcal{O}(k^{-1}),\quad as \ k\to \infty.$
\end{itemize}
\end{RHP}

It is well known fact that the localized model $ M^{(pc,k_j)}(y,t;k)$ mentioned above can be constructed  by the solution of the parabolic cylinder (Webb) equation. To match the parabolic cylinder equation with the localized models in this paper, we need to introduce a scaling function $\mathrm{P}_{k_j}$ which maps $k_j$ to the origin and unifies the free variables.

For $k$ near $k_j, j=1,\dots,\Lambda$, we have
\begin{equation}\label{eq21}
\theta(k)=\theta\left(k_j\right)+\frac{\theta^{\prime \prime}\left(k_j\right)}{2}\left(k-k_j\right)^2+\mathcal{O}\left(\left|k-k_j\right|^3\right), \quad k \rightarrow k_j.
\end{equation}


\begin{remark}\label{r3}
    In the expansion of $\theta(k)$ in \eqref{eq21}, the higher order term as $k\to k_j$ can be ignored as $t\to+\infty$. Rewrite $\theta(k)$ as\[
\theta(k) = \theta(k_j) + \frac{\theta''(k_j)}{2} (k - k_j)^2 + \theta_c (k - k_j)^3,
\]

where $\theta_c = \frac{\theta'''(\lambda k_j + (1-\lambda)k)}{3!}$, $\lambda \in (0, 1)$ is the coefficient of remainder.
Recall the scaling function $\mathrm{P}_{k_j}$ we define in \eqref{eq22}, we have the following transformation
\begin{equation*}
e^{2it\theta(k)} = e^{2it(\mathrm{P}_{k_j}\theta)(\zeta)}= e^{2it\theta(k_j)} \cdot e^{i\zeta^2}\cdot e^{\mathrm{P}_{k_j}(\theta_c(k-k_j)^3)}.
\end{equation*}
It can be calculated that with $\zeta$ near $0$,
\begin{equation*}
\left| e^{\mathrm{P}_{k_j}(\theta_c(k-k_j)^3)} \right|  \to 1, \quad \text{as } t \to +\infty.
\end{equation*}
\end{remark}

As a result, for $k \in U_{\varrho}\left(k_j\right)$, we define the rescaled variable $\zeta$ by
\begin{align}
\zeta(k)=\left[2 \eta(k_j)t \theta^{\prime \prime}\left(k_j\right)\right]^{\frac{1}{2}}\left(k-k_j\right),\quad j=1,\dots,\Lambda. \label{varib}
\end{align}

And the scaling function $\mathrm{P}_{k_j}$ admits the following mapping

\begin{equation}\label{eq22}
\begin{aligned}
\mathrm{P}_{k_j}  : U_{\varrho}\left(k_j\right) &\longrightarrow U_0, \quad j=1,\dots,\Lambda, \\
 k &\longmapsto \zeta
\end{aligned}
\end{equation}
where $U_0$ is a neighborhood of $\zeta=0$.
Through this change of variable \eqref{varib}, each  local RH problem  for $M^{(pc, k_j)}(k)$, $j = 1,\dots, \Lambda$ can  match up with the jump of a parabolic cylinder model  in Appendix \ref{appendix1}.

For $j=1,3$ of case I and $j=2$ of case IV,
by setting $r_0$ with
$$
 r_j\equiv r\left(k_j\right) T_0^{-2}(k_j) e^{2it\theta(k_j)}\exp {\left[i \eta(k_j) \nu\left(k_j\right) \log \left(2\eta(k_j) t \theta^{\prime \prime}\left(k_j\right)\right)\right]},
$$
we have
\begin{equation}\label{eq27}
M^{(pc, k_j)}( k) =  M^{(pc)} \left( \zeta(k)  \right)
 = I + \frac{1}{ \zeta}
\begin{pmatrix}
0 &  -i \beta_{12}(r_j)  \\
  i \beta_{21} (r_j) & 0
\end{pmatrix} + \mathcal{O}(\zeta^{-2}),
\end{equation}
where $ \beta_{12}(r_j),  \beta_{21}(r_j)$ are defined by \eqref{eq24}.

For $j=2,4$ of case \uppercase\expandafter{\romannumeral1} and $j=1$ of case \uppercase\expandafter{\romannumeral4},  by setting $r_0$ with
\begin{equation*}
    r_j\equiv -\frac{\bar r\left(k_j\right)}{ 1+|r(k_j)|^2} T_0^{ 2}(k_j) e^{2it\theta(k_j)}\exp {\left[i \eta(k_j) \nu\left(k_j\right) \log \left(2\eta(k_j) t \theta^{\prime \prime}\left(k_j\right)\right)\right]}
\end{equation*}
we have
\begin{equation}\label{eq27-1}
M^{(pc, k_j)}( k) =\sigma_1 M^{(pc)} \left( \zeta(k)  \right) \sigma_1
 = I + \frac{1}{ \zeta}
\begin{pmatrix}
0 &   i \beta_{21}(r_j)  \\
 - i \beta_{12}(r_j)  & 0
\end{pmatrix} + \mathcal{O}(\zeta^{-2}),
\end{equation}
where $\beta_{12}(r_j) $ and $\beta_{21} (r_j)$ are defined by \eqref{eq24}.

Now we consider a new RH problem $M^{(pc)}(k)$ which takes all models near saddle points into consideration.

\begin{RHP}\label{rhp14}
    Find a $2 \times 2$ matrix-valued function $M^{(pc)}(k)$ such that
    \begin{itemize}
        \item Analyticity: $M^{(pc)}(k)$ is analytical in  $\mathbb{C}\setminus\Sigma^{(pc)}$;
        \item Symmetry: $M^{(pc)}(k)=\sigma_2 \overline{ M^{(pc)}(\bar{k})}\sigma_2=\sigma_2 M^{(pc)}(-k)\sigma_2;$
        \item Jump condition:  $M^{(pc)}(k)$ takes continuous boundary values $M^{(pc)}_\pm(k)$ on $\Sigma^{(pc)}$ with jump relation
\begin{equation*}
M_+^{(pc)}(k) = M_-^{(pc)}(k)V^{(pc)}(k), \quad k \in \Sigma^{(pc)},
\end{equation*}
where
\begin{equation*}
    V^{(pc)}(k)=V^{(2)}(k)|_{\Sigma^{(pc)}};
\end{equation*}
\item Asymptotic behavior:
\begin{equation*}
    M^{(pc)}(k)=I+\mathcal{O}(k^{-1}),\quad k\to\infty.
\end{equation*}
    \end{itemize}
\end{RHP}

 As $V^{(2)}(k)$ is either a lower or a upper matrix with $1$ on the diagonal, for $k\in\Sigma^{(pc,k_j)},$ we denote
 \begin{equation*}
     V^{(pc)}(k)=I+w_{j}(k),\quad j=1,\dots,\Lambda.
 \end{equation*}

Recall the Cauchy projection operator $C_\pm$ on $\Sigma^{(pc,k_j)}$, $j= 1,\dots, \Lambda,$
\begin{equation*}
C_\pm f(k) = \lim_{s \to k^{\pm}, k \in \Sigma^{(pc,k_j)}} \frac{1}{2\pi i} \int_{\Sigma^{(pc,k_j)}} \frac{f(s)}{s - k} \mathrm{d}s.
\end{equation*}

Define the following operator on $\Sigma^{(pc,k_j)}$, $j = 1,\dots, \Lambda$ as follows
\begin{equation*}
    C_{w_{j}}(f) := C_-\left(f w_{j}\right).
\end{equation*}
 Then we give some notations as follows:
\begin{align*}
     w=\sum_{j=1}^\Lambda w_j,\quad C_w=\sum_{j=1}^\Lambda C_{w_j}.
\end{align*}

\begin{Proposition}\label{p15}
    RH problem \ref{rhp14} has a unique solution which can be expressed by the following equation: \begin{equation*}
        M^{(pc)}(k) = I + \frac{1}{2\pi i} \int_{\Sigma^{(pc)}} \frac{(1 - C_w)^{-1}  w}{s - k} \mathrm{d}s.
    \end{equation*}
    And $M^{(pc)}(k)$ has the following asymptotics as $t\to\infty$
    \begin{equation*}
        M^{(pc)}(k) = I + t^{-\frac{1}{2}} \sum_{j=1}^{\Lambda} \frac{i  A_j^{mat}}{\left[2  \eta(k_j)  \theta^{\prime\prime}(k_j)\right]^{\frac{1}{2}} (k - k_j)} + \mathcal{O}(t^{-1}),
    \end{equation*}
    where
    \begin{equation}\label{eq29}
        A_j^{mat} =
        \begin{cases}
           \begin{pmatrix}
0 & -\beta_{12}(r_j)\\
\beta_{21}(r_j) & 0
\end{pmatrix}, & j=1,3 \;\textup{of case \uppercase\expandafter{\romannumeral1}} ,j=2 \;\textup{of case \uppercase\expandafter{\romannumeral4}}, \\
 \begin{pmatrix}
0 & \beta_{21}(r_j)\\
-\beta_{12}(r_j) & 0
\end{pmatrix}, & j=2,4 \;\textup{of case \uppercase\expandafter{\romannumeral1}} ,j=1 \;\textup{of case \uppercase\expandafter{\romannumeral4}}.
        \end{cases}
    \end{equation}
\end{Proposition}

To prove Proposition \ref{p15}, we need the following lemmas.
\begin{lemma}\label{l4}
    The matrix functions $w_{j}$ we define above admit the following asymptotics as $t\to\infty$:
\begin{equation*}
    \| w_{j} \|_{L^2(\Sigma^{(pc)})} = \mathcal{O}(t^{-\frac{1}{2}}).
\end{equation*}
\end{lemma}

\begin{lemma}\label{l5}
    As $t \to +\infty$, for $j \neq m$
\begin{equation*}
    \|C_{w_j} C_{w_m}\|_{L^2(\Sigma^{(pc)})} = \mathcal{O}(t^{-1}), \quad
\|C_{w_j} C_{w_m}\|_{L^\infty(\Sigma^{(pc)}) \to L^2(\Sigma^{(pc)})} = \mathcal{O}(t^{-1}).
\end{equation*}
\end{lemma}

\begin{lemma}\label{l6}
    As $t \to +\infty$,
\begin{equation*}
    \int_{\Sigma^{(pc)}} \frac{(1 - C_w)^{-1}  w}{s - k} \mathrm{d}s
= \sum_{j=1}^\Lambda \int_{\Sigma^{(pc, k_j)}} \frac{(1 - C_{w_j})^{-1} w_j}{s - k}\mathrm{d}s + \mathcal{O}(t^{-1}).
\end{equation*}
\end{lemma}
The last two lemmas reveal that the contribution to $M^{(pc)}(k)$ can be separated by each $M^{(pc,k_j)}(k),j=1,\dots,\Lambda.$ Combined with the result we reach at \eqref{eq27}-\eqref{eq27-1}, we can finally prove the Proposition \ref{p15}.

\subsubsection{Small normed RH problem}
As the idea we show in \eqref{eq28}, the error matrix function is defined by
\begin{equation*}
    E(k) =
\begin{cases}
M^{(2)}_{RHP}(k)M^{(out)}(k)^{-1}, & k \in \mathbb{C} \setminus U_{\varrho}, \\
M^{(2)}_{RHP}(k)\left(M^{(out)}(k)M^{(pc)}(k)\right)^{-1}, & k \in U_{\varrho}.
\end{cases}
\end{equation*}

RH problem for $E(k)$ are as follows.
\begin{RHP}
    Find a $2 \times 2$ matrix-valued function $E(k)$ such that  \begin{itemize}
        \item Analyticity: $E(k)$ is analytical in $\mathbb{C} \setminus \Sigma^{(E)}$, where
\begin{equation*}
    \Sigma^{(E)} := \partial U_{\varrho} \cup \left( \Sigma^{(2)} \setminus U_{\varrho}\right);
\end{equation*}
\item Jump condition: $E(k)$ takes continuous boundary values $E_\pm(k)$ on $\Sigma^{(E)}$ and
\begin{equation*}
E_+(k) = E_-(k) V^{(E)}(k),
\end{equation*}
where
\begin{equation*}
V^{(E)}(k) =
\begin{cases}
M^{(out)}(k)V^{(2)}(k)M^{(out)}(k)^{-1}, & k \in \Sigma^{(2)} \setminus U_{\varrho}; \\
M^{(out)}(k)M^{(pc)}(k)M^{(out)}(k)^{-1}, & k \in \partial U_{\varrho};
\end{cases}
\end{equation*}
\item Asymptotic behavior: $E(k) = I + \mathcal{O}(k^{-1}), \quad k \to \infty .$
  \end{itemize}
\end{RHP}

\begin{figure}[ht]
    \centering
    \subfigure[case \uppercase\expandafter{\romannumeral1}]
{\begin{tikzpicture}[scale=1.55]
\draw[-latex,dotted](-4.7,0)--(4.7,0)node[right]{ \textcolor{black}{Re$k$}};

 \draw [thick](3,0) circle (0.32);
  \draw [thick](1,0) circle (0.32);
   \draw [thick](-3,0) circle (0.32);
    \draw [thick](-1,0) circle (0.32);
\draw[thick,-latex] (2.68,0) to [out=90,in=180] (3.03,0.312);
\draw[thick,-latex] (0.68,0) to [out=90,in=180] (1.03,0.312);
\draw[thick,-latex] (-1.32,0) to [out=90,in=180] (-0.97,0.312);
\draw[thick,-latex] (-3.32,0) to [out=90,in=180] (-2.97,0.312);

\draw [teal,thick](-4,0.6) to [out=0,in=135] (-3.24,0.25);
\draw [brown,thick] (-4,-0.6) to [out=0,in=-135] (-3.24,-0.25);

\draw [brown,thick](-2.76,0.25) to [out=45,in=135] (-1.24,0.25);
\draw [teal,thick](-2.76,-0.25) to [out=-45,in=-135] (-1.24,-0.25);

\draw [teal,thick](-0.76,0.25) to [out=30,in=120] (0,0);
\draw [brown,thick](-0.76,-0.25) to [out=-30,in=-120] (0,0);

\draw [teal,thick] (0.76,0.25)to [out=150,in=60](0,0) ;
\draw [brown,thick](0.76,-0.25) to [out=-150,in=-60] (0,0);

\draw [brown,thick](2.76,0.25) to [out=135,in=45] (1.24,0.25);
\draw [teal,thick](2.76,-0.25) to [out=-135,in=-45] (1.25,-0.25);

\draw [teal,thick](3.23,0.24) to [out=45,in=180] (4,0.6);
\draw [brown,thick](3.23,-0.24) to [out=-45,in=-180] (4,-0.6);

\draw[-latex,teal,thick](-3.8,0.58)--(-3.65,0.53);
\draw[-latex,brown,thick](-3.8,-0.58)--(-3.65,-0.53);
\draw[-latex,teal,thick ](-1.8,-0.55)--(-2,-0.58);
\draw[-latex,brown,thick](-1.8,0.55)--(-2,0.58);

\draw[-latex,teal,thick](-0.24,0.25)--(-0.15,0.18);
\draw[-latex,teal,thick](0.15,0.18)--(0.24,0.25);
\draw[-latex,brown,thick](-0.24,-0.25)--(-0.15,-0.18);
\draw[-latex,brown,thick](0.15,-0.18)--(0.24,-0.25);

\draw[-latex,brown,thick](2.2,0.55)--(2,0.58);
\draw[-latex,teal,thick](2.2,-0.55)--(2,-0.58);

\draw[-latex,teal,thick](3.65,0.53)--(3.8,0.58);
\draw[-latex,brown,thick](3.65,-0.53)--(3.8,-0.59);

\coordinate (A) at (1,0);
\fill (A) circle (1.2pt) node[right]{$k_2$};
\coordinate (B)  at (3,0);
\fill (B) circle (1.2pt) node[right]{$k_1$};
\coordinate (C)  at (-1,0);
\fill (C) circle (1.2pt) node[right]{$k_3$};
\coordinate (D)  at (-3,0);
\fill (D) circle (1.2pt) node[right]{$k_4$};

\coordinate (I) at (0,0);
		\fill[black] (I) circle (1pt) node[below,scale=0.9] {$0$};

\end{tikzpicture}}
   \subfigure[case \uppercase\expandafter{\romannumeral4}]{\begin{tikzpicture}[scale=1.55]
\draw[-latex,dotted](-4.7,0)--(4.7,0)node[right]{ \textcolor{black}{Re$k$}};

  \draw [thick](1,0) circle (0.32);
    \draw [thick](-1,0) circle (0.32);

\draw[thick,-latex] (0.68,0) to [out=90,in=180] (1.03,0.312);
\draw[thick,-latex] (-1.32,0) to [out=90,in=180] (-0.97,0.312);

\draw [brown,thick](-2,0.6) to [out=0,in=135] (-1.24,0.25);
\draw [teal,thick] (-2,-0.6) to [out=0,in=-135] (-1.24,-0.25);

\draw [teal,thick](-0.76,0.25) to [out=30,in=120] (0,0);
\draw [brown,thick](-0.76,-0.25) to [out=-30,in=-120] (0,0);

\draw [teal,thick] (0.76,0.25)to [out=150,in=60](0,0) ;
\draw [brown,thick](0.76,-0.25) to [out=-150,in=-60] (0,0);

\draw [brown,thick](1.24,0.25) to [out=45,in=180] (2,0.6);
\draw [teal,thick](1.24,-0.25) to [out=-45,in=-180] (2,-0.6);

\draw[-latex,brown,thick](-1.65,0.53)--(-1.8,0.58);
\draw[-latex,teal,thick](-1.65,-0.53)--(-1.8,-0.58);

\draw[-latex,teal,thick](-0.24,0.25)--(-0.15,0.18);
\draw[-latex,teal,thick](0.15,0.18)--(0.24,0.25);
\draw[-latex,brown,thick](-0.24,-0.25)--(-0.15,-0.18);
\draw[-latex,brown,thick](0.15,-0.18)--(0.24,-0.25);

\draw[-latex,brown,thick](1.8,0.58)--(1.65,0.54);
\draw[-latex,teal,thick](1.8,-0.58)--(1.65,-0.54);

\coordinate (A) at (1,0);
\fill (A) circle (1.2pt) node[right]{$k_1$};
\coordinate (C)  at (-1,0);
\fill (C) circle (1.2pt) node[right]{$k_2$};

\coordinate (I) at (0,0);
		\fill[black] (I) circle (1pt) node[below,scale=0.9] {$0$};

\end{tikzpicture}}
    \caption{Jump contour of $E(k)$.}
    \label{fig:jump-E}
\end{figure}
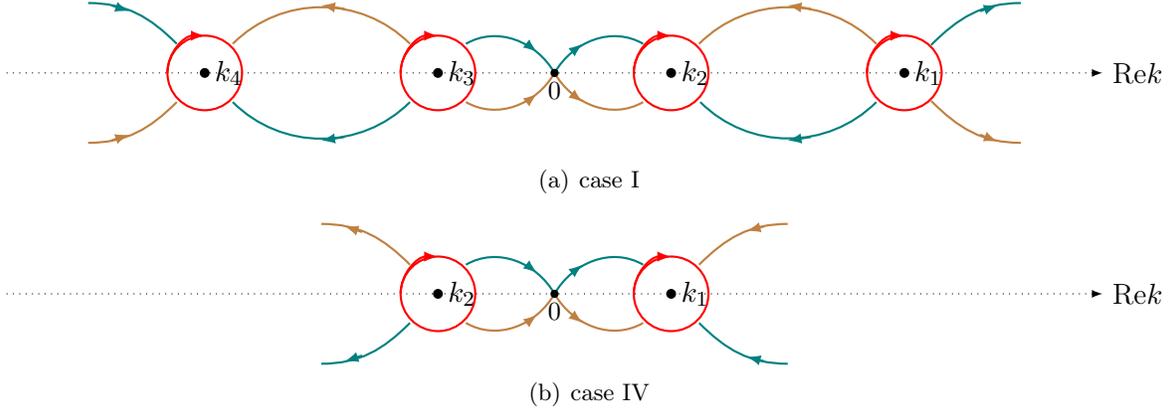
\FloatBarrier

Considering Proposition \ref{p12}, we can know that $V^{(E)}(k)$ exponentially decay to $I$ for $k \in \Sigma^{(2)} \setminus U_{\varrho}$. For $k \in \partial U_{\varrho}$, as $M^{(out)}(k)$ is bounded, we obtain that
\begin{align}
|V^{(E)} - I| &= |M^{(out)}(k) M^{(PC)}(k) M^{(out)}(k)^{-1} - I| \notag \\
&= |M^{(out)}(k)(M^{(PC)}(k) - I)M^{(out)}(k)^{-1}| \notag \\
&= \mathcal{O}(t^{-\frac{1}{2}}).
\end{align}

According to Beals-Coifman theory, the solution for $E(k)$ can be given by
\begin{equation}\label{eq30}
E(k) = I + \frac{1}{2\pi i} \int_{\Sigma^{(E)}} \frac{(I + \varpi_E(s))(V^{(E)}(s) - I)}{s - k} \mathrm{d}s,
\end{equation}
where $\varpi_E \in L^2(\Sigma^{(E)})$ is the unique solution of $(1 - C_{V^{(E)}})\varpi_E = C_{V^{(E)}} I$. And $C_{V^{(E)}}: L^2(\Sigma^{(E)}) \to L^2(\Sigma^{(E)})$ is the Cauchy operator on $\Sigma^{(E)}$, which is defined as:
\begin{align*}
C_{V^{(E)}}(f)(k) = C_{-}f(V^{(E)} - I) = \lim_{s \to k^-, k \in \Sigma^{(E)}} \int_{\Sigma^{(E)}} \frac{f(s)(V^{(E)}(s) - I)}{s - k} \mathrm{d}s.
\end{align*}

Existence and uniqueness of $\varpi_E$ comes from the boundedness of the Cauchy operator $C_{-}$, which admits
\begin{equation*}
\|C_{V^{(E)}}\|_{L^2(\Sigma^{(E)})} \leqslant \|C_{-}\|_{L^2(\Sigma^{(E)}) \to L^2(\Sigma^{(E)})} \|V^{(E)} - I\|_{L^\infty(\Sigma^{(E)})} = \mathcal{O}(t^{-\frac{1}{2}}).
\end{equation*}

In addition,
\begin{equation}\label{eq31}
    \|\varpi_E\|_{L^2(\Sigma^{(E)})} \lesssim \frac{\|C_{V^{(E)}}\|_{L^2(\Sigma^{(E)})}}{1 - \|C_{V^{(E)}}\|_{L^2(\Sigma^{(E)})}} \lesssim t^{-\frac{1}{2}}.
\end{equation}

For the convenience of the long time asymptotics, we need to give the asymptotic of $E(k)$ as $k\to0$. Denote
\begin{equation}
    E(k)=E_0+E_1k+\mathcal{O}(k^2),\quad k\to0,
\end{equation}
we can obtain the following asymptotics as $t\to\infty$:

\begin{Proposition}\label{p16}
    As $t\to\infty$, we have
    \begin{align}
     E_0&=I+t^{-\frac{1}{2}}\widehat E_0 + \mathcal{O}(t^{-1}),\label{eq32}\\
        E_1&=t^{-\frac{1}{2}}\widehat E_1+\mathcal{O}(t^{-1}),
    \end{align}
    where
    \begin{align}
       \widehat E_0&= \sum_{j=1}^\Lambda \frac{i  }{\left[2  \eta(k_j)  \theta^{\prime\prime}(k_j)\right]^{\frac{1}{2}} k_j}M^{(out)}(k_j)A_j^{mat}M^{(out)}(k_j)^{-1},\\
    \widehat E_1&=\sum_{j=1}^\Lambda \frac{i  }{\left[2  \eta(k_j) \theta^{\prime\prime}(k_j)\right]^{\frac{1}{2}} k_j^2}M^{(out)}(k_j)A_j^{mat}M^{(out)}(k_j)^{-1},
    \end{align}
    with $A_j^{mat}$ is defined in \eqref{eq29}.
\end{Proposition}
\begin{proof}
    Recall \eqref{eq30}, we know that
    \begin{equation}
        E_0=I+\frac{1}{2\pi i}\int_{\Sigma^{(E)}}\frac{(I + \varpi_E(s))(V^{(E)}(s) - I)}{s} \mathrm{d}s:=I+I_1+I_2+I_3,
    \end{equation}
where
\begin{align}
I_1 &= \frac{1}{2\pi i} \oint_{\partial U_\varrho} \frac{V^{(E)}(s) - I }{s}\mathrm{d}s,\\
I_2 &= \frac{1}{2\pi i} \int_{\Sigma^{(E)} \setminus U_\varrho} \frac{V^{(E)}(s) - I }{s}\mathrm{d}s,\\
I_3& = \frac{1}{2\pi i} \int_{\Sigma^{(E)}}  \frac{\varpi(s)\left(V^{(E)}(s) - I\right) }{s} \mathrm{d}s.
\end{align}
Using Proposition \ref{p12} and \eqref{eq31}, we obtain $|I_2| =|I_3|= \mathcal{O}(t^{-1})$.
To calculate $I_1$ ,
\begin{align*}
    I_1 &= \frac{1}{2\pi i} \oint_{\partial U_{\varrho}} \frac{M^{(out)}(s) \left( M^{(pc)}(s) - I \right) M^{(out)}(s)^{-1}}{s} \mathrm{d}s\\
&= \frac{1}{2\pi i} \sum_{j=1}^\Lambda \oint_{\partial U_{\varrho}(k_j)} \frac{i M^{(out)}(s) A_j^{mat}M^{(out)}(s)^{-1}}{\left[2\eta(k_j)t \theta''(k_j)\right]^{\frac{1}{2}} s(s - k_j)}  \mathrm{d}s
+ \mathcal{O}(t^{-1})\\
&=t^{-\frac{1}{2}}\sum_{j=1}^\Lambda \frac{iM^{(out)}(k_j)A_j^{mat}M^{(out)}(k_j)^{-1} }{\left[2 \eta(k_j)  \theta^{\prime\prime}(k_j)\right]^{\frac{1}{2}} k_j} + \mathcal{O}(t^{-1}),
\end{align*}
where the last equation comes from the residue theorem. Summarizing $I_1$, $I_2$, and $I_3$, we obtain \eqref{eq32}. And $E_1$ can be proved similarly, we only give the formula for $E_1$ here
\begin{equation*}
        E_1=\frac{1}{2\pi i}\int_{\Sigma^{(E)}}\frac{(I + \varpi_E(s))(V^{(E)}(s) - I)}{s^2} \mathrm{d}s.
    \end{equation*}

\end{proof}
\subsection{\texorpdfstring{Analysis on pure $\bar{\partial}$-problem}{Analysis on pure dbar-problem}}\label{subs3.4}
In this section, we deal with matrix function $M^{(3)}(k)$ which generates the contribution from the non-analytical part of  $M^{(2)}(k)$. Define
\begin{equation}\label{eq33}
    M^{(3)}(k) = M^{(2)}(k)M_{RHP}^{(2)}(k)^{-1},
\end{equation}
Then $M^{(3)}$ satisfies the following $\bar{\partial}$ problem.
\begin{dbar}\label{dbar1}
    Find a $2 \times 2$ matrix-valued function $M^{(3)}(k)$ such that
\begin{itemize}
    \item Analyticity: $M^{(3)}(k)$ is continuous in $\mathbb{C}$ and analytic in $\mathbb{C} \setminus \overline{\Omega}$;
    \item Asymptotic behavior: $M^{(3)}(k) = I + \mathcal{O}(k^{-1}), \quad k \to \infty$;
    \item $\bar{\partial}$-Derivative: For $k \in \mathbb{C}$, we have
    \begin{equation*}
        \quad \bar{\partial} M^{(3)}(k) = M^{(3)}(k) W^{(3)}(k),
    \end{equation*}
    with
    \begin{equation*}
        W^{(3)} = M_{RHP}^{(2)}(k) \bar{\partial} R^{(2)}(k) M_{RHP}^{(2)}(k)^{-1}.
    \end{equation*}
\end{itemize}
\end{dbar}
\begin{proof}
    From RH problem \ref{rhp4}-\ref{rhp5}, the analyticity can be proved immediately. As $M^{(2)}(k)$ and $M^{(2)}_{RHP}$ share the same jump matrix, which brings up to  \begin{align*}
        M_{-}^{(3)}(k)^{-1} M_{+}^{(3)}(k) &= M_{RHP-}^{(2)} \left(M_{-}^{(2)}\right)^{-1} M_{+}^{(2)} \left(M_{RHP+}^{(2)}\right)^{-1}= I.
    \end{align*}
    To prove the continuity of $M^{(3)}(k)$, we only consider $z_n\in\mathcal{Z}\cup\overline{\mathcal{Z}}$. As $z_n$ is the pole of the first order for $M^{(2)}$ and $ M_{RHP}^{(2)}$, by the residue conditions we can obtain their Laurent expansions in $z_n$:
    \begin{align*}
         M^{(2)}(k) &= \mathcal{M}(z_n) \left[ \frac{\tau^{\Delta^-}_n}{k - z_n} + I \right] + \mathcal{O}(k - z_n),\\
         M^{(2)}_{RHP}(k) &= \mathcal{M}^{\prime}(z_n) \left[ \frac{\tau^{\Delta^-}_n}{k -z_n} + I \right] + \mathcal{O}(k - z_n),
    \end{align*}
where $ \mathcal{M}(z_n)$ and $\mathcal{M}^{\prime}(z_n)$ are constant matrices, $\tau_n^{\Delta^-}$ is nilpotent we define in \eqref{eq34}, here we suppose $z_n\in\mathcal{Z}$. Then
\begin{align*}
    M^{(3)}(k) &= \left\{ \mathcal{M}(z_n) \left[ \frac{\tau^{\Delta^-}_n}{k - z_n} + I \right] \right\}
\left\{ \left[ \frac{-\tau^{\Delta^-}_n}{k - z_n} + I \right] \sigma_2 \mathcal{M}^{\prime}(z_n)^\mathrm{T} \sigma_2 \right\}
+ \mathcal{O}(k - z_n),\\
&= \mathcal{O}(1).
\end{align*}
This implies that $z_n$ is removable singularities of $M^{(3)}(k)$.
\end{proof}
Then we prove the existence and asymptotics for $M^{(3)}$ sequentially.

The solution of $\bar{\partial}$-Problem \ref{dbar1} can be solved by the following integral equation
\begin{equation}\label{eq35}
    M^{(3)}(k) = I - \frac{1}{\pi} \iint_\mathbb{C} \frac{M^{(3)}(s) W^{(3)}(s)}{s - k} \, \mathrm{d}A(s),
\end{equation}
where $A(s)$ is the Lebesgue measure on $\mathbb{C}$. Denote $S$ as the Cauchy-Green integral operator
\begin{equation}\label{eq36}
    S\left[f\right](k) = -\frac{1}{\pi} \iint_\mathbb{C} \frac{f(s) W^{(3)}(s)}{s - k} \, \mathrm{d}A(s),
\end{equation}
then \eqref{eq35} can be written as the following equation
\begin{equation}
    (1 - S)M^{(3)}(k) = I.
\end{equation}
To prove the existence of the operator at large time, we present the following proposition.

\begin{Proposition}\label{p17}
    Consider the operator $S$ defined by \eqref{eq36},  we can obtain $S : L^\infty(\mathbb{C}) \to L^\infty(\mathbb{C}) \cap C^0(\mathbb{C})$ and
\begin{equation}
    \|S\|_{L^\infty(\mathbb{C}) \to L^\infty(\mathbb{C})} \lesssim t^{-\frac{1}{4}}.
\end{equation}

\end{Proposition}

\begin{proof}
    For any $f \in L^\infty$, we have
\begin{equation*}
    \|Sf\|_{L^\infty} \leq \|f\|_{L^\infty} \frac{1}{\pi} \iint_\mathbb{C} \frac{|W^{(3)}(s)|}{|s - k|} \, \mathrm{d}A(s).
\end{equation*}
Recalling our definition $W^{(3)} =  M_{RHP}^{(2)}(k) \bar{\partial} R^{(2)}(k) M_{RHP}^{(2)}(k)^{-1}.$ First we know that $W^{(3)}(k) \equiv 0$ for $k \in \mathbb{C} \setminus \bar{\Omega}$. Besides, we only take into account the matrix-valued functions have support in sector $\bar{\Omega}$. Moreover, we know that $M_{RHP}^{(2)}(k)$ and $M_{RHP}^{(2)}(k)^{-1}$ are all bounded on $\bar{\Omega}$, which means
\begin{equation}\label{eq37}
    \iint_{\Omega_\ell}\frac{|W^{(3)}(s)|}{|s-k|}\mathrm{d}A(s)\lesssim\iint_{\Omega_\ell}\frac{|\bar{\partial}R_\ell(s)e^{\pm2it\theta}|}{|s-k|}\mathrm{d}A(s),\quad\ell=1,2,
\end{equation}
where the superscript takes $+$ for $\ell=1$, takes $-$ for $\ell=2$.
To shorten the length of this paper, we only consider the region $\Omega_1\cap\left\{k\in\mathbb{C}:\mathrm{Re}k>k_1\right\}:=\widehat\Omega_1$ of case \uppercase\expandafter{\romannumeral1}. Together with Proposition \ref{p11}, we can break right side of the equation \eqref{eq37} into two parts:
\begin{equation*}
\iint_{\widehat\Omega_1} \frac{|\bar{\partial} R_{1}(s)| e^{-2t \mathrm{Im} \theta}}{|k - s|} \, \mathrm{d}A(s)
\lesssim L_1 + L_2,
\end{equation*}
with
\begin{align*}
    L_1 = \iint_{\widehat\Omega_1} \frac{|r'(\mathrm{Re}s)| e^{-2t \mathrm{Im} \theta}}{|k - s|} \, \mathrm{d}A(s),\quad L_2 = \iint_{\widehat\Omega_1} \frac{|s - k_1|^{-\frac{1}{2}} e^{-2t \mathrm{Im} \theta}}{|k - s|} \, \mathrm{d}A(s).
\end{align*}
Denote $k=x+yi,s=k_1+u+iv$ with $x,y,u,v\in\mathbb{R}$, then Lemma \ref{p10} implies that
\begin{align*}
L_1&\lesssim\int_0^{+\infty}\int_v^{+\infty}\frac{|r'(\mathrm{Re}s)| e^{- t uv}}{|k-s|}\mathrm{d}u\mathrm{d}v\leqslant\int_0^{+\infty}e^{- t v^2}\mathrm{d}v\int_v^{+\infty}\frac{|r^\prime(k_1+u)| }{|k-s|}\mathrm{d}u\\
    &\leqslant \int_0^{+\infty}e^{- t v^2}\|r^\prime\|_{L^2}\|\frac{1}{|k-s|}\|_{L^2(v,+\infty)}\mathrm{d}v\lesssim  \int_0^{+\infty}e^{- t v^2}\|\frac{1}{|k-s|}\|_{L^2(v,+\infty)}\mathrm{d}v.
\end{align*}
For further calculation, we introduce the following estimate for $q>1$,
\begin{equation}\label{eq41}
    \begin{aligned}
    \|\frac{1}{|k-s|}\|_{L^q(v,+\infty)}&=\left(\int_v^{+\infty}\frac{1}{|k-s|^q}\mathrm{d}u\right)^{\frac{1}{q}}\\
    &\leqslant|v-y|^{\frac{1}{q}-1}\int_0^{+\infty}\left[\left(\frac{u+k_1-x}{v-y}\right)^2+1\right]^{-\frac{q}{2}}\mathrm{d}\left(\frac{u+k_1-x}{v-y}\right)\\
    &\lesssim|v-y|^{\frac{1}{q}-1}.
\end{aligned}
\end{equation}

Then back to the calculation of $L_1$, we have
\begin{equation}
    L_1\lesssim\int_0^{+\infty}\frac{e^{-t v^2}}{\sqrt{|v-y|}}\mathrm{d}v=L_1^{(1)}+L_1^{(2)},
\end{equation}
where
\begin{equation*}
     L_1^{(1)}=\int_0^{y}\frac{e^{-t v^2}}{\sqrt{y-v}}\mathrm{d}v,\quad L_1^{(2)}=\int_y^{+\infty}\frac{e^{-t v^2}}{\sqrt{v-y}}\mathrm{d}v.
\end{equation*}
Therefore,
\begin{align*}
     L_1^{(1)}\lesssim t^{-\frac{1}{4}}\int_0^1\frac{\mathrm{d}m}{\sqrt{m(1-m)}}\lesssim t^{-\frac{1}{4}},\quad L_1^{(2)}\lesssim \int_0^{+\infty}\frac{e^{-tm^2}}{\sqrt{m}}\mathrm{d}m\lesssim t^{-\frac{1}{4}},
\end{align*}
which implies $L_1\lesssim t^{-\frac{1}{4}}$.\\
As for $L_2$, by H\"older inequality with $\frac{1}{p} + \frac{1}{q} = 1,p>2$,
\begin{align}\label{eq38}
    L_2&\lesssim \int_0^{+\infty}e^{-t v^2}\|\frac{1}{\sqrt{|s-k_1|}}\|_{L^p(\mathbb{R}_+)}\|\frac{1}{k-s}\|_{L^q(\mathbb{R}_+)}\mathrm{d}v,
\end{align}
where
\begin{align*}
    \|\frac{1}{\sqrt{|s-k_1|}}\|_{L^p(\mathbb{R}_+)}&=\left(\int_0^{+\infty}\left(u^2+v^2\right)^{-\frac{p}{4}}\mathrm{d}u\right)^{\frac{1}{p}}\\
    &=v^{\frac{1}{p}-\frac{1}{2}}\left[\int_0^{+\infty}\left(1+m^2\right)^{-\frac{p}{4}}\mathrm{d}m\right]^{\frac{1}{p}}\lesssim v^{\frac{1}{p}-\frac{1}{2}}.
\end{align*}
Taking this estimate into equation \eqref{eq38}, we obtain
\begin{equation*}
    L_2\lesssim\int_0^{+\infty}e^{-t v^2}v^{\frac{1}{p}-\frac{1}{2}}|v-y|^{\frac{1}{q}-1}\mathrm{d}v=L_2^{(1)}+L_2^{(2)},
\end{equation*}
where
\begin{equation*}
    L_2^{(1)}=\int_0^{y}e^{-t v^2}v^{\frac{1}{p}-\frac{1}{2}}(y-v)^{\frac{1}{q}-1}\mathrm{d}v,\quad L_2^{(2)}=\int_y^{+\infty}e^{-t v^2}v^{\frac{1}{p}-\frac{1}{2}}(v-y)^{\frac{1}{q}-1}\mathrm{d}v.
\end{equation*}
Let $v=my$,  $L_2^{(1)}$ becomes
\begin{align*}
    L_2^{(1)}=\int_0^1e^{- ty^2m^2}y^{\frac{1}{2}}m^{\frac{1}{p}-\frac{1}{2}}(1-m)^{\frac{1}{q}-1}\mathrm{d}m\lesssim t^{-\frac{1}{4}}\int_0^1m^{\frac{1}{p}-1}(1-m)^{\frac{1}{q}-1}\mathrm{d}m\overset{{p,q>2}}\lesssim t^{-\frac{1}{4}}.
\end{align*}
Let $n=v-y$,  $L_2^{(2)}$ becomes
\begin{align*}
    L_2^{(2)}=\int_0^{+\infty}e^{-t(y+n)^2}(y+n)^{\frac{1}{p}-\frac{1}{2}}n^{\frac{1}{q}-1}\mathrm{d}n\leqslant \int_0^{+\infty}\frac{e^{- tn^2}}{\sqrt{n}}\mathrm{d}n\lesssim t^{-\frac{1}{4}}.
\end{align*}
From the above calculation, we obtain $L_2\lesssim t^{-\frac{1}{4}}$. Summarizing the results we give above,  $\|S\|_{L^\infty(\mathbb{C}) \to L^\infty(\mathbb{C})} \lesssim t^{-\frac{1}{4}}$ as $t\to\infty$.
\end{proof}

Consider the asymptotic expansion of $M^{(3)}( y, t;k)$ at $k = 0$
\begin{equation*}
    M^{(3)}(y, t;k) = I + M^{(3)}_0(y, t) + M^{(3)}_1(y, t) k + \mathcal{O}(k^2), \quad k \to 0,
\end{equation*}

where
\begin{align}
    M^{(3)}_0(y, t)& = \frac{1}{\pi} \iint_\mathbb{C} \frac{M^{(3)}(s) W^{(3)}(s)}{s} \, \mathrm{d}A(s),\label{eq39}\\
    M^{(3)}_1(y, t)& = \frac{1}{\pi} \iint_\mathbb{C} \frac{M^{(3)}(s) W^{(3)}(s)}{s^2} \, \mathrm{d}A(s).\label{eq40}
\end{align}
To reconstruct the solution $u(y, t)$ of the WKI-SP equation \eqref{wkisp-1}, we need the asymptotic behavior of $M^{(3)}_0(y, t)$ and $M^{(3)}_1(y, t)$ as $t\to\infty$.
\begin{Proposition}\label{p18}
As $k \rightarrow 0$, $M^{(3)}(y,t;k)$ has the asymptotic expansion:
\begin{align}
|M^{(3)}_0(y,t)|\lesssim t^{-\frac{3}{4}},\quad |M^{(3)}_1(y,t)|\lesssim t^{-\frac{3}{4}},\quad as\ t\rightarrow \infty.
\end{align}
\end{Proposition}
\begin{proof}
    Since the integration region passes through the origin, which is a singularity for integral \eqref{eq39} and \eqref{eq40}, we need to consider the estimate near the origin and away from the origin respectively. Here we only consider case \uppercase\expandafter{\romannumeral1} as an example.

For  $s$ away from the origin, we take $\Omega_1\cap\left\{k\in\mathbb{C}:\mathrm{Re}k>k_1\right\}:=\widehat\Omega_1$. As $|s|>|k_1|$ for $s\in\widehat\Omega_1$, then
\begin{equation*}
|M^{(3)}_0(y,t)|_{\widehat\Omega_1}\lesssim\iint_{\widehat\Omega_1}\left|M^{(3)}(s) W^{(3)}(s) \right|\, \mathrm{d}A(s)=\iint_{\widehat\Omega_1} |\bar{\partial} R_{1}(s)| e^{-2t \mathrm{Im} \theta} \mathrm{d}A(s)\lesssim Q_1+Q_2,
\end{equation*}
where
\begin{equation}
     Q_1 = \iint_{\widehat\Omega_1} |r'(\mathrm{Re}s)| e^{-2t \mathrm{Im} \theta} \, \mathrm{d}A(s),\quad
    Q_2 = \iint_{\widehat\Omega_1} |s - k_1|^{-\frac{1}{2}} e^{-2t \mathrm{Im} \theta} \, \mathrm{d}A(s).
\end{equation}
Take the notations in Proposition \ref{p17}, we can obtain
\begin{align*}
    Q_1&\lesssim\int_0^{+\infty}\int_v^{+\infty}|r'(\mathrm{Re}s)|e^{-tuv}\mathrm{d}u\mathrm{d}v\\
    &\leqslant \int_0^{+\infty}\|r'(\mathrm{Re}s)\|_{L^2}\left(\int_v^{+\infty} e^{-tuv}\mathrm{d}u\right)^{\frac{1}{2}}\mathrm{d}v\lesssim t^{-\frac{1}{2}}\int_0^{+\infty}\frac{e^{-tv^2}}{\sqrt{v}}\mathrm{d}v\lesssim t^{-\frac{3}{4}}.
\end{align*}
By H\"older equality satisfying $\frac{1}{p}+\frac{1}{q}=1$ with $2<p<4$, we can estimate $Q_2$ as follows
\begin{align*}
    Q_2&\lesssim\int_0^{+\infty}\||s-k_1|^{-\frac{1}{2}}\|_{L^p(\mathbb{R}_+)}\left(\int_v^{+\infty}e^{-tuv}\mathrm{d}u\right)^{\frac{1}{q}}\mathrm{d}v\\
    &\lesssim t^{-\frac{1}{q}}\int_0^{+\infty}v^{\frac{2}{p}-\frac{3}{2}}e^{-tv^2}\mathrm{d}v\lesssim t^{\frac{2}{p}-\frac{7}{4}}\lesssim t^{-\frac{3}{4}},
\end{align*}
 here the constraints on $p$ is used to ensure the convergence of the second improper integral. For the asymptotics of $M^{(3)}_1(y, t)$ in the same region, we can do the same estimate as above.

    For  $s$ near the origin, we take $\Omega_{1}\cap\{k:k_3<\mathrm{Re}k<0\}:=\widetilde\Omega_1$ as an example. First we divide $\widetilde\Omega_1$ into two parts
    \begin{equation*}
        B(0)=\widetilde\Omega_1\cap\{k:|k|<\epsilon<\frac{|k_3|}{4}\},\quad B_{c}=\widetilde\Omega_1\setminus B(0).
    \end{equation*}
For $k\in B_{c}$, the calculation is similar with $k\in\widetilde\Omega_1$, which implies
\begin{equation*}
    |M^{(3)}_n(y,t)|_{B_{c}}\lesssim t^{-\frac{3}{4}},\quad for\ n=0,1.
\end{equation*}
For $k\in B(0)$, consider the estimate \eqref{eq9} we make for $k$ near the origin in Proposition \ref{p11} and the estimate we make for $\mathrm{Im}\theta$ in Corollary \ref{c1},
\begin{equation*}
    |\bar{\partial}R_1|\lesssim |k|,\quad for \ k\in B(0) ,
\end{equation*}
then we can simply get the following estimates
\begin{align*}
     |M^{(3)}_0(y, t) |_{B(0)}&= \frac{1}{\pi} \iint_{B(0)} \frac{|M^{(3)}(s) W^{(3)}(s)|}{|s|} \, \mathrm{d}A(s)\lesssim \iint_{B(0)}\frac{|\bar{\partial}R_1|e^{-tv}}{|s|}\mathrm{d}A(s)\\
     &\lesssim \iint_{B(0)}e^{-tv}\mathrm{d}A(s)\lesssim t^{-1}.
\end{align*}
As for $|s|<\frac{|k_3|}{4}$, taking $p>2,\ k=0$ in \eqref{eq41}, we find
\begin{align*}
    |M^{(3)}_1(y, t) |_{B(0)}&\lesssim\iint_{B(0)}\frac{e^{-t\mathrm{Im}\theta}}{|s|}\mathrm{d}A(s) \lesssim
\int_0^{\frac{|k_3|}{4}}\| s^{-1}\|_{L^p}\| e^{-tv}\|_{L^q}\mathrm{d}v\\
&\lesssim t^{-\frac{1}{q}}\int_0^{\frac{|k_3|}{4}}v^{\frac{1}{p}-1}e^{-tv}\mathrm{d}v
\lesssim t^{-1}.
\end{align*}

Thus, summarizing the estimates above, we conclude the proof of this proposition.
\end{proof}

\subsection{Proof of Theorem 1-\uppercase\expandafter{\romannumeral1}}\label{subs3.5}
Finally, we construct the long-time asymptotic approximation for the solution   of the WKI-SP equation \eqref{wkisp-1}. Inverting the transformations \eqref{eq2},\eqref{eq17},\eqref{eq28},\eqref{eq33}, we have
\begin{equation}
    M(k)=M^{(3)}(k)E(k)M^{(out)}(k)R^{(2)}(k)^{-1}T(k)^{-\sigma_3}.
\end{equation}
We take $k\to 0$ out of $\Omega$ so that $R^{(2)}(k)=I$. Then by the results of Proposition \ref{p14},\ref{p16},\ref{p18}, we obtain the follow asymptotic expansion of ${M}(k)$ as $k\to0$:
\begin{align}
    M(k)=\left[I+\mathcal{O}(t^{-\frac{3}{4}})+\mathcal{O}(t^{-\frac{3}{4}})k\right]\left[E_0+E_1k\right]M^{(out)}(k)\left(T_0+iT_0T_1k\right)^{-\sigma_3}+\mathcal{O}(k^2).
\end{align}
By the reconstruction formula
\begin{equation*}
    u(x, t) = u(y(x, t), t) = -i \lim_{k \to 0} k^{-1} \left[M^{-1}(0) M(k)\right]_{12},
\end{equation*}
further using Corollary \ref{c2}, we then  obtain the proof of Theorem \ref{thm1}-\uppercase\expandafter{\romannumeral1}.

\section{Long-time asymptotics in region without saddle point}\label{sec5}
In this section, we consider case \uppercase\expandafter{\romannumeral3} ($\alpha>0,\beta>0,\xi>-2\sqrt{3\alpha\beta} $). Also, we start from the basic RH problem \ref{rhp1} with the jump matrix
\begin{equation*}
V(k)=\begin{array}{ll}
\left(\begin{array}{cc}1&0\\
\bar{r}e^{-2it\theta}&1\end{array}\right) \left(\begin{array}{cc}1&re^{2it\theta}\\ 0&1\end{array}\right),&k \in \mathbb{R}.
\end{array}
\end{equation*}

We define function $T(k)$ as
\begin{equation*}
    T(k)=\prod\limits_{n\in\Delta^-}\frac{k-\bar{z}_n}{k-z_n},
\end{equation*}
which has the following properties.
\begin{Proposition}\label{p19}
    The function $T(k)$ we defined above has the following properties:
    \begin{enumerate}[label=\textnormal{(\arabic*)}]
      \item $T(k)$ is meromorphic in $\mathbb{C} $. And for each $n \in \Delta^-, T(k)$ has a simple pole at $z_n$ and a simple zero at $\bar{z}_n$;
       \item For $k\in \mathbb{C},\;T(k)\overline{T(\bar{k})}=1$;
\item As $|k|\to+\infty,\;|\mathrm{arg}k|\leqslant c<\pi$,
$$T(k)=1+\frac{i}{k}\left(2\sum_{n\in\Delta^-}\mathrm{Im}z_n\right)+\mathcal{O}(k^{-2});$$

\item $T(k)$ is continuous at $k=0$, and
\begin{equation*}
    T(k)=T_0(1+T_1k)+\mathcal{O}(k^2),
\end{equation*}
where $$T_0=\prod\limits_{n\in\Delta^-}\frac{\bar{z}_n}{z_n}=\mathrm{exp}\left[-2i\sum_{n\in\Delta^-}\mathrm{arg}(z_n)\right],\quad T_1=-\sum_{n\in\Delta^-}\frac{2\mathrm{Im}(z_n)}{|z_n|^2}.$$

    \end{enumerate}
  \end{Proposition}

Make transformation
\begin{equation}\label{eq71}
    M^{(1)}(y,t;k)=M(y,t;k)T(k)^{\sigma_3},
\end{equation}

where $M^{(1)}(y,t;k)$ is the solution to the following RH problem.

\begin{RHP}
    Find a $2\times 2$ matrix-valued function $M^{(1)}(k)$ with the following properties:
\begin{itemize}
  \item Analyticity: $M^{(1)}(k)$ is analytical in $\mathbb{C}\setminus\mathbb{R};$

  \item Jump condition: $M^{(1)}(k)$ has continuous boundary values $M^{(1)}_{\pm}(k)$ on $\mathbb{R}$ and
 \begin{equation*}
M^{(1)}_+(k)=M^{(1)}_-(k)V^{(1)}(k),
\end{equation*}
where
\begin{equation}
V^{(1)}(k)=\begin{array}{ll}
\left(\begin{array}{cc}1&0\\ \bar{r}(k)T^2(k)e^{-2it\theta}&1\end{array}\right)
\left(\begin{array}{cc}1&r(k)T^{-2}(k)e^{2it\theta}\\ 0&1\end{array}\right), &k\in \mathbb{R};
\end{array}
\end{equation}

  \item Asymptotic behavior:
  $\ M^{(1)}(k)=I+\mathcal{O}(k^{-1}),\ as \ k\to \infty;$
\item Residue condition: $M^{(1)}(k)$ has simple poles at each $z_n\in\mathcal{Z}\cup\bar{\mathcal{Z}}$, which  has the same residue condition in \eqref{res1}-\eqref{res4}.
\end{itemize}
\end{RHP}

\subsection{Deformation of the RH problem and hybrid \texorpdfstring{$\bar{\partial}$}{\bar{\partial}}-RH problem}\label{sec5.1}
We open the jump line $\mathbb{R}$ at $\pm1$ respectively with small enough angle to form two open regions $\Omega_1$ and $\Omega_2$, enclosed by $\Sigma_1$ and $\Sigma_2$ with $\mathbb{R}$ respectively, which is depicted in Figure \ref{Fig20}. The reason why we choose $\pm1$ is to make sure the extension function we define below hold the property of $|\bar\partial R^{(2)}(k)|\lesssim|k|$ near $k=0.$

\begin{figure}
    \centering
 \begin{tikzpicture}[scale=1.75]
\draw[-latex,dotted](-3,0)--(3,0)node[right]{ \textcolor{black}{Re$k$}};

\draw [brown,thick] (-2.23,-0.6) to [out=0,in=-125] (-1.18,0);
\draw [brown,thick] (-1.18,0) to [out=-55,in=-125] (0,0)  ;

\draw [brown,thick] (0,0) to [out=-55,in=-125] (1.18,0);
\draw [brown,thick] (1.18,0)  to  [out=-55,in=180] (2.23,-0.6);

\draw [teal,thick](-2.23,0.6) to [out=0,in=125] (-1.18,0);
\draw [teal,thick] (-1.18,0)to  [out=55,in=125] (0,0);

\draw [teal,thick](0,0) to [out=55,in=125](1.18,0);
\draw [teal,thick]  (1.18,0)   to [out=55,in=180] (2.23,0.6);

\draw[-latex,teal,thick](-1.38,0.24)--(-1.33,0.19);
\draw[-latex,brown,thick](-1.38,-0.24)--(-1.33,-0.19);

\draw[-latex,teal,thick](1.33,0.19)--(1.38,0.24);
\draw[-latex,brown,thick](1.33,-0.19)--(1.38,-0.24);

\draw[-latex,brown,thick](-0.95,-0.19)--(-0.86,-0.24);
\draw[-latex,teal,thick](-0.95,0.19)--(-0.86,0.24);

\draw[-latex,brown,thick](0.86,-0.24)--(0.95,-0.19);
\draw[-latex,teal,thick](0.86,0.24)--(0.95,0.19);

\node  at (2.33,0.6) {$\Sigma_1$};
\node  at (2.33,-0.6) {$\Sigma_2$};
 \coordinate (a) at (1.77,1.2);
\fill[purple] (a) circle (1pt);

\coordinate (a) at (1.77,-1.2);
\fill[purple] (a) circle (1pt);
\coordinate (a) at (-1.77,1.2);
\fill[purple] (a) circle (1pt);

\coordinate (a) at (-1.77,-1.2);
\fill[purple] (a) circle (1pt);
\coordinate (a) at (0.71,0.6);
\fill[purple] (a) circle (1pt);
\coordinate (a) at (-0.71,0.6);
\fill[purple] (a) circle (1pt);
\coordinate (a) at (0.71,-0.6);
\fill[purple] (a) circle (1pt);
\coordinate (a) at (-0.71,-0.6);
\fill[purple] (a) circle (1pt);
\coordinate (A) at (1.18,0);
\fill (A) circle (1pt) node[below,scale=0.9]{$1$};

\coordinate (C)  at (-1.18,0);
\fill (C) circle (1pt) node[below,scale=0.9]{$-1$};

\coordinate (I) at (0,0);
		\fill[red] (I) circle (1pt) node[below,scale=0.9] {$0$};
\end{tikzpicture}
    \caption{Opening the jump line $\mathbb{R}$ at  $\pm 1$ with sufficient small angle $\phi$.
 The opened contours   $\Sigma_1$ ($\textcolor{teal}{\bullet} $) and   $\Sigma_2$ ($\textcolor{brown}{\bullet} $)  decay  in   blue region  and   white  region
    in Figure  \ref{figtheta-1}, respectively. The discrete spectrum
    on $\mathbb{C}$  denoted by ($\textcolor{purple}{\bullet} $).   }
     \label{Fig20}
\end{figure}
\FloatBarrier

\begin{lemma}\label{l9}
    In the region $\Omega$, the imaginary part of $\theta(k)$ satisfies the following estimates respectively,
    \begin{align}
         &\mathrm{Im}\theta(k)\gtrsim\mathrm{Im}k,\quad k\in\Omega_{1},\\
        &\mathrm{Im}\theta(k)\lesssim \mathrm{Im}k,\quad k\in\Omega_{2}.
    \end{align}
\end{lemma}

We define the extension functions by the following proposition.

\begin{Proposition}\label{p30}
    There exist  the functions $R_{\ell}(k)$: $\bar{\Omega}_{\ell}\to \mathbb{C}$, $\ell=1,2$  with the boundary values
	\begin{align} R_{1}(k)=\left\{\begin{array}{ll}
	r(k)T(k)^{-2}, &\hspace{0.4cm}k\in \mathbb{R},\\[4pt]
 r(\pm1)T(\pm1)^{-2},  &\hspace{0.4cm}k\in \Sigma_{1},\\
	\end{array}\right. \quad\quad
	R_{2}(k)=\left\{\begin{array}{ll} \bar{r}(k)T(k)^{2},
 &k\in  \mathbb{R},\\[4pt] \bar{r}(\pm1)T(\pm1)^{2}, &k\in \Sigma_{2}.
	\end{array} \right.
	\end{align}	
The functions  $R_{\ell}(k), \ell=1,2$  admit the following estimates:
	\begin{align*}
    &|R_{\ell}(k)|\lesssim 1+\left[1+\mathrm{Re}^2(k)\right]^{-\frac{1}{2}},\quad\;for\ k \in\Omega,\\
    &|\bar{\partial}R_{\ell}(k)|\lesssim\chi(\mathrm{Re}k)+|r^\prime(\mathrm{Re}k)|+|k\pm1|^{-\frac{1}{2}}, \quad   \;for\  k\in \Omega\cap\{\mathrm{Re}k<1\},\\
	&|\bar{\partial}R_{\ell}(k)|\lesssim|r'(\mathrm{Re}k)|+|k\pm1|^{-\frac{1}{2}}, \quad  \;for\   k\in \Omega\cap\{\mathrm{Re}k>1\},\\
&|\bar{\partial}R_{\ell}(k)|\lesssim|k| \quad as\ k\to0,\;for\ k\in \Omega,\\
&\bar{\partial}R_{\ell}(k)=0,\quad \;for\ k\in\mathbb{C}\setminus\Omega,\nonumber
	\end{align*}
where $\chi \in C^{\infty}_0(\mathbb{R},[0,1])$ is a fixed cut-off function with support near 0.
\end{Proposition}
\begin{proof}
    The proof for this proposition is similar with Proposition \ref{p11}.
\end{proof}

Define  a new  function
\begin{equation}
R^{(2)}(k)=\left\{\begin{array}{lll}
\left(\begin{array}{cc}
1 & -R_{1}(k)e^{2it\theta}\\
0 & 1
\end{array}\right), & k\in \Omega_{1};\\
\\
\left(\begin{array}{cc}
1 & 0\\
R_{2}(k)e^{-2it\theta} & 1
\end{array}\right),  &k\in \Omega_{2};\\
\\
I,  &elsewhere;\\
\end{array}\right.
\end{equation}
where   the functions $R_{\ell}(k)$, $\ell=1,2$ are given by  Proposition \ref{p30}.

Make  a transformation
\begin{equation}\label{eq72}
M^{(2)}(k):=M^{(2)}(y,t;k)=M^{(1)}(k)R^{(2)}(k),
\end{equation}
then $M^{(2)}(k)$ is a hybrid RH problem:

\begin{RHP}
    Find a $2\times 2$ matrix-valued function $M^{(2)}(k)$ with the following properties:
\begin{itemize}
  \item Analyticity: $M^{(2)}(k)$ is continuous in $\mathbb{C}$, sectionally continuous for first-order partial derivatives in $\mathbb{C}\setminus(\Sigma^{(2)}\cup\mathcal{Z}\cup\bar{\mathcal{Z})}$  , where $\Sigma^{(2)}=\Sigma_1\cup\Sigma_2$  ;
  \item Jump condition: $M^{(2)}(k)$ has continuous boundary values $M^{(2)}_{\pm}(k)$ on $\Sigma^{(2)}$ and
 \begin{equation*}
M^{(2)}_+(k)=M^{(2)}_-(k)V^{(2)}(k),
\end{equation*}
where
\begin{equation*}
V^{(2)}(k)=\left\{\begin{array}{lll}
\left(\begin{array}{cc}
1 & R_{1}(k)e^{2it\theta}\\
0 & 1
\end{array}\right), & k\in \Sigma_{1};\\
\\
\left(\begin{array}{cc}
1 & 0\\
R_{2}(k)e^{-2it\theta} & 1
\end{array}\right),  &k\in \Sigma_{2};\\

\end{array}\right.
\end{equation*}

  \item Asymptotic behavior:
  $\ M^{(2)}(k)=I+\mathcal{O}(k^{-1}),\ as \ k\to \infty;$
  \item $\bar{\partial}$-Derivative: For $k\in\mathbb{C}$, we have the $\bar{\partial}$-Derivative equation
\begin{equation}
    \bar{\partial}M^{(2)}(k)=M^{(2)}(k)\bar{\partial}R^{(2)}(k),
\end{equation}
where
  \begin{equation}
\bar{\partial}R^{(2)}(k)=\left\{\begin{array}{lll}
\left(\begin{array}{cc}
0 & -\bar{\partial}R_{1}(k)e^{2it\theta}\\
0 & 0
\end{array}\right), & k\in \Omega_{1};\\
\\
\left(\begin{array}{cc}
0 & 0\\
\bar{\partial}R_{2}(k)e^{-2it\theta} & 0
\end{array}\right),  &k\in \Omega_{2};\\
\\
0,  &elsewhere;\\
\end{array}\right.
\end{equation}
\item Residue condition: $M^{(2)}(k)$ has simple poles at each $z_n\in\mathcal{Z}\cup\bar{\mathcal{Z}}$, which  has the same residue condition with $M^{(1)}(k)$ in \eqref{res1}-\eqref{res4}.

\end{itemize}
\end{RHP}
To solve $M^{(1)}(k)$,  we decompose it into $M^{(R)}(k):= M^{(R)}(y,t;k)$ with $\bar\partial M^{(R)}=0$ and a pure $\bar\partial$-problem $M^{(2)}(k)$.

\subsection{Analysis on a pure RH problem}\label{sec5.2}
First we give a RH problem for $M^{(R)}(y,t;k)$:

\begin{RHP}\label{rhp29}
    Find a $2\times 2$ matrix-valued function $M^{(R)}(k)$ with the following properties:
\begin{itemize}
  \item Analyticity: $M^{(R)}(k)$ is analytic in $\mathbb{C}\setminus(\Sigma^{(2)}\cup\mathcal{Z}\cup\bar{\mathcal{Z}})$ ;
  \item Jump condition: $M^{(R)}(k)$ has continuous boundary values $M^{(R)}_{\pm}(k)$ on $\Sigma^{(2)}$ and
 \begin{equation*}
M^{(R)}_{+}(k)=M^{(R)}_{-}(k)V^{(2)}(k);
\end{equation*}
\item Symmetry:
$M^{(R)}(k)=\sigma_2 \overline{ M^{(R)}(\bar{k})}\sigma_2=\sigma_2 M^{(R)}(-k)\sigma_2;$

  \item Asymptotic behavior:
  $\ M^{(R)}(k)=I+\mathcal{O}(k^{-1}),\ as \ k\to \infty;$

\item Residue condition: $M^{(R)}(k)$ has simple poles at each $z_n\in\mathcal{Z}\cup\bar{\mathcal{Z}}$ with residue condition \eqref{res1}-\eqref{res4}.
\end{itemize}
\end{RHP}

As the RH problem \ref{rhp29} contains spectrum points and jump line, we need to consider their contributions to the solution respectively. For this purpose, we define
\begin{equation}\label{eq73}
    M^{(R)}(k)=M^{(J)}(k)M^{(out)}(k),
\end{equation}
where $M^{(out)}(k)$ denotes the part for spectrum points and $M^{(J)}(k)$ contains the contribution from jump line, which is a small normed RH problem.

\begin{RHP}\label{rhp30}
 Find a matrix-valued function $M^{(out)}(k) = M^{(out)}(y, t;k)$ with the following properties:

\begin{itemize}
    \item Analyticity: $M^{(out)}(k)$ is analytical in $\mathbb{C} \setminus (\mathcal{Z} \cup \overline{\mathcal{Z}})$;
    \item Symmetry: $M^{(out)}(\overline{k}) = \overline{M^{(out)}(-k)} = \sigma_2 \overline{M^{(out)}(k) }\sigma_2$;
    \item Asymptotic behaviors: $M^{(out)}(k) \sim I + \mathcal{O}(k^{-1}), \quad k \to \infty; $
    \item Residue conditions: $M^{(out)}(k)$ has simple poles at each point in $\mathcal{Z} \cup \overline{\mathcal{Z}}$ satisfying the same residue relations  with $M^{(R)}(k)$.
\end{itemize}
\end{RHP}

Similar with Proposition \ref{p14}, we can solve $M^{(out)}$ with the help of the reflection-less version.

\begin{Proposition}\label{p31}
    There exists a unique solution for the RH Problem \ref{rhp30}. Moreover, the $N$-soliton solution of WKI-SP encoded by RH problem \ref{rhp30} can be reconstructed by
\begin{align*}
    &u_{sol}(x, t;\sigma_d^{(out)})=u_{sol}(x, t;\sigma_d)=u_{sol}(y(x,t), t;\sigma_d),\\
    &y(x,t)=x-c_+(x,t;\sigma_d),
\end{align*}
where $\sigma_d^{(out)}$ is the given scattering data for $M^{(out)}(k)$, and $\sigma_d$ is the given scattering data for $M^{(out)}(k)$ under the condition that $r(k)=0$.
\end{Proposition}

By the define of $M^{(J)}(k)$ in \eqref{eq73}, we obtain

\begin{RHP}
    Find a $2 \times 2$ matrix-valued function $M^{(J)}(k)$ such that  \begin{itemize}
        \item Analyticity: $M^{(J)}(k)$ is analytical in $\mathbb{C} \setminus \Sigma^{(2)}$;
\item Jump condition: $M^{(J)}(k)$ takes continuous boundary values $M^{(J)}_\pm(k)$ on $\Sigma^{(2)}$ and
\begin{equation*}
M^{(J)}_+(k) = M^{(J)}_-(k) V^{(J)}(k),
\end{equation*}
where
\begin{equation*}
V^{(J)}(k) =
M^{(out)}(k)V^{(2)}(k)M^{(out)}(k)^{-1};
\end{equation*}
\item Asymptotic behavior: $M^{(J)}(k) = I + \mathcal{O}(k^{-1}), \quad k \to \infty .$
  \end{itemize}
\end{RHP}
To solve the RH problem for $M^{(J)}(k)$, we need the following estimate on $V^{(2)}(k)$.
\begin{Proposition}\label{p32}
    As $t\to+\infty$, we have
    \begin{equation*}
        \|V^{(2)}(k)-I\|_{L^\infty\left(\Sigma^{(2)}\right)}=\mathcal{O}(t^{-1}).
    \end{equation*}
\end{Proposition}
\begin{proof}
    We take $k\in\Sigma_1$ as an example:
    \begin{align*}
         \|V^{(2)}(k)-I\|_{L^\infty\left(\Sigma^{(2)}\right)}=\|r(1)T(1)^{-2}e^{2it\theta(k)}\|_{L^\infty\left(\Sigma_1\right)}\lesssim e^{-tl}\lesssim t^{-1},
    \end{align*}
    where $k=1+le^{i\varphi}$.
\end{proof}

According to Beals-Coifman theory, the solution for $M^{(J)}(k)$ can be given by
\begin{equation*}
M^{(J)}(k) = I + \frac{1}{2\pi i} \int_{\Sigma^{(2)}} \frac{(I + \varpi_J(s))(V^{(2)}(s) - I)}{s - k} \mathrm{d}s,
\end{equation*}
where $\varpi_J \in L^2(\Sigma^{(2)})$ is the unique solution of $(1 - C_{V^{(2)}})\varpi_J = C_{V^{(2)}} I$. And $C_{V^{(2)}}: L^2(\Sigma^{(2)}) \to L^2(\Sigma^{(2)})$ is the Cauchy operator on $\Sigma^{(2)}$, which is defined as:
\begin{align*}
C_{V^{(2)}}(f)(k) = C_{-}f(V^{(2)} - I) = \lim_{s \to k^-, k \in \Sigma^{(2)}} \int_{\Sigma^{(2)}} \frac{f(s)(V^{(2)}(s) - I)}{s - k} \mathrm{d}s.
\end{align*}

Existence and uniqueness of $\varpi_J$ comes from the boundedness of the Cauchy operator $C_{-}$, which admits
\begin{equation*}
\|C_{V^{(2)}}\|_{L^2(\Sigma^{(2)})} \leqslant \|C_{-}\|_{L^2(\Sigma^{(2)}) \to L^2(\Sigma^{(2)})} \|V^{(2)} - I\|_{L^\infty(\Sigma^{(2)})} = \mathcal{O}(t^{-1}).
\end{equation*}

In addition,
\begin{equation*}
    \|\varpi_J\|_{L^2(\Sigma^{(2)})} \lesssim \frac{\|C_{V^{(2)}}\|_{L^2(\Sigma^{(2)})}}{1 - \|C_{V^{(2)}}\|_{L^2(\Sigma^{(2)})}} \lesssim t^{-1}.
\end{equation*}

For the convenience of the last long time asymptotics, we need to give the asymptotic of $M^{(J)}(k)$ as $k\to0$. Denote
\begin{equation*}
    M^{(J)}(k)=M^{(J)}_0+M^{(J)}_1k+\mathcal{O}(k^2),\quad k\to0,
\end{equation*}
we can obtain the following asymptotics as $t\to+\infty$:

\begin{Proposition}\label{p33}
    As $t\to+\infty$, we have
    \begin{align}
     M^{(J)}_0=I+ \mathcal{O}(t^{-1}),\quad \quad M^{(J)}_1=\mathcal{O}(t^{-1}),
    \end{align}

\end{Proposition}

\subsection{Analysis on pure \texorpdfstring{$\bar{\partial}$}{\bar{\partial}}-problem}\label{sec5.3}

Define
\begin{equation}\label{eq74}
    M^{(3)}(k)=M^{(2)}(k)M^{(R)}(k)^{-1},
\end{equation}
$M^{(3)}(k)$ is the solution of a new $\bar{\partial}$-problem as follows:

\begin{dbar}\label{dbar3}
    Find a $2 \times 2$ matrix-valued function $M^{(3)}(k)$ such that
\begin{itemize}
    \item Analyticity: $M^{(3)}(k)$ is continuous in $\mathbb{C}$ and analytic in $\mathbb{C} \setminus \overline{\Omega}$;
    \item Asymptotic behavior: $M^{(3)}(k) = I + \mathcal{O}(k^{-1}), \quad k \to \infty$;
    \item $\bar{\partial}$-Derivative: For $k \in \mathbb{C}$, we have
    \begin{equation*}
        \quad \bar{\partial} M^{(3)}(k) = M^{(3)}(k) W^{(3)}(k),
    \end{equation*}
    with
    \begin{equation*}
        W^{(3)} = M^{(R)}(k) \bar{\partial} R^{(2)}(k) M^{(R)}(k)^{-1}.
    \end{equation*}
\end{itemize}
\end{dbar}

The solution of $\bar{\partial}$-Problem \ref{dbar3} can be solved by the following integral equation
\begin{equation}\label{eq75}
    M^{(3)}(k) = I - \frac{1}{\pi} \iint_\mathbb{C} \frac{M^{(3)}(s) W^{(3)}(s)}{s - k} \, \mathrm{d}A(s).
\end{equation}
Denote $S$ as the Cauchy-Green integral operator
\begin{equation}\label{eq76}
    S\left[f\right](k) = -\frac{1}{\pi} \iint_\mathbb{C} \frac{f(s) W^{(3)}(s)}{s - k} \, \mathrm{d}A(s),
\end{equation}
then () can be written as the following equation
\begin{equation}\label{eq77}
    (1 - S)M^{(3)}(k) = I.
\end{equation}
To prove the existence of the operator at large time, we present the following proposition.

\begin{Proposition}\label{p34}
    Consider the operator $S$ defined by \eqref{eq76},  we can obtain $S : L^\infty(\mathbb{C}) \to L^\infty(\mathbb{C}) \cap C^0(\mathbb{C})$ and
\begin{equation}
    \|S\|_{L^\infty(\mathbb{C}) \to L^\infty(\mathbb{C})} \lesssim t^{-\frac{1}{2}},
\end{equation}
which implies that $(I-S)^{-1}$ exists.
\end{Proposition}

Consider the asymptotic expansion of $M^{(3)}( y, t;k)$ at $k = 0$
\begin{equation*}
    M^{(3)}(y, t;k) = I + M^{(3)}_0(y, t) + M^{(3)}_1(y, t) k + \mathcal{O}(k^2), \quad k \to 0,
\end{equation*}

where
\begin{align}
    M^{(3)}_0(y, t)& = \frac{1}{\pi} \iint_\mathbb{C} \frac{M^{(3)}(s) W^{(3)}(s)}{s} \, \mathrm{d}A(s),\\
    M^{(3)}_1(y, t)& = \frac{1}{\pi} \iint_\mathbb{C} \frac{M^{(3)}(s) W^{(3)}(s)}{s^2} \, \mathrm{d}A(s).
\end{align}

\begin{Proposition}\label{p35}
As $k \rightarrow 0$, $M^{(3)}(y,t;k)$ has the asymptotic expansion:
\begin{align}
|M^{(3)}_0(y,t)|\lesssim t^{-1},\quad |M^{(3)}_1(y,t)|\lesssim t^{-1},\quad as\ t\rightarrow \infty.
\end{align}
\end{Proposition}

\subsection{Proof of Theorem 1-\uppercase\expandafter{\romannumeral2}}
Inverting the transformations \eqref{eq71},\eqref{eq72},\eqref{eq73},\eqref{eq74}, we have
\begin{equation}
    M(k)=M^{(3)}(k)M^{(J)}(k)M^{(out)}(k)R^{(2)}(k)^{-1}T(k)^{-\sigma_3}
\end{equation}
We take $k\to 0$ out of $\Omega$ so that $R^{(2)}(k)=I$. Then by the results of Proposition \ref{p35}, we obtain the proof of Theorem \ref{thm1}-\uppercase\expandafter{\romannumeral2}.

\section{Long-time asymptotics in transition region}\label{sec4}

In this section, we  consider  the asymptotics    in the region $\mathcal{P_-}$ given by
\begin{equation*}
    \mathcal{P}_{-}:= \left\{ (y,t)\in \mathbb{R}\times\mathbb{R}^+: -C< \left(\frac{y}{t}+2\sqrt{3\alpha\beta}\right) t^{\frac{2}{3}}<0 \right\}
\end{equation*}
where $C>0$ is a constant, which corresponds to the case in  Figure \ref{fig:desmos-graph-1.2}. In this region, the four  saddle points $k_j,j=1,2,3,4,$ defined by \eqref{eq42}  approach $\pm k_0$ on the line at least the speed of $t^{-1/3}$ as $t\to +\infty$ with $k_0=\left(\frac{\beta}{48\alpha}\right)^{1/4}$.

 First we make some modifications to the basic  RH problem, which is similar with the method we used in Subsection \ref{subs3.1}.

\subsection{Deformation of the RH problem and hybrid $\bar\partial$-RH problem}\label{sec4.1}
 To start form the RH problem \ref{rhp1}, we first need to decompose the jump matrix and classify the poles. Different from the modification in \eqref{T1}, we keep the jump line of $I$ on the line in this section, which brings up to a new matrix function $T(k)$,

 \begin{equation} \label{eq51}
    T(k)=\prod\limits_{n\in\Delta^-}\frac{k-\bar{z}_n}{k-z_n},
\end{equation}
where $z_n$ and $\Delta^-$ are defined in \eqref{eq43}. Moreover, $T(k)$ has the same properties as in Proposition \ref{p19}.

Make transformation
\begin{equation}\label{eq44}
    N^{(1)}(y,t;k)=M(y,t;k)T(k)^{\sigma_3},
\end{equation}

$N^{(1)}(y,t;k)$ is the solution to the following RH problem.

\begin{RHP}
    Find a $2\times 2$ matrix-valued function $N^{(1)}(k)$ with the following properties:
\begin{itemize}
  \item Analyticity: $N^{(1)}(k)$ is analytical in $\mathbb{C}\setminus\mathbb{R};$

  \item Jump condition: $N^{(1)}(k)$ has continuous boundary values $N^{(1)}_{\pm}(k)$ on $\mathbb{R}$ and
 \begin{equation}
N^{(1)}_+(k)=N^{(1)}_-(k)V^{(1)}(k),
\end{equation}
where
\begin{equation}
V^{(1)}(k)=\begin{cases}
    \left(\begin{array}{cc}1&0\\ \bar{r}(k)T^2(k)e^{-2it\theta}&1\end{array}\right)\left(\begin{array}{cc}1&r(k)T^{-2}(k)e^{2it\theta}\\ 0&1\end{array}\right), &k\in \mathbb{R}\setminus I;\\
\ T(k)^{-\sigma_3}V(k)T(k)^{\sigma_3}, &k\in I;
\end{cases}
\end{equation}

  \item Asymptotic behavior:
  $\ N^{(1)}(k)=I+\mathcal{O}(k^{-1}),\ as \ k\to \infty;$
\item Residue condition: $N^{(1)}(k)$ has simple poles at each $z_n\in\mathcal{Z}\cup\bar{\mathcal{Z}}$, which  has the same residue condition in \eqref{res1}-\eqref{res4}.
\end{itemize}
\end{RHP}

In the transition region, we open the jump contour $\mathbb{R}$ differently, which means the $[k_4,k_3]$ and the $[k_2,k_1]$ parts are kept on the line, while the rest part is opened through $\bar{\partial}$ extension for a fixed small angle $\phi$, which can be shown in Figure \ref{Fig13}. Denote the regions surrounded by $\Sigma_\ell,\ell=1,2,$ as $\Omega_\ell$, and $\Sigma^{(N)}=\Sigma_1\cup\Sigma_1\cup I$.

\begin{figure}[h]
    \centering
   \begin{tikzpicture}[scale=1.55]
\draw[-latex,dotted](-3.7,0)--(3.7,0)node[right]{ \textcolor{black}{Re$k$}};

\draw [brown,thick] (-3.2,-0.6) to [out=0,in=-125] (-2.15,0);
\draw [brown,thick] (-1.18,0) to [out=-55,in=-125] (0,0)  ;

\draw [brown,thick] (0,0) to [out=-55,in=-125] (1.18,0);
\draw [brown,thick] (2.15,0)  to  [out=-55,in=180] (3.2,-0.6);

\draw [teal,thick](-3.2,0.6) to [out=0,in=125] (-2.15,0);
\draw [teal,thick] (-1.18,0)to  [out=55,in=125] (0,0);

\draw [teal,thick](0,0) to [out=55,in=125](1.18,0);
\draw [teal,thick]  (2.15,0)   to [out=55,in=180] (3.2,0.6);

\draw[ thick](-2.15,0)--(-1.18,0);
\draw[ thick](1.18,0)--(2.15,0);

\draw[-latex,teal,thick](-2.35,0.24)--(-2.3,0.19);
\draw[-latex,brown,thick](-2.35,-0.24)--(-2.3,-0.19);

\draw[-latex,teal,thick](2.3,0.19)--(2.35,0.24);
\draw[-latex,brown,thick](2.3,-0.19)--(2.35,-0.24);

\draw[-latex,brown,thick](-0.95,-0.19)--(-0.86,-0.24);
\draw[-latex,teal,thick](-0.95,0.19)--(-0.86,0.24);

\draw[-latex,brown,thick](0.86,-0.24)--(0.95,-0.19);
\draw[-latex,teal,thick](0.86,0.24)--(0.95,0.19);
\draw[-latex,thick](-1.72,0)--(-1.52,0);
\draw[-latex,thick](1.52,0)--(1.72,0);
\node  at (3.3,0.6) {$\Sigma_1$};
\node  at (3.3,-0.6) {$\Sigma_2$};

\coordinate (A) at (1.18,0);
\fill (A) circle (1pt) node[below]{$k_2$};
\coordinate (B)  at (2.15,0);
\fill (B) circle (1pt) node[below]{$k_1$};
\coordinate (C)  at (-1.18,0);
\fill (C) circle (1pt) node[below]{$k_3$};
\coordinate (D)  at (-2.15,0);
\fill (D) circle (1pt) node[below]{$k_4$};
\coordinate (I) at (0,0);
		\fill[red] (I) circle (1pt) node[below,scale=0.9] {$0$};
        \coordinate (a) at (1.77,1.2);
\fill[purple] (a) circle (1pt);

\coordinate (a) at (1.77,-1.2);
\fill[purple] (a) circle (1pt);
\coordinate (a) at (-1.77,1.2);
\fill[purple] (a) circle (1pt);

\coordinate (a) at (-1.77,-1.2);
\fill[purple] (a) circle (1pt);
\coordinate (a) at (0.71,0.6);
\fill[purple] (a) circle (1pt);
\coordinate (a) at (-0.71,0.6);
\fill[purple] (a) circle (1pt);
\coordinate (a) at (0.71,-0.6);
\fill[purple] (a) circle (1pt);
\coordinate (a) at (-0.71,-0.6);
\fill[purple] (a) circle (1pt);
\end{tikzpicture}
    \caption{Opening the jump line $\mathbb{R}\setminus I$ at saddle points  $k_j, \ j=1,\cdots, 4$ with sufficient small angle $\phi$.
 The opened contours   $\Sigma_1$ ($\textcolor{teal}{\bullet} $) and   $\Sigma_2$ ($\textcolor{brown}{\bullet} $)  decay  in   blue region  and   white  region
    in Figure  \ref{figtheta-1}, respectively. The discrete spectrum
    on $\mathbb{C}$  denoted by ($\textcolor{purple}{\bullet} $).   }
    \label{Fig13}
 \end{figure}
\FloatBarrier

 Here, We also need to do some estimates on $\mathrm{Im}\theta(k)$ near the saddle points.

\begin{lemma}\label{p21}\textup{(}near $k=k_j$\textup{)}
     Let $(y,t)\in \mathcal{P}_{-}$, then the following estimates hold for $k$ near $k_j,j=1,2,3,4.$
     \begin{align*}
         \mathrm{Im}\theta(k)&\gtrsim  \mathrm{Im}k\left(\mathrm{Re}k-k_j\right)^2,\quad k\in\Omega_1,\\
         \mathrm{Im}\theta(k)&\lesssim  \mathrm{Im}k\left(\mathrm{Re}k-k_j\right)^2,\quad k\in\Omega_2.
     \end{align*}
\end{lemma}
\begin{proof}
    We only give the proof for $k\in\Omega_1\cap\{k\in\mathbb{C}:\mathrm{Re}k>k_1\}$. Define $k=le^{i\varphi}=k_1+u+vi$, with $u,v\in\mathbb{R}^+,\varphi\in[0,\phi]$, then we have
\begin{equation*}
    v=u\tan\varphi,\quad |k|^2=(u+k_1)^2+\tan ^2\varphi u^2\geqslant k_1^2.
\end{equation*}
By \eqref{eq1}, we have
\begin{equation}
    \xi=\frac{-\beta-48\alpha k_1^4}{4k_1^2}.
\end{equation}
Substitute the above formula into \eqref{Imthata}, we obtain
\begin{align*}
    \mathrm{Im}\theta(k)&=\frac{v}{4k_1^2|k|^2}\left\{48\alpha k_1^2[(u+k_1)^2+\tan ^2\varphi u^2]^2\right.\\
    &\left.-(\beta+48\alpha k_1^4+64\alpha v^2k_1^2)[(u+k_1)^2+\tan ^2\varphi u^2]+\beta k_1^2\right\}.
\end{align*}
By simple calculation and removing the terms $u^4$ and $u^3$, whose coefficient is positive, we get
\begin{equation}
    \mathrm{Im}\theta(k)\gtrsim v\left[h_1(k_1)u^2+h_2(k_1)u\right],
\end{equation}
where
\begin{align*}
    h_1(k_1)&=-\tan^2\varphi(\beta+16\alpha k_1^4)+240\alpha k_1^4-\beta,\\
    h_2(k_1)&=96\alpha k_1^5-2\beta k_1.
\end{align*}
We can find that $h_1(k_1)>0$ for sufficiently small $\phi$, and $h_2(k_1)>0$ for $k_1>k_0$ with $h_2(k_1=k_0)=0$. Therefore,
\begin{equation*}
    \mathrm{Im}\theta(k)\gtrsim  u^2v.
\end{equation*}
For $k\in\Omega_2$, it can be proved similarly.
\end{proof}

\begin{Proposition}\label{p22}
There exist  the functions $R_{\ell}(k)$: $\bar{\Omega}_{\ell}\to \mathbb{C}$, $\ell=1,2$  with the boundary values
	\begin{align} &R_{1}(k)=\left\{\begin{array}{ll}
	r(k)T(k)^{-2}, &\hspace{0.4cm}k\in \mathbb{R},\\[4pt]
 r(k_j)T(k_j)^{-2} ,  &\hspace{0.4cm}k\in \Sigma_{1},\\
	\end{array}\right. \\[5pt]
	&R_{2}(k)=\left\{\begin{array}{ll} \bar{r}(k)T(k)^{2},
 &k\in  \mathbb{R},\\[4pt] \bar{r}(k_j)T(k_j)^{2}, &k\in \Sigma_{2},
	\end{array} \right.
	\end{align}	
where $j=1,\cdots, 4$.
The functions  $R_{\ell}(k), \ell=1,2$  admit the following estimates:
	\begin{align*}
    &|R_{\ell}(k)|\lesssim 1+\left[1+\mathrm{Re}^2(k)\right]^{-\frac{1}{2}},\quad\;for\ k \in\Omega,\\
    &|\bar{\partial}R_{\ell}(k)|\lesssim\chi(\mathrm{Re}k)+|r^\prime(\mathrm{Re}k)|+|k-k_j|^{-\frac{1}{2}}, \quad   \;for\  k\in \Omega,\;j=2,3,\\
	&|\bar{\partial}R_{\ell}(k)|\lesssim|r'(\mathrm{Re}k)|+|k-k_j|^{-\frac{1}{2}}, \quad  \;for\   k\in \Omega, \;j=1,4,\\
&|\bar{\partial}R_{\ell}(k)|\lesssim|k| \quad as\ k\to0,\;for\ k\in \Omega,\\
&\bar{\partial}R_{\ell}(k)=0,\quad \;for\ k\in\mathbb{C}\setminus\Omega,\nonumber
	\end{align*}
where $\chi \in C^{\infty}_0(\mathbb{R},[0,1])$ is a fixed cut-off function with support near 0.
\end{Proposition}
\begin{proof}
    The proof is similar with the proof for Proposition \ref{p11}, which is omitted here.
\end{proof}

Define  a new  function
\begin{equation}
R^{(2)}(k)=\left\{\begin{array}{lll}
\left(\begin{array}{cc}
1 & -R_{1}(k)e^{2it\theta}\\
0 & 1
\end{array}\right), & k\in \Omega_{1},\\
\\
\left(\begin{array}{cc}
1 & 0\\
R_{2}(k)e^{-2it\theta} & 1
\end{array}\right),  &k\in \Omega_{2},\\
\\
I,  &elsewhere.\\
\end{array}\right.
\end{equation}
where   the functions $R_{\ell}(k)$, $\ell=1,2$ are given by  Proposition \ref{p22}.

Make  a transformation
\begin{equation} \label{eq45}
N^{(2)}(k):=N^{(2)}(y,t;k)=N^{(1)}(k)R^{(2)}(k),
\end{equation}
then $N^{(2)}(k)$ is a hybrid RH problem as follows:

\begin{RHP}\label{rhp19}
    Find a $2\times 2$ matrix-valued function $N^{(2)}(k)$ with the following properties:
\begin{itemize}
  \item Analyticity: $N^{(2)}(k)$ is continuous in $\mathbb{C}\setminus\Sigma^{(N)}$, analytical in $\mathbb{C}\setminus(\Omega_1\cup\Omega_2)$ ;
  \item Jump condition: $N^{(2)}(k)$ has continuous boundary values $N^{(2)}_{\pm}(k)$ on $\Sigma^{(N)}$ and
 \begin{equation}
N^{(2)}_+(k)=N^{(2)}_-(k)V^{(2)}_N(k),
\end{equation}
where
\begin{equation}
V^{(2)}_N(k)=\left\{\begin{array}{lll}
\left(\begin{array}{cc}
1 & R_{1}(k)e^{2it\theta}\\
0 & 1
\end{array}\right), & k\in \Sigma_{1};\\
\\
\left(\begin{array}{cc}
1 & 0\\
R_{2}(k)e^{-2it\theta} & 1
\end{array}\right),  &k\in \Sigma_{2};\\
\ T(k)^{-\sigma_3}V(k)T(k)^{\sigma_3}, &k\in I;

\end{array}\right.
\end{equation}

  \item Asymptotic behavior:
  $\ N^{(2)}(k)=I+\mathcal{O}(k^{-1}),\ as \ k\to \infty;$
  \item $\bar{\partial}$-Derivative: For $k\in\mathbb{C}$, we have the $\bar{\partial}$-Derivative equation
\begin{equation}
    \bar{\partial}N^{(2)}(k)=N^{(2)}(k)\bar{\partial}R^{(2)}(k),
\end{equation}
where
  \begin{equation}\label{eq68}
\bar{\partial}R^{(2)}(k)=\left\{\begin{array}{lll}
\left(\begin{array}{cc}
0 & -\bar{\partial}R_{1}(k)e^{2it\theta}\\
0 & 0
\end{array}\right), & k\in \Omega_{1};\\
\\
\left(\begin{array}{cc}
0 & 0\\
\bar{\partial}R_{2}(k)e^{-2it\theta} & 0
\end{array}\right),  &k\in \Omega_{2};\\
\\
0,  &elsewhere.\\
\end{array}\right.
\end{equation}

\end{itemize}
\end{RHP}

So far, we have obtained the hybrid  $\bar{\partial}$-RH problem \ref{rhp19} for  $N^{(2)}(k)$  to analyze the long-time asymptotics of the original RH problem \ref{rhp1} for $M(k)$. We construct the solution for $N^{(2)}(k)$ by the following two steps.
\begin{enumerate}
    \item We first remove the  $\bar{\partial}R^{(2)}\neq0$ part of the solution $N^{(2)}(k)$ and demonstrate the existence of a solution for the resulting pure RH problem $N^{(2)}_{RHP}(k)$. Furthermore, we calculate its asymptotics.
    \item Separating off the solution of the first step,  a pure $\bar{\partial}$-problem $N^{(3)}(k)$ can be obtained. Then, we establish the asymptotic solution to this problem.
\end{enumerate}

\subsection{Analysis on a pure RH problem}\label{sec4.3}
First, we give the pure RH problem $N^{(2)}_{RHP}(k)$.
\begin{RHP}\label{rhp20}
      Find a $2\times 2$ matrix-valued function $N^{(2)}_{RHP}(k)$ with the following properties:
\begin{itemize}
  \item Analyticity: $N^{(2)}_{RHP}(k)$ is analytic in $\mathbb{C}\setminus \Sigma^{(N)}$ ;
  \item Jump condition: $N^{(2)}_{RHP}(k)$ has continuous boundary values $N^{(2)}_{RHP\pm}(k)$ on $\Sigma^{(N)}$ and
 \begin{equation}
N^{(2)}_{RHP+}(k)=N^{(2)}_{RHP-}(k)V^{(2)}_N(k);
\end{equation}
\item Symmetry:
$N^{(2)}_{RHP-}(k)=\sigma_2 \overline{ N^{(2)}_{RHP-}(\bar{k})}\sigma_2=\sigma_2 N^{(2)}_{RHP-}(-k)\sigma_2;$

  \item Asymptotic behavior:
  $\ N^{(2)}_{RHP}(k)=I+\mathcal{O}(k^{-1}),\ as \ k\to \infty;$
  \item $\bar{\partial}$-Derivative: For $k\in\mathbb{C}$, $\bar{\partial}R^{(2)}(k)=0$.
\item Residue condition: $N^{(2)}_{RHP}(k)$ has simple poles at each $z_n\in\mathcal{Z}\cup\bar{\mathcal{Z}}$, which  has the same residue condition in \eqref{res1}-\eqref{res4}.
\end{itemize}
\end{RHP}

In the Painlev\'e region $\mathcal{P}_-$, the two pair of saddle points are close to $\pm k_0$ respectively. It can be easily found out that the leading part of the solution $N^{(2)}_{RHP}$ comes from discrete spectrum and jump lines in a small neighborhood of $k=k_0$ and $k=-k_0$ as $V^{(2)}_N$ decays exponentially and uniformly outside.

\subsubsection{\texorpdfstring{Localized RH problem near $\pm k_0$}{Localized RH problem near ±k₀}}
The phase factor $t\theta(k)$ can be approximated with the help of scaled spectral variables:
\begin{itemize}
    \item For $k$ close to $k_0$(for small $ \zeta \tau^{-1/3}$),
    \begin{equation}\label{eq46}
        \begin{aligned}
    t\theta(k)&=t\theta(k_0)+\left(y+12\alpha k_0^2t+\frac{\beta t}{4k_0^2}\right)(k-k_0)\\
    &+\left(4\alpha t+\frac{\beta t}{4k_0^4}\right)(k-k_0)^3+\sum_{n=4}^{+\infty}\frac{(-1)^{n+1}\beta t}{4k_0^{n+1}}(k-k_n)^n\\
    &:=t\theta(k_0)+\frac{4}{3}\zeta^3+s\zeta+S(t,\zeta),
    \end{aligned}
    \end{equation}

where the scaled parameters are given by
\begin{align}\label{eq49}
    \zeta=\tau ^{\frac{1}{3}}(k-k_0),\quad s=\frac{\xi+2\sqrt{3\alpha \beta}}{12\alpha}\tau ^{\frac{2}{3}},\quad \tau=12\alpha t.
\end{align}
 The first two terms  $ \frac{4}{3} \zeta^3 + s \zeta $   play  the key  role in
matching  the   Painlev\'e  model  in the local region, and   the  remainder in \eqref{eq46} is given by
\begin{equation}
    S(t,\zeta)=\sum_{n=4}^{+\infty}\frac{(-1)^{n+1}\beta}{48\alpha k_0^{n+1}}\tau^{1-\frac{n}{3}}\zeta^n.
\end{equation}
\item For $k$ close to $-k_0$(for small $\hat \zeta\tau^{-1/3}$),
\begin{equation}
    t\theta(k)=t\theta(-k_0)+\frac{4}{3}\hat{\zeta}^3+ s\hat \zeta+\hat S(t,\hat{\zeta}),
\end{equation}
where
\begin{align}\label{eq50}
    \hat \zeta=\tau ^{\frac{1}{3}}(k+k_0),\quad \hat S(t,\hat{\zeta})=\sum_{n=4}^{+\infty}\frac{\beta}{48\alpha k_0^{n+1}}\tau^{1-\frac{n}{3}}\hat{\zeta}^n.
\end{align}
\end{itemize}
Notice that in the transition region $\mathcal{P}_{-}$, as $t \to + \infty$, according to  formula \eqref{eq42}, two  pair of saddle points merge to $\pm k_0$ in the $k$-plane. There are some properties we need to consider under the rescaling given above. We can find that  two scaled phase points $\zeta_j= \tau^{1/3} ( k-k_0  ), \ j=1,2$ are always in a bounded interval in the $\zeta$-plane. Also, the other two scaled phase points $\hat \zeta_j= \tau^{1/3} ( k+k_0  ), \ j=3,4$ are always in a bounded interval in the $\hat \zeta$-plane.  To simplify the statement, we only consider the rescaling from $k$ to $\zeta$.

\begin{Proposition}\label{p23}
     In the transition region  $\mathcal{P}_{-}$, under scaling transformation \eqref{eq49}, for large enough $t$, we have
  \begin{align}
 \zeta_j\in ( -\left(\alpha^{-3}\beta\right)^{1/4} \sqrt{C}, \left(\alpha^{-3}\beta\right)^{1/4}  \sqrt{C}), \quad \ j=1, 2.
        \end{align}
\end{Proposition}
\begin{proof}
    We take $\zeta=\zeta_1$ on the $\zeta$-plane as an example. Since $k_1\to k_0$ as $t\to +\infty$, we can take $t$ large enough to make sure that $k_0<k_1<2k_0$. By \eqref{eq1}, $k_1$ satisfies the equation
    \begin{equation*}
        48\alpha k_1^2+\frac{\beta}{k_1^2}+4\xi=0.
    \end{equation*}
Take $\eta_1=4\sqrt{3\alpha}k_1+\frac{\sqrt{\beta}}{k_1}>0$, the above formula can be written as
\begin{equation}
\eta_1^2=8\sqrt{3\alpha \beta}-4\xi.
\end{equation}
Moreover, we can obtain
\begin{equation}\label{eq47}
    4\sqrt{3\alpha}(k_1-k_0)^2=\left[\eta_1-4(3\alpha\beta)^{1/4}\right]k_1.
\end{equation}
Recalling the expression of $k_1$ in \eqref{eq42}, which implies that
$\eta_1-4(3\alpha\beta)^{1/4}<-\left(\xi+2\sqrt{3\alpha\beta}\right)$. Take this into \eqref{eq47}, we can obtain
\begin{align*}
    |k_1-k_0|\leqslant \left(\alpha^{-3}\beta\right)^{1/4} \sqrt{C} \tau^{-1/3},\quad| \zeta_1|\leqslant \left(\alpha^{-3}\beta\right)^{1/4} \sqrt{C}.
\end{align*}

\end{proof}

Let  $t$  be  large enough so that $\left(\alpha^{-3}\beta\right)^{1/4} \sqrt{C} \tau^{-1/3+\mu}<\rho_1$ where $0<\mu<1/30$ and  $\rho_1$ is defined as
\begin{equation}\label{eq48}
 0<\rho_1 <  \frac{1}{2} {\rm min}  \{   \operatorname*{min}\limits_{ \lambda,\mu\in  \mathcal{Z}\cup\overline{Z} }  |\lambda-\mu|,     \operatorname*{min}\limits_{z_n\in \mathcal{Z}, \, { \mathrm{Im}[i\theta(k)] =0 }}  | z_n-k|  \}.
 \end{equation}
For a fix constant  $\varepsilon \leqslant \left(\alpha^{-3}\beta\right)^{1/4} \sqrt{C}$, define  two open disks
  \begin{align}
&   U_{\varepsilon}(k_0) = \{k \in \mathbb{C}: |k-k_0|< \varepsilon \tau^{-1/3+\mu}\},\nonumber\\
&  U_{\varepsilon }(0) = \{\zeta \in \mathbb{C}: |\zeta|< \tau^{\mu}\varepsilon \},\nonumber
\end{align}
whose   boundaries are  oriented counterclockwise.
The rescaling  defined by \eqref{eq49}
operates the following  map
\begin{align}
     U_{\varepsilon}(k_0) &\to U_{\varepsilon }(0), \ \
     k \mapsto \zeta= \tau ^{\frac{1}{3}}(k-k_0),
\end{align}
which takes $\Sigma^{N}(k) \cap U_{\varepsilon}(k_0)$ onto $\Sigma^{N}(\zeta)\cap U_{\varepsilon}(0)$, where $\Sigma^{N}(\zeta)=\Sigma^{N}(k(\zeta))$ depicted in Figure \ref{Fig14}.
Proposition \ref{p23}  implies that  for  large $t$,   we have   $k_1, k_2 \in U_{\varepsilon} (k_0)$,
and also  $\zeta_1, \zeta_2 \in U_{\varepsilon} (0)$.

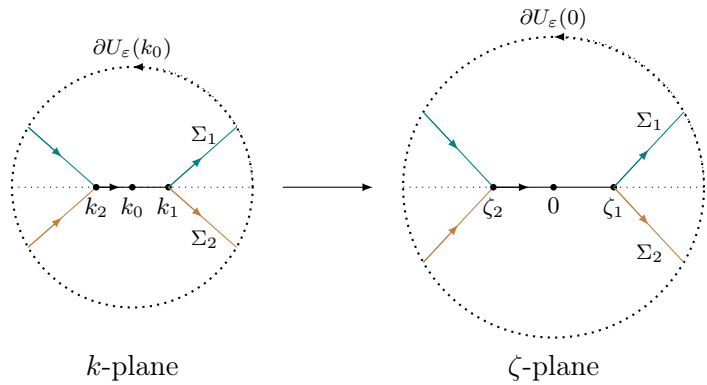
\begin{figure}[H]
 \begin{center}
   \begin{tikzpicture}[scale=0.7]
    \draw[dotted,thick](-3,0) circle (2);
        \draw[teal](-4.7321,1)--(-3.6,0);
        \draw [-latex,teal] (-4.7321,1)--(-4.166,0.5);
        \draw[brown](-4.7321,-1)--(-3.6,0);
        \draw [-latex,brown] (-4.7321,-1)--(-4.166,-0.5);
        \draw[](-3.6,0)--(-2.4,0);
        \draw[-latex](-3.6,0)--(-3.2,0);
        \node[shape=circle,fill=black, scale=0.13]  at (-3.6,0){0};
        \node[below] at (-3.6,0) {\footnotesize $k_2$};
        \node[shape=circle,fill=black, scale=0.13]  at (-3,0){0};
        \node[below] at (-3,0) {\footnotesize $k_0$};
        \node[shape=circle,fill=black, scale=0.13]  at (-2.4,0){0};
        \node[below] at (-2.4,0) {\footnotesize $k_1$};
        \draw[teal](-2.4,0)--(-1.268,1);
        \draw [-latex,teal] (-2.4,0)--(-1.834,0.5);
        \draw[brown](-2.4,0)--(-1.268,-1);
        \draw [-latex,brown] (-2.4,0)--(-1.834,-0.5);
        \draw [dotted] (-5, 0)--(-3.6, 0);
        \draw [dotted] (-3.6, 0)--(-1, 0);

        \node  at (-1.8,0.86) {\scriptsize $\Sigma_{1}$};
\node  at (-1.8,-0.86) {\scriptsize $\Sigma_{2}$};

        \node[]  at (-3,2.3) {\scriptsize $\partial U_{\varepsilon}(k_0)$};
        \draw [dotted,-latex ]  (-1, 0)  to  [out=90,  in=0] (-3, 2);
    \draw[dotted,thick](4,0) circle (2.5);
        \draw[teal](1.8349,1.25)--(3,0);
        \draw [-latex,teal] (1.8349,1.25)--(2.41745,0.625);
        \draw[brown](1.8349,-1.25)--(3,0);
        \draw [-latex,brown] (1.8349,-1.25)--(2.41745,-0.625);
        \draw[](3,0)--(5,0);
        \draw[-latex](3,0)--(3.6,0);
        \node[shape=circle,fill=black, scale=0.13]  at (3,0){0};
        \node[below] at (3,0) {\footnotesize $\zeta_2$};
        \node[shape=circle,fill=black, scale=0.13]  at (4,0){0};
        \node[below] at (4,0) {\footnotesize $0$};
        \node[shape=circle,fill=black, scale=0.13]  at (5,0){0};
        \node[below] at (5,0) {\footnotesize $\zeta_1$};
        \draw[teal](5,0)--(6.1651,1.25);
        \draw [-latex,teal] (5,0)--(5.58255,0.625);
        \draw[brown](5,0)--(6.1651,-1.25);
        \draw [-latex,brown] (5,0)--(5.58255,-0.625);
        \draw [dotted] (1.5, 0)--(3, 0);
        \draw [dotted] (5, 0)--(6.5, 0);

        \node  at (5.6,1.1) {\scriptsize $\Sigma_{1}$};
        \node  at (5.6,-1.1) {\scriptsize ${\Sigma}_{2}$};

        \node[]  at (4,2.8) {\scriptsize $\partial U_{\varepsilon}(0)$};
        \draw [dotted,-latex ]  (6.5, 0)  to  [out=90,  in=0] (4, 2.5);
    \node    at (-3, -3 )  {  $ k$-plane};
    \node    at (4, -3 )  {  $ \zeta$-plane};
 \draw [-latex] (-0.5,0)--(1,0);
    \end{tikzpicture}
    \caption{\footnotesize{ {The   map  between  $U_\varepsilon(k_0)$  and $U_\varepsilon(0)$.}  }}
     \label{Fig14}
  \end{center}
\end{figure}
\FloatBarrier

We show   that when $t$ is sufficiently large, $\xi$ is close  to $-2\sqrt{3\alpha\beta}$, the phase function $ t \theta(k)$ can be  approximated by $ t\theta(k_0)+\frac{4}{3} \zeta^3 + s \zeta $.
For this purpose,  we need the following two lemmas. Lemma \ref{l7} proves that   $S(t,\zeta)$ converges uniformly in  $U_{\varepsilon}(0)$ and  decays  with respect to  $t$. Lemma \ref{l8} proves that $ \left| e^{ \pm 2 i  (\frac{4}{3} \zeta^3 + s\zeta  )} \right|$ is bounded in  $U_\varepsilon(0)$ respectively.

\begin{lemma}\label{l7}
     Let  $(y,t)\in  \mathcal{P}_{-}$,   then for  $\zeta\in U_\varepsilon(0)$, we have
$$| S(t,\zeta)| \lesssim   t^{-\frac{1}{3}+4\mu},  \ t\to +\infty. $$
\end{lemma}

\begin{lemma}\label{l8}
    Let   $(y,t)\in \mathcal{P}_{-}$, then for  large $t$, we have
\begin{align}
  &\mathrm{Im} \left(\frac{4}{3} \zeta^3 + s\zeta \right) \geqslant \frac{8}{3}u^2v, \ \  k\in \Omega_1(\zeta)\cap U_\varepsilon(0),\\
  &\mathrm{Im} \left(\frac{4}{3} \zeta^3 + s\zeta \right) \leqslant \frac{8}{3}u^2v,\ \   k\in \Omega_2(\zeta)\cap U_\varepsilon(0),
\end{align}
where  $ \Omega_{\ell}(\zeta):=\Omega_{\ell}(k(\zeta)),\ell=1,2$, and $\zeta=\zeta_j+u+iv,j=1,2$ are the scaled variables.
\end{lemma}
\begin{proof}
    The proof is similar with Lemma \ref{p21}.
\end{proof}
While under the second rescaling defined in \eqref{eq50}, we can map the disk  $U_{\varepsilon}(-k_0)$ to  $U_{\varepsilon}(0)$ on the $\hat \zeta$-plane similarly. Denote
\begin{align*}
    &U_{\varepsilon}=U_{\varepsilon}(-k_0)\cup U_{\varepsilon}(k_0),\ \
    \Sigma^{(pl,\pm k_0)}=\Sigma^{(N)}\cap U_{\varepsilon}(\pm k_0),\\
    &\Sigma^{(pl)}=\Sigma^{(pl, k_0)}\cup\Sigma^{(pl,- k_0)}.
\end{align*}
Based on the analysis above, we could construct the $N^{(2)}_{RHP}$ by the following scheme
\begin{equation}\label{eq62}
    N_{R H P}^{(2)}(k)= \begin{cases}N^{(err)}(k)M^{(out)}(k), & k \in \mathbb{C} \backslash U_{\varepsilon}, \\ N^{(err)}(k)M^{(out)}(k) N^{(pl)}\left(k\right), & k \in U_{\varepsilon},\end{cases}
\end{equation}
where $M^{(out)}(k)$ solves RH problem \ref{rhp20} as $r=0$. The solution for $M^{(out)}(k)$ is the same in Section \ref{subs3.2}.  $N^{(pl)}\left(k\right)$ is a local RH problem as follows.

\begin{RHP}\label{rhp25}
     Find a   matrix-valued function $N^{(pl)}(y,t;k)$ with the following properties:
       \begin{itemize}
  \item Analyticity: $N^{(pl)}(y,t;k)$ is meromorphic in $\mathbb{C}\setminus \Sigma^{(pl)}$ ;
  \item Jump condition: $N^{(pl)}(y,t;k)$ has continuous boundary values $N^{(pl)}_{\pm}(k)$ on $\Sigma^{(pl)}$ and
 \begin{equation}
N^{(pl)}_+(k)=N^{(pl)}_-(k)V^{(pl)}(k),
\end{equation}
where
\begin{equation}
V^{(pl)}(k)=\left\{\begin{array}{lll}
\left(\begin{array}{cc}
1 & r(k_j)T(k_j)^{-2} e^{2it\theta}\\
0 & 1
\end{array}\right), & k\in \Sigma_{1}\cap\Sigma^{(pl)};\\
\\
\left(\begin{array}{cc}
1 & 0\\
\bar{r}(k_j)T(k_j)^{2}e^{-2it\theta} & 1
\end{array}\right),  &k\in \Sigma_{2}\cap\Sigma^{(pl)};\\
\ T(k)^{-\sigma_3}V(k)T(k)^{\sigma_3},&k\in I\cap U_{\varepsilon};

\end{array}\right.
\end{equation}
  \item Asymptotic behavior:
  $\ N^{(pl)}(y,t;k)=I+\mathcal{O}(k^{-1}),\quad as \ k\to \infty.$
\end{itemize}
\end{RHP}

 The RH problem \ref{rhp25} consists  of the following two local RH  models  near $\pm k_0$

\begin{RHP}\label{rhp21}
       Find a   matrix-valued function $N^{(pl,\pm k_0)}(y,t;k)$ with the following properties:
       \begin{itemize}
  \item Analyticity: $N^{(pl,\pm k_0)}(y,t;k)$ is meromorphic in $\mathbb{C}\setminus \Sigma^{(pl,\pm k_0)}$ ;
  \item Jump condition: $N^{(pl,\pm k_0)}(y,t;k)$ has continuous boundary values $N^{(pl,\pm k_0)}_{\pm}(k)$ on $\Sigma^{(pl,\pm k_0)}$ and
 \begin{equation*}
N^{(pl,\pm k_0)}_+(k)=N^{(pl,\pm k_0)}_-(k)V^{(pl,\pm k_0)}(k),
\end{equation*}
where
\begin{equation*}
V^{(pl,\pm k_0)}(k)=\left\{\begin{array}{lll}
\left(\begin{array}{cc}
1 & r(k_j)T(k_j)^{-2} e^{2it\theta}\\
0 & 1
\end{array}\right), & k\in \Sigma_{1}\cap\Sigma^{(pl,\pm k_0)};\\
\\
\left(\begin{array}{cc}
1 & 0\\
\bar{r}(k_j)T(k_j)^{2}e^{-2it\theta} & 1
\end{array}\right),  &k\in \Sigma_{2}\cap\Sigma^{(pl,\pm k_0)};\\ \\
\ T(k)^{-\sigma_3}V(k)T(k)^{\sigma_3},&k\in I\cap U_{\varepsilon}(\pm k_0);

\end{array}\right.
\end{equation*}
  \item Asymptotic behavior:
  $\ N^{(pl,\pm k_0)}(y,t;k)=I+\mathcal{O}(k^{-1}),\quad as \ k\to \infty.$
\end{itemize}
\end{RHP}

\begin{figure}
    \centering
    \begin{tikzpicture}[scale=0.65]
        \newcommand{\radius}{1.5}
        \newcommand{\gap}{2.8} 


        \draw[dotted,->] (-3*\gap,0) -- (3*\gap,0) node[right] {$\mathrm{Re}k$};
        \draw[dotted,->] (0,-2.3*\radius) -- (0,2.3*\radius) node[above] {$\mathrm{Im}k$};


        \draw[teal,thick] (-1.4*\gap - 1.06, 1.06) -- (-1.4*\gap, 0);
        \draw[brown,thick] (-1.4*\gap, 0) -- (-1.4*\gap - 1.06, -1.06);

        \draw[teal,thick] (-\gap + 1.06, 1.06) -- (-\gap, 0);

        \draw[brown,thick] (-\gap, 0) -- (-\gap + 1.06, -1.06);


        \draw[teal,thick] (\gap - 1.06, 1.06) -- (\gap, 0);
        \draw[brown,thick] (\gap, 0) -- (\gap - 1.06, -1.06);

        \draw[teal,thick] (1.4*\gap + 1.06, 1.06) -- (1.4*\gap, 0);
        \draw[brown,thick] (1.4*\gap, 0) -- (1.4*\gap + 1.06, -1.06);

        \draw[dotted,thick](-1.2*\gap,0) circle (1.88);
        \draw [dotted,-latex ]  (-1.2*\gap+1.88, 0)  to  [out=90,  in=0] (-1.2*\gap, 1.88);

         \draw[dotted,thick](1.2*\gap,0) circle (1.88);
        \draw [dotted,-latex ]  (1.2*\gap+1.88, 0)  to  [out=90,  in=0] (1.2*\gap, 1.88);
        \node[left,scale=0.8] at (-1.4*\gap, 0) {$k_4$};
        \node[right,scale=0.8] at (-\gap, 0) {$k_3$};
        \node[left,scale=0.8] at (\gap, 0) {$k_2$};
        \node[right,scale=0.8] at (1.4*\gap, 0) {$k_1$};
        \draw[-latex] (1.1*\gap,0)--(1.2*\gap+0.001,0);
        \draw[-latex] (-1.3*\gap,0)--(-1.2*\gap+0.01,0);
\node[below,scale=0.8] at (1.2*\gap, 0) {$k_0$};
  \node[below,scale=0.8] at (-1.2*\gap, 0) {$-k_0$};
\draw[-latex,teal,thick](-1.6*\gap,0.2*\gap)--(-1.55*\gap,0.15*\gap);
\draw[-latex,brown,thick](-1.6*\gap,-0.2*\gap)--(-1.55*\gap,-0.15*\gap);

\draw[-latex,brown,thick](-0.9*\gap,-0.1*\gap)--(-0.75*\gap,-0.25*\gap);
\draw[-latex,teal,thick](-0.9*\gap,0.1*\gap)--(-0.75*\gap,0.25*\gap);

\draw[-latex,teal,thick](0.8*\gap,0.21*\gap)--(0.85*\gap,0.15*\gap);
\draw[-latex,brown,thick](0.8*\gap,-0.21*\gap)--(0.85*\gap,-0.15*\gap);

\draw[-latex,brown,thick](1.55*\gap,-0.15*\gap)--(1.65*\gap,-0.25*\gap);
\draw[-latex,teal,thick](1.55*\gap,0.15*\gap)--(1.65*\gap,0.25*\gap);

  \node[scale=1] at (1.2*\gap,2.5) {$\Sigma^{(pl,k_0)}$};
   \node[scale=1] at (-1.2*\gap,2.5) {$\Sigma^{(pl,-k_0)}$};

    \coordinate (I) at (0,0);
		\fill[red] (I) circle (1.2pt) node[below] {$0$};
        \coordinate (a) at (1.4*\gap,0);
		\fill[black] (a) circle (1.7pt);
        \coordinate (b) at (1*\gap,0);
		\fill[black] (b) circle (1.7pt);
        \coordinate (c) at (-1*\gap,0);
		\fill[black] (c) circle (1.7pt);
        \coordinate (d) at (-1.4*\gap,0);
		\fill[black] (d) circle (1.7pt);
        \coordinate (e) at (1.2*\gap,0);
		\fill[black] (e) circle (1.7pt);
        \coordinate (f) at (-1.2*\gap,0);
		\fill[black] (f) circle (1.7pt);
        \draw[thick](1*\gap,0)--(1.4*\gap,0);
         \draw[thick](-1*\gap,0)--(-1.4*\gap,0);
    \end{tikzpicture}
    \caption{Jump contour $\Sigma^{(pl,\pm k_0)}$ of $N^{(pl, \pm k_0)}(k)$.}
     \label{jump-pl}
\end{figure}
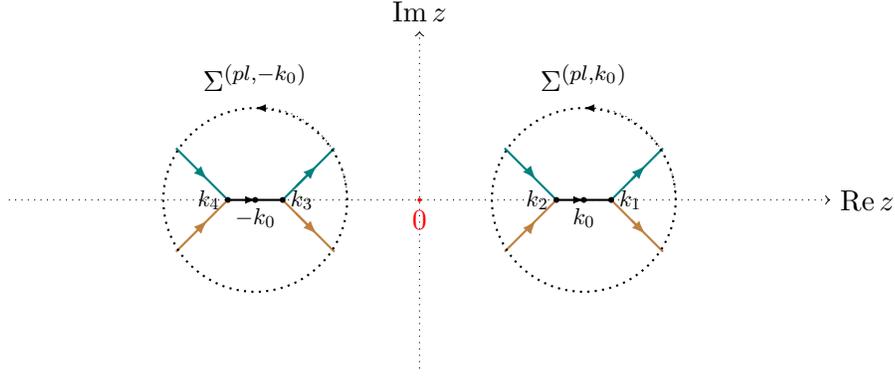
\FloatBarrier

Denote
\begin{align*}
\gamma(k): =r(k)T^{-2}(k),
\end{align*}
then $\gamma(\pm k_0) =r(\pm k_0)T^{-2}(\pm k_0)$.
We show that  in the $U_{\varepsilon}(k_0)$,   $N^{(pl,  k_0)}(k)$ can be approximated by the solution $N^{(\infty,k_0)}(\zeta)$ defined on the disk $U_{\varepsilon}(0)$ in the $\zeta$-plane based on  the following estimates. As for the model $N^{(pl,  -k_0)}(k)$, it can be obtained by the symmetry.

\begin{Proposition}\label{p24}
    Let  $(y,t)\in \mathcal{P}_{-}$,  then
\begin{align}
 & \Big| \widehat\gamma\left(\zeta\right)   e^{2it\widehat\theta \left(\zeta\right)}- \gamma(k_0)   e^{8i\zeta^3/3+2is\zeta+2it\theta(k_0) }  \Big|   \lesssim t^{-1/3+4\mu}, \ \ \ \zeta   \in \left(\zeta_2,  \zeta_1\right), \label{eq55}\\
 &  \Big| \widehat\gamma \left(\zeta_j \right)   e^{2it\widehat\theta \left(\zeta\right)}- \gamma(k_0) e^{8i\zeta^3/3+2is\zeta+2it\theta(k_0) }  \Big|   \lesssim t^{-1/3+4\mu},    \ \  \zeta \in  \Sigma^{(pl,k_0)}(\zeta), \ j=1,2,\label{eq56}
\end{align}
where $\widehat\gamma(\zeta)=\gamma(k(\zeta)),\widehat\theta(\zeta)=\theta(k(\zeta)),\Sigma^{(pl,k_0)}(\zeta)=\Sigma^{(pl,k_0)}(k(\zeta))$ with $k(\zeta)=\tau^{-1/3}\zeta+k_0$, which is defined in \eqref{eq49}.
\end{Proposition}

\begin{proof}
For   $\zeta\in (\zeta_2, \zeta_1)$, $k \in \left(k_2, k_1\right)$,
$$\left|  e^{2i t   \widehat\theta (\zeta) } \right|=1, \quad \left| e^{  i(\frac{8}{3} \zeta^3+ 2s \zeta+2t\theta(k_0))  } \right|=1.$$
  Thus,  we have
\begin{align}\label{eq54}
&\left| \widehat\gamma\left(\zeta\right)   e^{2it\widehat\theta \left(\zeta\right)}- \gamma(k_0)   e^{8i\zeta^3/3+2is\zeta+2it\theta(k_0) }   \right| \nonumber\\
&\leqslant \left|  \widehat\gamma (\zeta) - \widehat\gamma\left(0\right) \right|
+ \left|\widehat \gamma\left(0\right)   \right|       \left|   e^{2iS(t,\zeta)}-   1  \right|.
 \end{align}
 Noticing that $|\zeta|\lesssim \tau^\mu$,  with \eqref{eq49}, we have
  \begin{align}\label{eq52}
& \left|\widehat\gamma (\zeta) - \widehat\gamma (0)  \right|=\left|\gamma (k) - \gamma (k_0)  \right|=     \left|\int_{k_0}^{k} \gamma'(s) ds
    \right| \leqslant \|\gamma' \|_{L^\infty } |k-k_0| \nonumber\\
    &  \leqslant  \|r \|_{H^1} |\zeta|    t^{-1/3}  \lesssim t^{-1/3+\mu}.
 \end{align}
By Lemma \ref{l7},
  \begin{align}\label{eq53}
&  \left|   e^{2iS(t,\zeta)}-   1  \right| \leq e^{|S(t;k)|}- 1   \lesssim  t^{-1/3+4\mu}.
 \end{align}
 Substituting \eqref{eq52} and \eqref{eq53} into \eqref{eq54} gives the estimate  \eqref{eq55}.

For $\zeta \in  \Sigma^{(pl,k_0)}(\zeta)$,
\begin{align*}
   & \Big| \widehat\gamma \left(\zeta_j \right)   e^{2it\widehat\theta \left(\zeta\right)}- \gamma(k_0) e^{8i\zeta^3/3+2is\zeta+2it\theta(k_0) }  \Big| \nonumber\\
   &\leqslant \left|\widehat\gamma \left(\zeta_j \right) \right| \left| e^{8i\zeta^3/3+2is\zeta} \right|\left| e^{2iS(t,\zeta)}-1\right| \nonumber+ \left|e^{8i\zeta^3/3+2is\zeta} \right| \left| \widehat\gamma (\zeta_j) - \widehat\gamma (0) \right|.
\end{align*}

By Lemma \ref{l8}, $ \left| e^{8i\zeta^3/3+2is\zeta}\right|$ is bounded on $\widehat\Sigma^{(pl,k_0)}$.
Similarly to the case on the real axis, we can obtain the estimate \eqref{eq56}.

\end{proof}

As $t\to +\infty$, $N^{(pl,k_0)}(k)$  can be approximated by the following  RH problem.



\begin{RHP}\label{rhp22}
    Find   a $2\times2$ matrix function $N^{(\infty,k_0)}( \zeta)=N^{(\infty,k_0)}(\zeta;s )$   with the following properties:
	\begin{itemize}
		\item Analyticity: $N^{(\infty,k_0)}(\zeta)$ is analytical in $ \mathbb{C} \setminus \Sigma^{\infty}$ with
        $$\Sigma^{\infty}=[\zeta_2,\zeta_1]\cup\{\zeta_1+\mathbb{R}^+e^{\pm i\phi}\}\cup\{\zeta_2+\mathbb{R}^-e^{\pm i\phi}\};$$
		\item Jump condition: $N^{(\infty,k_0)}(\zeta)$  satisfies the jump condition\begin{equation*}
			N^{(\infty,k_0)}_+( \zeta)=N^{(\infty,k_0)}_-(\zeta)V^{(\infty,k_0)}(\zeta),
		\end{equation*}
		where
		\begin{align}
			V^{(\infty,k_0)}(\zeta)= \begin{cases}
				 \left( \begin{array}{cc}
					1& 0\\
					\bar r_0   e^{ -2i \left(\frac{4}{3}  \zeta^3 +  s\zeta \right) }   & 1
				\end{array}\right):=C_-,  &k\in \{\zeta_1+\mathbb{R}^+e^{- i\phi}\}\cup\{\zeta_2+\mathbb{R}^-e^{ i\phi}\},\\
				 \left(	\begin{array}{cc}
					1&   r_0 e^{ 2i \left(\frac{4}{3}  \zeta^3 +  s\zeta \right) }  \\
					0 & 1
				\end{array}\right) :=C_+, & k \in  \{\zeta_1+\mathbb{R}^+e^{ i\phi}\}\cup\{\zeta_2+\mathbb{R}^-e^{- i\phi}\},\\
			\ C_- C_+,\quad & k \in [\zeta_2,\zeta_1],
			\end{cases}
		\end{align}
     with  $r_0=r(k_0)T^{-2}(k_0)e^{2it\theta(k_0)}$. The jump contour for $N^{(\infty,k_0)}(\zeta)$  is given by  Figure \ref{Fig16};
		\item  Asymptotic behavior: $N^{(\infty,k_0)}( \zeta)=I+\mathcal{O}(\zeta ^{-1}),	\quad \zeta \to  \infty.$	
	\end{itemize}
\end{RHP}

\begin{figure}[H]
	\begin{center}
		\begin{tikzpicture}[scale=1.4]
		\draw [dotted ](-6.5,0)--(-0.5,0);
			\draw [black ](-4.5,0)--(-2.5,0);
			\draw [black, -latex](-4.5,0)--(-4,0);
			\draw [black, -latex](-3.5,0)--(-3,0);
			\draw [ black, ](-4.5,0)--(-6.5,1);
			\draw [black,-latex ] (-6.5,1)--(-11/2,1/2);
			\draw [ black, ](-4.5,0)--(-6.5,-1);
			\draw [black,-latex ] (-6.5,-1)--(-11/2,-1/2);
			\draw [black,  ](-2.5,0)--(-0.5, 1 );
			\draw [black,-latex ] (-2.5,0)--(-1.5, 1/2);
			\draw[ black, ](-2.5,0)--(-0.5,-1);
			\draw [black,-latex ] (-2.5,0)--(-1.5,-1/2);
			\node[black]    at (-4.5,0)  {\scriptsize $\bullet$};
			\node[black]    at (-2.5,0)  {\scriptsize $\bullet$};
			\node    at (-3.5,0)  {\scriptsize $\bullet$};
			\node    at (-4.5,-0.2)  {\scriptsize $\zeta_2$};
			\node    at (-2.4,-0.2)  {\scriptsize $\zeta_1$};
				\node    at (-3.5,-0.2)  {\scriptsize $0$};
			
		\end{tikzpicture}
	\end{center}
	\caption{\footnotesize The jump contour  $\Sigma^{\infty}$.}
     \label{Fig16}
\end{figure}

\FloatBarrier

By using  the estimates given in Proposition \ref{p24}, we have

\begin{Proposition}\label{p25}  Let  $(y,t)\in \mathcal{P}_{-}$, then for  large $t$, and $\zeta\in U_{\varepsilon}(0)$, we have
\begin{align*}
 V^{(pl,k_0)}\left(k\right) &= V^{(\infty,k_0)}(\zeta)+\mathcal{O}(t^{-1/3+4\mu}),\\
N^{(pl,k_0)}\left(k\right)& = N^{(\infty,k_0)}(\zeta)+\mathcal{O}(t^{-1/3+4\mu}),
\end{align*}
 where $\mu$ is a constant with $0<\mu<1/30$.
\end{Proposition}

The above RH problem \ref{rhp22} can be transformed into a standard Painlev\'{e} \uppercase\expandafter{\romannumeral2} model
through an  appropriate  deformation.
For this purpose, we add four auxiliary lines  crossing through the point $\zeta=0$
, which can
 divide the complex plane into eight regions $\Omega_n, n=1,\cdots,8$ along with the original contour $\Sigma^\infty$ . See Figure \ref{Fig17}.

Define a sectional matrix function
\begin{align}
    &P(\zeta) =\begin{cases}
    C_+  ^{-1},\  \ & \zeta\in \Omega_2\cup\Omega_4,\cr
 C_-,\  \ & \zeta\in \Omega_6\cup\Omega_8,\cr
I,\  \ & \zeta \in \Omega_1\cup\Omega_3\cup\Omega_5\cup\Omega_7,
. \nonumber
\end{cases}
\end{align}
and  make a transformation
\begin{align}\label{eq57}
&\widehat N^{P} (\zeta)=N^{(\infty,k_0)} (\zeta)P(\zeta),
\end{align}
we can obtain a  Painlev\'e model.

\begin{figure}[H]
\begin{center}
\begin{tikzpicture}[scale=1.4]

		\draw [dotted,-latex ](-6.8,0)--(-0.3,0);
			\draw [black ](-4.5,0)--(-2.5,0);
			\draw [black, -latex](-4.5,0)--(-4,0);
			\draw [black, -latex](-3.5,0)--(-3,0);
			\draw [ black, ](-4.5,0)--(-6.5,1);
			\draw [black,-latex ] (-6.5,1)--(-11/2,1/2);
			\draw [ black, ](-4.5,0)--(-6.5,-1);
			\draw [black,-latex ] (-6.5,-1)--(-11/2,-1/2);
			\draw [black,  ](-2.5,0)--(-0.5, 1 );
			\draw [black,-latex ] (-2.5,0)--(-1.5, 1/2);
			\draw[ black, ](-2.5,0)--(-0.5,-1);
			\draw [black,-latex ] (-2.5,0)--(-1.5,-1/2);

			\node[black]    at (-4.5,0)  {\scriptsize $\bullet$};
			\node[black]    at (-2.5,0)  {\scriptsize $\bullet$};
			
			\node    at (-4.5,-0.2)  {\scriptsize $\zeta_2$};
			\node    at (-2.4,-0.2)  {\scriptsize $\zeta_1$};

 \draw[dotted,thick,black](-5.5,-1)--(-1.5, 1 );
 \draw[dotted,thick, -latex,black  ](-5.5,-1)--(-4.5, -0.5 );
 \draw[dotted,thick, -latex,black  ](-3.5,0)--(-2.5, 0.5 );
\draw[dotted,thick,black](-5.5,1)--(-1.5,-1 );
  \draw[dotted,thick, -latex ,black ](-5.5, 1)--(-4.5, 0.5 );
   \draw[dotted,thick, -latex,black   ](-3.5,0)--(-2.5, -0.5  );

  \node    at (-1.3, -0.2)  {\scriptsize $\Omega_1$};
    \node    at (-1.6, 0.7)  {\scriptsize $\Omega_2$};
 \node    at (-3.5,0.7)  {\scriptsize $\Omega_3$};
 \node    at (-5.4, 0.7)  {\scriptsize $\Omega_4$};
 \node    at (-6, -0.2)  {\scriptsize $\Omega_5$};
  \node    at (-5.4,-0.7)  {\scriptsize $\Omega_6$};
  \node    at (-3.5, -0.7)  {\scriptsize $\Omega_7$};
    \node    at (-1.6, -0.7)  {\scriptsize $\Omega_8$};

\node    at (-3.5,-0.2)  {\scriptsize $0$};
\node    at (-3.5,0)  {\scriptsize $\bullet$};
\end{tikzpicture}
\end{center}
\caption {\footnotesize  Add four auxiliary lines on  the jump contour  of  $N^{(\infty,k_0)} (\zeta)$, by which  $N^{(\infty,k_0)}$ can be deformed
into the  Painlev\'e model $\widehat  N^{P}(\zeta)$  with  the jump contour  in   four dotted rays.}
 \label{Fig17}
\end{figure}
\FloatBarrier

\begin{RHP}\label{rhp23}
     Find  a $2\times2$ matrix function $ \widehat  N^{P} ( \zeta)=\widehat N^{P} (\zeta;s )$ with the following properties:
    \begin{itemize}
        \item Analyticity: $\widehat N^{P}( \zeta)$ is analytical in $\mathbb{C}\setminus  \widehat \Sigma^{P} $, where $\widehat \Sigma^{P}=\displaystyle\bigcup_{j=1}^2 \left\{\mathbb{R} e^{i j\pi/3 }  \right\}; $
        \item Jump condition: $\widehat N^{P}( \zeta)$ satisfies the  jump condition \begin{equation*}
           \widehat N^{P}_+( \zeta)=\widehat N^{P}_-(\zeta) \widehat V^{P}(\zeta), \ \ \zeta\in \widehat \Sigma^{P},
        \end{equation*}
        where
       \begin{align}
       \widehat V^{P}(\zeta)= \begin{cases}
        \left( \begin{array}{cc}
       		1& 0\\
       	\bar r_0 e^{-2i \left(\frac{4}{3}  \zeta^3 +  s\zeta \right) }	& 1
       	\end{array}\right),\quad& k\in \mathbb{R}^- e^{ \frac{\pi}{3} i  }\cup \mathbb{R}^- e^{ \frac{2\pi}{3} i  };\\
         \left(	\begin{array}{cc}
       			1&  r_0 e^{2i \left(\frac{4}{3}  \zeta^3 +  s\zeta \right)  }   \\
       			0 & 1
       		\end{array}\right),\quad &k \in  \mathbb{R}^+ e^{ \frac{\pi}{3} i  }\cup \mathbb{R}^+ e^{ \frac{2\pi}{3} i  };
       	\end{cases}
       \end{align}
        \item Asymptotic behavior:  $\widehat N^{P}( \zeta)=I+\mathcal{O}(\zeta ^{-1}),	\quad \zeta \to  \infty. $

\end{itemize}
\end{RHP}

 Unlike  the  case of defocusing mKdV equation  and  defocusing   NLS  equation \cite{zhou-1993}, here    $r_0=r(k_0)T^{-2}(k_0)e^{2it\theta(k_0)}$  may be   non-real,  which leads to the fact that   the solution to the RH problem \ref{rhp23} is related to the Painlev\'{e} \uppercase\expandafter{\romannumeral34} equation.  Also
$$|r_0|^2=|r(k_0)|^2=\frac{1}{|a(k_0)|^2}-1 $$
 implies that  $|r_0|$ may be larger than 1.
To reduce  the RH problem \ref{rhp23}  to a new RH problem    associated with  the Painlev\'{e} \uppercase\expandafter{\romannumeral2} equation, we  define
\begin{equation}
r_0 = |r_0 | e^{i\varphi_0}= |r(k_0) | e^{i\varphi_0}, \ \ \varphi_0=\arg r_0. \label{phi0}
\end{equation}
Following the idea in \cite{Monvel}, we make a similar transformation
\begin{equation}\label{eq58}
    {N}^{P}(\zeta) =e^{-i \left( \frac{ \varphi_0}{2} -\frac{\pi}{4}\right)\widehat\sigma_3} \widehat {N}^{P}( \zeta),
\end{equation}
then ${N}^{P}(\zeta)$ satisfies the RH problem.

\begin{RHP}\label{rhp24}
    Find  a $2\times2$ function $    N^{P} ( \zeta)=  N^{P} (\zeta;s )$ with properties:
    \begin{itemize}
        \item Analyticity: $N^{P} ( \zeta)$ is analytical  in $\mathbb{C} \setminus \Sigma^{P} $, where $\Sigma^P=\displaystyle\bigcup_{j=1}^2 \left\{\mathbb{R} e^{{ i} j\pi/3 }  \right\}$, which is shown in Figure \ref{Fig18};
        \item Jump condition: $N^{P} ( \zeta)$ satisfies the  jump condition \begin{equation*}
             N^{P}_+( \zeta)=  N^{P}_-(\zeta)  V^{P}(\zeta), \ \ \zeta\in   \Sigma^{P},
        \end{equation*}
        where
            \begin{align}
       	 {V}^{P}( \zeta)= \begin{cases}
        e^{ i \left(\frac{4}{3} \zeta^3 + s\zeta\right)\widehat{\sigma}_3  } \left( \begin{array}{cc}
       		1&
       		i |r(k_0)|\\0 & 1
       	\end{array}\right),\quad &k \in \mathbb{R}^+ e^{\frac{\pi}{3}i };\\[8pt]
       e^{ i \left(\frac{4}{3}  \zeta^3 + s\zeta\right)\widehat{\sigma}_3  } \left( \begin{array}{cc}
       		1&
       		-i |r(k_0)| \\0& 1
       	\end{array}\right),\quad &k \in   \mathbb{R}^+ e^{\frac{2\pi}{3}i };\\[8pt]
       	 e^{ i \left(\frac{4}{3}  \zeta^3 + s\zeta\right)\widehat{\sigma}_3  } \left(	\begin{array}{cc}
       			1&0\\  i | r(k_0)|  & 1
       		\end{array}\right) ,\quad &k \in \mathbb{R}^- e^{\frac{\pi}{3}i };\\
       e^{ i \left(\frac{4}{3}  \zeta^3 + s\zeta\right)\widehat{\sigma}_3  } \left(	\begin{array}{cc}
       			1& 0\\ -i | r(k_0) |  & 1
       		\end{array}\right) ,\quad& k \in  \mathbb{R}^- e^{\frac{2\pi}{3}i };	
       	\end{cases}
       \end{align}
        \item  Asymptotic behavior: $ N^{P}( \zeta)=I+\mathcal{O}(\zeta ^{-1}),	\quad \zeta \to  \infty. $

\end{itemize}
\end{RHP}

\begin{figure}[H]
	\begin{center}
		\begin{tikzpicture}[scale=1.4]
\node[shape=circle,fill=black,scale=0.15] at (0,0) {0};
			\node[below] at (0,-0.1) {\footnotesize $0$};
			\draw[dotted,thick,-latex] (-2.5,0 )--(2.5,0);
			\draw [black] (-1.5,-1 )--(1.5,1);
			\draw [black,-latex] (0,0)--(0.75,0.5);
			\draw [black,-latex] (0,0)--(0.75,-0.5);
			\draw [black] (-1.5,1 )--(1.5,-1);
			\draw [black,-latex] (0,0)--(-0.75,0.5);
			\draw [black, -latex] (0,0)--(-0.75,-0.5);
			
			\node  at (1.8,1.5) {\footnotesize $\begin{pmatrix} 1 & i|r(k_0)| \\0& 1 \end{pmatrix}$};
			\node  at (-1.8,1.5) {\footnotesize $\begin{pmatrix} 1 & -i|r(k_0)|\\0 & 1 \end{pmatrix}$};
			\node  at (1.8,-1.5) {\footnotesize $\begin{pmatrix} 1 &0\\  -i|r(k_0)|& 1 \end{pmatrix}$};
			\node  at (-1.8,-1.5) {\footnotesize $\begin{pmatrix} 1 &0\\  i|r(k_0)|  & 1 \end{pmatrix}$};
			
		\end{tikzpicture}
		\caption{ \footnotesize { The jump contour of $N^P(\zeta)$.}}
		 \label{Fig18}
	\end{center}
\end{figure}
\FloatBarrier

This  RH problem \ref{rhp24} is actually a special
case of the Painlev\'{e} RH   model 1  in Appendix \ref{appendix2} by setting $N^P(\zeta)=M^P(\zeta)$ with
$$ p=i|r(k_0)|, \ \ q=-i|r(k_0)|, \ \ r=-\frac{p+q}{1+pq}=0.$$
Therefore, the solution $N^P(\zeta)$ has the following asymptotic behavior

\begin{equation}
   {N}^{P}( \zeta) = I + \frac{  {N}_1^{P}(s)}{\zeta}+\mathcal{O}\left(\zeta^{-2}\right),  \quad \zeta\to \infty,
\end{equation}
where $ {N}_1^{P}(s)$ is given by
\begin{align}
N_1^P(s) = \frac{1}{2} \begin{pmatrix} -{ i}\int_s^\infty P^2(z)\mathrm{d}z & P(s) \\ P(s) &  { i}\int_s^\infty P^2(z)\mathrm{d}z \end{pmatrix},\end{align}
with $P(s)$ be a purely   imaginary  solution of the   Painlev\'{e} \uppercase\expandafter{\romannumeral2} equation \eqref{p23-1}.

With  transformations  \eqref{eq57} and \eqref{eq58},  we can expand $N^{(\infty,k_0)} (\zeta)$ along the  region $\Omega_3$ or $\Omega_7$ and obtain
\begin{equation}
N^{(\infty,k_0)} (\zeta) = I + \frac{ N_1^{(\infty,k_0)} ( s)}{\zeta}+ \mathcal{O} \left(\zeta^{-2 }\right), \quad \zeta \to \infty,
\end{equation}
where
 \begin{equation}\label{eq60}
N_1^{(\infty,k_0)} ( s)
=   \frac{{ i}}{2} \begin{pmatrix}
-\int_s^\infty P^2(z)\mathrm{d}z &- e^{i\varphi_0 }P(s) \\  e^{-i\varphi_0 }P(s) &   \int_s^\infty P^2(z)\mathrm{d}z
 \end{pmatrix}.
\end{equation}

By the symmetry between $N^{(pl,k_0)} ( k)$ and $N^{(pl,-k_0)} ( k)$,
\begin{equation}
    N^{(pl,-k_0)} (- k)=\sigma_2N^{(pl,k_0)} ( k)\sigma_2,
\end{equation}
it can be readily calculated that

\begin{equation}
N^{(\infty,-k_0)} (\hat\zeta) = I + \frac{ N_1^{(\infty,-k_0)} ( s)}{\hat\zeta}+ \mathcal{O} \left(\hat\zeta^{-2 }\right), \quad \hat\zeta \to \infty,
\end{equation}
where
 \begin{equation}\label{eq61}
N_1^{(\infty,-k_0)} ( s)
= \frac{{i}}{2} \begin{pmatrix}
  \int_s^\infty P^2(z)\mathrm{d}z &- e^{-i\varphi_0 }P(s) \\  e^{i\varphi_0 }P(s) &-\int_s^\infty P^2(z)\mathrm{d}z
 \end{pmatrix},
\end{equation}
with $P(s)$ as defined in Appendix \ref{appendix2} and $\varphi_0$ as defined in \eqref{phi0}.

We obtain the following asymptotic expansion with devotions from $\pm k_0$.

\begin{Proposition}\label{p26}
    RH problem \ref{rhp25} has a unique solution with the following asymptotics as $t\to +\infty$
    \begin{align}
        N^{(pl)}(k)
        &=I+\tau^{-\frac{1}{3}}\left[\frac{N^{(\infty,k_0)}_1(s)}{k-k_0}+\frac{N^{(\infty,-k_0)}_1(s)}{k+k_0}\right]+\mathcal{O}(t^{-\frac{2}{3}+4\mu}),\label{eq64}
    \end{align}
    where $N^{(\infty,k_0)}_1(s)$ and $N^{(\infty,-k_0)}_1(s)$ are defined as in \eqref{eq60} and \eqref{eq61} respectively. Moreover, $\varphi_0$ can be calculated as
    \begin{equation}
\varphi_0(s,t)=2\theta(k_0,\xi=-2\sqrt{3\alpha\beta})t+2k_0s\tau^{\frac{1}{3}}+\Theta,
    \end{equation}
where
\begin{align}
    \Theta=\arg r(k_0)-4\sum_{n\in\Delta^-}\arg(k_0-z_n),
\end{align}
    with $s=\frac{\xi+2\sqrt{3\alpha \beta}}{12\alpha}\tau ^{\frac{2}{3}},\tau=12\alpha t,k_0=\left(\frac{\beta}{48\alpha}\right)^{1/4},0<\mu<1/30.$
\end{Proposition}

\subsubsection{Small normed RH problem}
By the $N^{(err)}(k)$ we define in \eqref{eq62}, which represents the other part of the pure RH problem $N^{(2)}_{RHP}$ without the jump lines and discrete spectrum, we generate a small normed RH problem.
Define
\begin{equation}
    \Sigma^{(e)}=\partial U_{\varepsilon}\cup\left(\Sigma^{(N)}\setminus U_{\varepsilon}\right),
\end{equation}
see Figure \ref{fig:jump-e}. It's easy to find out that $N^{(err)}$ satisfies the following RH problem.
\begin{RHP}
     Find    a matrix function  $N^{(err)}(k)$ with  properties:
          \begin{itemize}
          	\item  Analyticity: $N^{(err)}(k)$ is analytical in $\mathbb{C}\setminus \Sigma^{(e)}$;
          	\item Jump condition: $N^{(err)}(k)$  takes continuous boundary values $N^{(err)}_\pm(k)$ on $\Sigma^{(e)}$ and
          	\begin{align*}
          		N^{(err)}_{+}(k)=N^{(err)}_-(k)V^{(e)}(k),
          	\end{align*}
          	where the jump matrix is given by
           \begin{align*}
                V^{(e)}(k)= \begin{cases}
                   M^{(out)}(k)V^{(2)}_N(k)M^{(out)}(k)^{-1},\quad& k \in \Sigma^{(N)} \backslash U_\varepsilon;\\
                 M^{(out)}(k)N^{(pl)}(k)M^{(out)}(k)^{-1} , \quad &k \in \partial U_\varepsilon;
                \end{cases}
            \end{align*}
          	\item    Asymptotic behavior:   $N^{(err)}(k)=I+\mathcal{O}(k^{-1}),	\quad  k\to  \infty.$	

          \end{itemize}
\end{RHP}

\begin{figure}[h!]
    \centering
\begin{tikzpicture}[scale=1.65]
\draw[-latex,dotted](-2.7,0)--(2.7,0)node[right]{ \textcolor{black}{Re$k$}};

  \draw [thick](1,0) circle (0.32);

    \draw [thick](-1,0) circle (0.32);

\draw[thick,-latex] (0.68,0) to [out=90,in=180] (1.03,0.312);
\draw[thick,-latex] (-1.32,0) to [out=90,in=180] (-0.97,0.312);

\draw [teal,thick](-2,0.6) to [out=0,in=135] (-1.24,0.25);
\draw [brown,thick] (-2,-0.6) to [out=0,in=-135] (-1.24,-0.25);

\draw [teal,thick](-0.76,0.25) to [out=30,in=120] (0,0);
\draw [brown,thick](-0.76,-0.25) to [out=-30,in=-120] (0,0);

\draw [teal,thick] (0.76,0.25)to [out=150,in=60](0,0) ;
\draw [brown,thick](0.76,-0.25) to [out=-150,in=-60] (0,0);

\draw[-latex,teal,thick](1.55,0.49)--(1.7,0.57);
\draw[-latex,teal,thick](-1.7,0.56)--(-1.55,0.48);
\draw[-latex,brown,thick](-1.7,-0.56)--(-1.55,-0.48);
\draw[-latex,brown,thick](1.55,-0.49)--(1.7,-0.57);

\draw [teal,thick](1.24,0.25) to [out=45,in=180] (2,0.6);

\draw [brown,thick] (1.24,-0.25) to [out=-45,in=-180] (2,-0.6);

\draw[-latex,teal,thick](-0.24,0.25)--(-0.15,0.18);
\draw[-latex,teal,thick](0.15,0.18)--(0.24,0.25);
\draw[-latex,brown,thick](-0.24,-0.25)--(-0.15,-0.18);
\draw[-latex,brown,thick](0.15,-0.18)--(0.24,-0.25);

\coordinate (A) at (1,0);
\fill (A) circle (1pt) node[below]{$k_0$};

\coordinate (C)  at (-1,0);
\fill (C) circle (1pt) node[below]{$-k_0$};

\coordinate (I) at (0,0);
		\fill[black] (I) circle (1pt) node[below,scale=0.9] {$0$};

\end{tikzpicture}

    \caption{Jump contour of $N^{(err)}(k)$.}
    \label{fig:jump-e}
\end{figure}
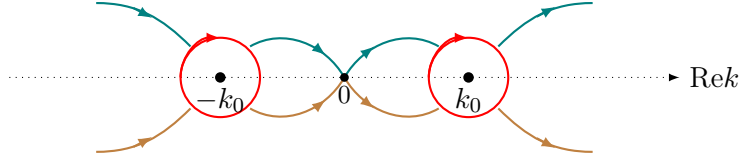
\FloatBarrier

We find that the jump matrix $V^{(e)}$ has the following estimates for $2\leqslant p\leqslant +\infty$ as $t\to +\infty$,
 \begin{equation}
                \|V^{(e)}(k)-I\|_{L^p(\Sigma^{(e)})}=\begin{cases}
                    \mathcal{O}(e^{-c t^{3\mu}}),\quad & k \in \Sigma^{(N)} \backslash U_\varepsilon, \\
                    \mathcal{O}(t^{\kappa_p}), \quad &k\in \partial U_\varepsilon,
                \end{cases}
            \end{equation}
where $c$ is a positive constant, $\kappa_2=-1/6-\mu/2,\kappa_\infty=-\mu$.

According to Beals-Coifman theory, the solution for $N^{(err)}(k)$ can be given by
\begin{equation}\label{eq63}
N^{(err)}(k) = I + \frac{1}{2\pi i} \int_{\Sigma^{(e)}} \frac{(I + \varpi_e(z))(V^{(e)}(z) - I)}{z - k} \mathrm{d}z,
\end{equation}
where $\varpi_e \in L^2(\Sigma^{(e)})$ is the unique solution of $(1 - C_{V^{(e)}})\varpi_e = C_{V^{(e)}} I$. And $C_{V^{(e)}}: L^2(\Sigma^{(e)}) \to L^2(\Sigma^{(e)})$ is the Cauchy operator on $\Sigma^{(e)}$, which is defined as:
\begin{align*}
C_{V^{(e)}}(f)(k) = C_{-}f(V^{(e)} - I) = \lim_{z \to k^-, k \in \Sigma^{(e)}} \int_{\Sigma^{(e)}} \frac{f(z)(V^{(e)}(z) - I)}{z - k} \mathrm{d}z.
\end{align*}

Existence and uniqueness of $\varpi_e$ comes from the boundedness of the Cauchy operator $C_{-}$, which admits
\begin{align}
\|C_{V^{(e)}}\|_{L^2(\Sigma^{(e)})} \leqslant \|C_{-}\|_{L^2(\Sigma^{(e)}) \to L^2(\Sigma^{(e)})} \|V^{(e)} - I\|_{L^\infty(\Sigma^{(e)})} = \mathcal{O}(t^{-\mu}).
\end{align}
In addition,
\begin{equation}
    \|\varpi_e\|_{L^2(\Sigma^{(e)})} \lesssim \frac{\|C_{V^{(e)}}\|_{L^2(\Sigma^{(e)})}}{1 - \|C_{V^{(e)}}\|_{L^2(\Sigma^{(e)})}} \lesssim t^{-\mu}.
\end{equation}
On the other hand, $\varpi_e$ can be written as
$$
 \varpi_e=\sum_{j=1}^4 C_{V^{(e)}}^j I+\left(1-C_{V^{(e)}}\right)^{-1}\left(C_{V^{(e)}}^5 I\right),
$$
then we can obtain the following estimates for $j=1,\dots,4,$
\begin{equation}
\|C^j_{V^{(e)}}I\|_{L^2(\Sigma^{(e)})}\lesssim t^{-(1/6+j\mu-\mu/2)},\quad  \|\varpi_e-\sum_{j=1}^4C^j_{V^{(e)}}I\|_{L^2(\Sigma^{(e)})}\lesssim t^{-1/6-9\mu/2}.\label{esti}
\end{equation}

For the convenience of the last long time asymptotics, we need to give the asymptotic of $N^{(err)}(k)$ as $k\to0$. Denote
\begin{equation*}
    N^{(err)}(k)=N^{(err)}_0+N^{(err)}_1k+\mathcal{O}(k^2),\quad k\to0,
\end{equation*}
we can obtain the following asymptotics as $t\to+\infty$:

\begin{Proposition}\label{p27}
    $N^{(err)}_0$ and $N^{(err)}_1$ can be estimated as follows:
    \begin{align}
        N^{(err)}_0=I+\tau^{-\frac{1}{3}}\widehat N^{(err)}_0+\mathcal{O}(t^{-1/3-5\mu    }),\quad N^{(err)}_1=\tau^{-\frac{1}{3}}\widehat N^{(err)}_1+\mathcal{O}(t^{-1/3-5\mu    }),\label{eq66}
    \end{align}
    where
     \begin{align*}
        \widehat N^{(err)}_0&=\frac{1}{k_0} \left(M^{(out)}(k_0)N^{(\infty,k_0)}_1(s)M^{(out)}(k_0)^{-1}
        -\overline{M^{(out)}(k_0)}N^{(\infty,-k_0)}_1(s)\overline{M^{(out)}(k_0)^{-1}}\right),\\
        \widehat N^{(err)}_1&=\frac{1}{k_0^2}\left(M^{(out)}(k_0)N^{(\infty,k_0)}_1(s)M^{(out)}(k_0)^{-1}+\overline{M^{(out)}(k_0)}N^{(\infty,-k_0)}_1(s)\overline{M^{(out)}(k_0)^{-1}}\right),
     \end{align*}
      with $\tau,s$ are defined in \eqref{eq49}, $N^{(\infty,k_0)}_1(s),N^{(\infty,-k_0)}_1(s)$ are given in \eqref{eq60} and \eqref{eq61} respectively.
\end{Proposition}
\begin{proof}
From \eqref{eq63}, $N^{(err)}_0$ can be calculated as
\begin{align}
    N^{(err)}_0&=I + \frac{1}{2\pi i} \oint_{\partial U_\varepsilon} \frac{ V^{(e)}(z)-I}{z} \mathrm{d}z + \mathcal{O}(t^{-1/3-5\mu       }),\\
    &=I+\frac{1}{2\pi i} \oint_{\partial U_\varepsilon} \frac{M^{(out)}(z) \left(N^{(pl)}(z)-I\right)M^{(out)}(z)}{z}^{-1} \mathrm{d}z + \mathcal{O}(t^{-1/3-5\mu    }),\label{eq65}
\end{align}
where the first equation comes from $C_-\left(\frac{1}{(\cdot)\pm k_0}\right)=0$ and the estimates \eqref{esti}. Substitute \eqref{eq64} into \eqref{eq65} and use the residue theorem can we obtain \eqref{eq66}. And the estimate for $N^{(err)}_1$ can be proved similarly. Detailed proof can be seen in \cite{xyz}.
\end{proof}

\subsection{\texorpdfstring{Analysis on pure $\bar{\partial}$-problem}{Analysis on pure dbar-problem}}\label{sec4.4}
Because we have proved the existence of the solution  $N^{(2)}_{RHP}(k)$,  we can use $N^{(2)}_{RHP}(k)$ to reduce $N^{(2)}(k)$ to a pure $\bar{\partial}$-problem which contains the part for $\bar\partial R^{(2)}\neq0$.
Define the function
            \begin{equation}\label{eq67}
                N^{(3)}(k):=N^{(2)}(k)N^{(2)}_{RHP}(k)^{-1}.
            \end{equation}
            By the properties of $N^{(2)}(k)$ and $N^{(2)}_{RHP}(k)$, we find that $N^{(3)}(k)$ satisfies the following $\bar{\partial}$-problem.
\begin{dbar}\label{dbar2}
    Find  a $2\times2$ matrix function $N^{(3)}(k)$ with the following properties:
            \begin{itemize}
                \item Analyticity:  $N^{(3)}(k)$ is continuous in $\mathbb{C}$ and analytic in $\mathbb{C}  \setminus \overline{\Omega}$;
                \item  Asymptotic behavior:  $N^{(3)}(k)=I+\mathcal{O}(k^{-1}), \quad k\to  \infty$;	

                \item $\bar\partial$-Derivative: For $k\in \mathbb{C}$,  $N^{(3)}(k)$ satisfies
                \begin{equation}
                    \bar{\partial} N^{(3)}(k) = N^{(3)}(k) W^{(3)}(k), \nonumber
                \end{equation}
                where
	\begin{align*}
	W^{(3)}(k)=N^{(2)}_{RHP}(k) \bar{\partial} R^{(2)}(k)N^{(2)}_{RHP}(k)^{-1},
	\end{align*}
 and  $ \bar{\partial}  R^{(2)}(k)$ has been given in \eqref{eq68}.
            \end{itemize}

\end{dbar}

The solution of $\bar{\partial}$-Problem \ref{dbar2} can be solved by the following integral equation
\begin{equation}\label{eq69}
    N^{(3)}(k) = I - \frac{1}{\pi} \iint_\mathbb{C} \frac{N^{(3)}(z) W^{(3)}(z)}{z - k} \, \mathrm{d}A(z),
\end{equation}
where $A(z)$ is the Lebesgue measure on $\mathbb{C}$. Denote $J$ as the Cauchy-Green integral operator
\begin{equation}\label{eq70}
    J\left[f\right](k) = -\frac{1}{\pi} \iint_\mathbb{C} \frac{f(z) W^{(3)}(z)}{z - k} \, \mathrm{d}A(z),
\end{equation}
then \eqref{eq69} can be written as the following equation
\begin{equation}
    (1 - J)M^{(3)}(k) = I.
\end{equation}
To prove the existence of the operator at large time, we present the following proposition.

\begin{Proposition}\label{p28}
    For $(y,t)\in\mathcal{P}_-$, consider the operator $J$ defined by \eqref{eq70},  we can obtain $J : L^\infty(\mathbb{C}) \to L^\infty(\mathbb{C}) \cap C^0(\mathbb{C})$ and
\begin{equation}
    \|J\|_{L^\infty(\mathbb{C}) \to L^\infty(\mathbb{C})} \lesssim t^{-\frac{1}{6}}.
\end{equation}
\end{Proposition}
\begin{proof}
    Similar with Proposition \ref{p17},
    \begin{align*}
\|Jf\|_{L^\infty}\lesssim\|f\|_{L^\infty}\iint_{\Omega_\ell}\frac{|\bar{\partial}R_\ell(z)e^{\pm2it\theta}|}{|z-k|}\mathrm{d}A(z),\quad\ell=1,2.
    \end{align*}
We take $\Omega_1\cap\left\{k\in\mathbb{C}:\mathrm{Re}k>k_1\right\}:=\widehat\Omega_1$ as an example, then
\begin{equation*}
\iint_{\widehat\Omega_1}\frac{|\bar{\partial}R_1(z)e^{2it\theta}|}{|z-k|}\mathrm{d}A(z)\lesssim L_1+L_2+L_3+L_4,
\end{equation*}
where
\begin{align*}
   & L_1=\iint_{\widehat\Omega_1\cap\{|z|\leqslant2|k_0|\}} \frac{|r'(\mathrm{Re}z)| e^{-2t \mathrm{Im} \theta}}{|z-k|}\,\mathrm{d}A(z),\;L_2=\iint_{\widehat\Omega_1\cap\{|z|\leqslant2|k_0|\}}\frac{|z - k_1|^{-\frac{1}{2}} e^{-2t \mathrm{Im} \theta}}{|z-k|} \, \mathrm{d}A(z)\\
   & L_3=\iint_{\widehat\Omega_1\cap\{|z|>2|k_0|\}} \frac{|r'(\mathrm{Re}z)| e^{-2t \mathrm{Im} \theta}}{|z-k|}\,\mathrm{d}A(z),\;L_4=\iint_{\widehat\Omega_1\cap\{|z|>2|k_0|\}}\frac{|z - k_1|^{-\frac{1}{2}} e^{-2t \mathrm{Im} \theta}}{|z-k|} \, \mathrm{d}A(z)
\end{align*}
    Denote $z=k_1+u+vi=|z|e^{i\omega},k=x+yi$ with $u,v>0,x,y,\omega\in\mathbb{R}$, then Lemma \ref{l8} and  the Cauchy-Schwartz inequality implies that
\begin{align*}
L_1&\lesssim\int_0^{2k_0\sin\omega}\|r'\|_{L^2(\mathbb{R})}   |v-y|^{-1/2} e^{-tv^3} \mathrm{d}v   \lesssim t^{-1/6},\\
L_3&\lesssim\int_{2k_0\sin\omega}^{+\infty}\|r'\|_{L^2(\mathbb{R})}   |v-y|^{-1/2} e^{-tv^2} \mathrm{d}v   \lesssim t^{-1/4},
\end{align*}
In a similar way, using Lemma \ref{l8} and   the H\"{o}lder inequality with $p>2$ and $1/p+1/q=1$, we obtain
    \begin{align*}
		 L_2  &\lesssim  \int_{0}^{2k_0\sin\omega} v^{1/p-1/2}|v-y|^{1/q-1}e^{-tv^3} \mathrm{d}v \lesssim t^{-1/6},\\
       L_4  &\lesssim  \int_{2k_0\sin\omega}^\infty v^{1/p-1/2}|v-y|^{1/q-1}e^{-tv^2} \mathrm{d}v \lesssim t^{-1/4}.
	\end{align*}
\end{proof}

Consider the asymptotic expansion of $N^{(3)}( k)$ at $k = 0$
\begin{equation*}
    N^{(3)}(k) = I + N^{(3)}_0(y, t) + N^{(3)}_1(y, t) k + \mathcal{O}(k^2), \quad k \to 0,
\end{equation*}

where
\begin{align*}
    N^{(3)}_0(y, t)& = \frac{1}{\pi} \iint_\mathbb{C} \frac{N^{(3)}(z) W^{(3)}(z)}{z} \, \mathrm{d}A(z),\\
    N^{(3)}_1(y, t)& = \frac{1}{\pi} \iint_\mathbb{C} \frac{N^{(3)}(z) W^{(3)}(z)}{z^2} \, \mathrm{d}A(z).
\end{align*}
 We need the asymptotic behavior of $N^{(3)}_0(y, t)$ and $N^{(3)}_1(y, t)$ as $t\to+\infty$.
\begin{Proposition}\label{p29}
As $k \to 0$, $N^{(3)}(y,t;k)$ has the asymptotic expansion:
\begin{align*}
|N^{(3)}_0(y,t)|\lesssim t^{-\frac{1}{2}},\quad |N^{(3)}_1(y,t)|\lesssim t^{-\frac{1}{2}},\quad as\ t\to + \infty.
\end{align*}
\end{Proposition}

\begin{proof}
    For $z$ away from the origin, we take $\Omega_1\cap\left\{k\in\mathbb{C}:\mathrm{Re}k>k_1\right\}:=\widehat\Omega_1$ as an example.
    \begin{align*}
    |N^{(3)}_0(y,t)|_{\widehat\Omega_1}\lesssim Q_1+Q_2+Q_3+Q_4,
    \end{align*}
    where
    \begin{align*}
        & Q_1=\iint_{\widehat\Omega_1\cap\{|z|\leqslant2|k_0|\}} |r'(\mathrm{Re}z)| e^{-2t \mathrm{Im} \theta}\,\mathrm{d}A(z),\;Q_2=\iint_{\widehat\Omega_1\cap\{|z|\leqslant2|k_0|\}}|z - k_1|^{-\frac{1}{2}} e^{-2t \mathrm{Im} \theta} \, \mathrm{d}A(z)\\
   & Q_3=\iint_{\widehat\Omega_1\cap\{|z|>2|k_0|\}} |r'(\mathrm{Re}z)| e^{-2t \mathrm{Im} \theta}\,\mathrm{d}A(z),\;Q_4=\iint_{\widehat\Omega_1\cap\{|z|>2|k_0|\}}|z - k_1|^{-\frac{1}{2}} e^{-2t \mathrm{Im} \theta} \, \mathrm{d}A(z)
    \end{align*}
    Take the notations in Proposition \ref{p28}, By Lemma \ref{l8} and Cauchy-Schwartz inequality, we have
\begin{align*}
		Q_1 & \lesssim \int_{0}^{2k_0 \sin w} \int_{v}^{2k_0\cos\omega-k_1} \left|r'(u) \right|e^{-tv^3} \mathrm{d}u \mathrm{d}v  \lesssim t^{-1/2},\\
Q_3 & \lesssim \int_{2k_0 \sin w}^\infty \int_{2k_0\cos\omega-k_1}^{+\infty} \left|r'(\mathrm{Re}z) \right| e^{-tuv} \mathrm{d}u \mathrm{d}v  \lesssim t^{-3/4}.
	\end{align*}
By Lemma \ref{l8} and H\"{o}lder inequality with $p>2$ and $1/p+1/q=1$, we have
	\begin{align*}
		Q_2 &\lesssim \int_{0}^{2k_0 \sin w} \int_{v}^{2k_0\cos\omega-k_1}   \left| u + i v\right|^{-1/2} e^{-tv^3} \mathrm{d}u  \mathrm{d}v \lesssim t^{-1/2},\\
Q_4 &\lesssim \int_{2k_0 \sin w}^\infty \int_{2k_0\cos\omega-k_1}^{+\infty}  \left| u + i v\right|^{-1/2}e^{-tuv} \mathrm{d}u  \mathrm{d}v \overset{p<4}\lesssim t^{2/p-7/4}\lesssim t^{-3/4}.
	\end{align*}
   We can prove $|N^{(3)}_1(y,t)|_{\widehat\Omega_1}\lesssim t^{-1/2}$ similarly.

   For $z$ near the origin, by the method we used in Proposition \ref{p18} and $|\bar\partial R^{(2)}(z)|\lesssim|z|$ as $z\to 0$ in Proposition \ref{p22}, we obtain
   \begin{align}
       |N^{(3)}_0(y, t) |_{B(0)}\lesssim t^{-1},\quad |N^{(3)}_1(y, t) |_{B(0)}\lesssim t^{-1}.
   \end{align}
   Summering all the conditions we consider above, we can finish the prove.
\end{proof}

\subsection{Proof of Theorem 1-\uppercase\expandafter{\romannumeral2}}
Now we focus the long-time analysis for the WKI-SP equation \eqref{wkisp-1}. Inverting  the sequence  of transformations \eqref{eq44}, \eqref{eq45}, \eqref{eq67}, we have
\begin{equation}
    M(k)=N^{(3)}(k)N^{(err)}(k)M^{(out)}(k)R^{(2)}(k)^{-1}T(k)^{-\sigma_3}.
\end{equation}

We take $k\to 0$ out of $\Omega$ so that $R^{(2)}(k)=I$. Then by the results of Proposition \ref{p26},\ref{p27}, we obtain the follow asymptotic expansion of ${N}(k)$ as $k\to0$:
\begin{align*}
    M(k)=\left[I+\mathcal{O}(t^{-\frac{1}{2}})+\mathcal{O}(t^{-\frac{1}{2}})k\right]\left[N^{(err)}_0+N^{(err)}_1k\right]M^{(out)}(k)\left(T_0+T_0T_1k\right)^{-\sigma_3}+\mathcal{O}(k^2).
\end{align*}
By setting $P\to { i} P$ and by the reconstruction formula of $u(x,t)$, we obtain proof of Theorem \ref{thm1}-\uppercase\expandafter{\romannumeral2}.

\begin{appendices}
\section{Parabolic cylinder   model near saddle points}\label{appendix1}

In this appendix. we describe the
parabolic cylinder RH model near saddle points that was first introduced in \cite{its1981asymptotics} and further in \cite{zhou-1993, borghese-2018-2}.


For $r_0 \in \mathbb{C} $,   let
$$\nu= -\frac{1}{2\pi} \log (1+|r_0|^2),$$
and jump contour   $\Sigma^{p c}=\left\{\mathbb{R} e^{i \phi}\right\} \cup\left\{\mathbb{R} e^{i(\pi-\phi)}\right\}$ is shown in Figure \ref{Fig9}. Then parabolic cylinder RH model is given as
  follows.

  \begin{figure}[!ht]
    \centering
    \hspace*{-1.5cm}
    \begin{tikzpicture}[scale=0.65]
        \draw (-4.28,2.4) -- (4.28,-2.4);
        \draw (-4.28,-2.4) -- (4.28,2.4);
        \draw[dashed] (-4,0) -- (0,0); 

\draw[-latex,scale=1.3](-1.9,1.05)--(-2,1.1);
\draw[-latex,scale=1.3](-1.9,-1.05)--(-2,-1.1);
\draw[-latex,scale=1.3](1.9,1.05)--(2,1.1);
\draw[-latex,scale=1.3](1.9,-1.05)--(2,-1.1);
        \node at (4.6,2.6)  {$\mathbb{R}_+ e^{\phi i}$};
        \node at (-4.6,2.6) {$\mathbb{R}_+ e^{(-\phi + \pi)i}$};
        \node at (-4.6,-2.6) {$\mathbb{R}_+ e^{(\phi - \pi)i}$};
        \node at (4.6,-2.6){$\mathbb{R}_+ e^{-\phi i}$};
\node at (-6,0){$\arg\zeta\in(-\pi,\pi)$};
        \node[below] at (0,0) {\(0\)};
    \end{tikzpicture}
    \caption{The contour \(\Sigma^{pc}\) for the case of $k_1$.}
     \label{Fig9}
\end{figure}
\FloatBarrier

\begin{parabolic}\label{rhp10}
Find a $2\times2$ matrix-valued function $M^{\left(pc \right)}(\zeta) $ satisfies the following conditions:
\begin{itemize}
    \item Analyticity:\ $M^{\left(pc \right)}(\zeta)$ is analytical in $\mathbb{C} \setminus \Sigma^{p c}$;
    \item Jump condition:\ $M^{\left(pc \right)}$ has continuous boundary values $M_{ \pm}^{\left(pc \right)}$ on $\Sigma^{p c}$ and

$$
M_{+}^{\left(pc \right)}(\zeta)=M_{-}^{\left(pc \right)}(\zeta) V^{(p c )}(\zeta), \quad \zeta \in \Sigma^{p c},
$$
where jump matrix is given by
$$
V^{(p c )}(\zeta)= \begin{cases}\zeta^{-i \nu \hat{\sigma}_3} e^{\frac{i \zeta^2}{4} \hat{\sigma}_3}\left(\begin{array}{cc}
1 & r_{0} \\
0 & 1
\end{array}\right), & \zeta \in \mathbb{R}_{+} e^{\phi i}; \\
\zeta^{-i \nu \hat{\sigma}_3} e^{\frac{i \zeta^2}{4} \hat{\sigma}_3}\left(\begin{array}{cc}
1 & 0 \\
\bar{r}_0 & 1
\end{array}\right), & \zeta \in \mathbb{R}_{+} e^{-\phi i};\\
\zeta^{-i \nu \hat{\sigma}_3} e^{\frac{i \zeta^2}{4} \hat{\sigma}_3}\left(\begin{array}{cc}
1 & -\frac{r_0}{1+\left|r_0\right|^2} \\
0 & 1
\end{array}\right), & \zeta \in \mathbb{R}_{+} e^{(\phi-\pi) i}; \\
\zeta^{-i \nu \hat{\sigma}_3} e^{\frac{i \zeta^2}{4} \hat{\sigma}_3}\left(\begin{array}{cc}
1 & 0 \\
-\frac{\bar{r}_{0}}{1+\left|r_0\right|^2} & 1
\end{array}\right), & \zeta \in \mathbb{R}_{+} e^{-(\phi-\pi) i};
\end{cases}
$$
\item Asymptotic behavior:
$M^{\left(pc \right)}(\zeta)=I+M_1^{\left(pc \right)} \zeta^{-1}+\mathcal{O}\left(\zeta^{-2}\right), \quad \zeta \rightarrow \infty.
$
\end{itemize}
\end{parabolic}
This  RH model  \ref{rhp10}  admits a unique solution with asymptotics.

\begin{equation}
    M^{\left(pc \right)}(\zeta)=I+\frac{1}{\zeta}\left(\begin{array}{cc}
0 & -i \beta_{12}  \\
 i \beta_{21}  & 0
\end{array}\right)+\mathcal{O}\left(\zeta^{-2}\right),
\end{equation}
where
\begin{align} &\beta_{12} =\beta_{12}(r_0)  =\frac{\sqrt{2 \pi} e^{ -\frac{  i \pi}{4} -\frac{  \pi \nu }{2}}}  {\bar{r}_{0} \Gamma\left(i \nu \right)}, \ \ \
  \beta_{21}   =\beta_{21}(r_0)  =-\frac{\sqrt{2 \pi} e^{  \frac{  i \pi}{4} -\frac{  \pi \nu }{2}}}  { {r}_{0} \Gamma\left(-i \nu \right)}.\label{eq24}
\end{align}

\section{Painlev\'e     model near merge points}\label{appendix2}

In this Appendix, we outline the RH model   to describe the solution of  the
    Painlev\'{e} \uppercase\expandafter{\romannumeral2} equation
 \begin{equation}\label{p23-1}
 	P_{ss} = 2 P^3 +s P, \quad s \in \mathbb{R}.
 \end{equation}
 The details can be found in  \cite{zhou-1993, DZ2}.

Let  $\Sigma^P$ denote the oriented contour consisting of six rays
$$\Sigma^P = \bigcup_{n=1}^6\left\{  \Sigma_n^P = e^{i(n-1)\frac{\pi}{3}} \mathbb{R}_+ \right\},$$
with associated jump matrix
$V^P: \Sigma^P \to M_2(\mathbb{C})$   as depicted in  Figure \ref{Sixrays}, where $p$, $q$ and $r$ are complex numbers satisfying the relation
 \begin{equation*}
 p+q+r+pqr =0.
 \end{equation*}
 Then the  equation \eqref{p23-1}
  is   related to  a   matrix-valued RH problem as follows.

 \begin{figure}[h!]
 	\begin{center}
 		\begin{tikzpicture}[scale=1.2]
 			\node[shape=circle,fill=black,scale=0.15] at (0,0) {0};
 			\node[below] at (0.3,0.3) {\footnotesize $0$};
 			\draw [] (-2.5,0 )--(2.5,0);
 			\draw [-latex] (0,0)--(1.25,0);
 			\draw [-latex] (0,0 )--(-1.25,0);
 			\draw [] (0,0 )--(2.5,2);
 			\draw [-latex] (0,0)--(1.25,1);
 			\draw [] (0,0 )--(2.5,-2);
 			\draw [-latex] (0,0)--(1.25,-1);
 			\draw [] (0,0 )--(-2.5,2);
 			\draw [-latex] (0,0)--(-1.25,1);
 			\draw [] (0,0 )--(-2.5,-2);
 			\draw [-latex] (0,0)--(-1.25,-1);

 			\node at (1.6,0.2) {\footnotesize$\Sigma^P_1$};
 			\node at (1.1,-1.2) {\footnotesize$\Sigma^P_6 $};
 			\node at (-1.2,1.2) {\footnotesize$\Sigma^P_3$};
 			\node at (-1.6,-0.2) {\footnotesize$\Sigma^P_4$};
 			\node at (1.2,1.2) {\footnotesize$\Sigma^P_2$};
 			\node at ( -1.2,-1.2) {\footnotesize$\Sigma^P_5$};

 					\node  at (3,0) {\footnotesize $\begin{pmatrix} 1 & p \\ 0 & 1 \end{pmatrix}$};
 					\node  at (2.5,2.5) {\footnotesize $\begin{pmatrix} 1 & 0\\ q& 1 \end{pmatrix}$};
 				\node  at (-2.5,2.5)  {\footnotesize $\begin{pmatrix} 1 & r\\ 0& 1 \end{pmatrix}$};
 				\node  at (2.5,-2.5) {\footnotesize $\begin{pmatrix} 1 & 0 \\ r& 1 \end{pmatrix}$};
 				\node  at (-2.5,-2.5) {\footnotesize $\begin{pmatrix} 1 & q\\ 0& 1 \end{pmatrix}$};
 				\node  at (-3,0) {\footnotesize $\begin{pmatrix} 1 & 0 \\ p & 1 \end{pmatrix}$};
 			
 		\end{tikzpicture}
 		\caption{ \footnotesize { The jump contour $\Sigma^P$.}}
 		\label{Sixrays}
 	\end{center}
 \end{figure}
\FloatBarrier

 \begin{painleve}\label{1modp2}
 	Find  a $2\times2$ matrix-valued function $M^{P}(\zeta)=M^{P}(\zeta,s)$ with the following properties:
 	\begin{itemize}
 		\item Analyticity: $ M^{P}(\zeta)$ is analytical in $\mathbb{C}\setminus \Sigma^{P}$;
 		\item Jump condition: $M^{P}(\zeta)$ satisfies the jump condition:
 		\begin{equation*}
 			M^{P}_+(\zeta)=M^{P}_-(\zeta) e^{-{ i} (\frac{4\zeta^3}{3}+s\zeta)\sigma_3} V^P(\zeta) e^{{ i} (\frac{4\zeta^3}{3}+s\zeta)\sigma_3 }, \quad \zeta \in \Sigma^P;\,
 		\end{equation*}
 		where $V^P(\zeta)$ is shown in Figure \ref{Sixrays}.
 		
 		\item Asymptotic behavior:
       \begin{equation}\label{stanp}
 M^P(\zeta) = I +\frac{M_1^P(s)}{\zeta} + \mathcal{O} \left(\zeta^{-2}\right), \; \zeta\to \infty,
 \end{equation}
 where the  coefficient $M_1^P(s)$ is given by
 \begin{align}\label{posee}
 M_1^P(s) = \frac{{ i} }{2} \begin{pmatrix} -\int_{s}^\infty P(z)^2 {\rm d } z & - P(s)  \\ P(s) &  \int_{s}^\infty P(z)^2 {\rm d } z \end{pmatrix}.
 \end{align}
 	\end{itemize}
 \end{painleve}
Then
 \begin{align}\label{up2}
 P(s)=2{ i} \left(\zeta M^{P}(\zeta)\right)_{12} = -2 { i} \left(\zeta M^{P}(\zeta)\right)_{21},
 \end{align}
 solves  the Painlev\'{e} \uppercase\expandafter{\romannumeral2} equation \eqref{p23}.

 Especially, for any $q \in \mathbb{C}$ and
 \begin{equation*}
  p=\bar{q}, \quad r = -\frac{q+\bar{q}}{1+|q|^2} \in \mathbb{R},
 \end{equation*}
 formula \eqref{up2} leads to a global, pure imaginary solution of the Painlev\'{e} \uppercase\expandafter{\romannumeral2}  equation \eqref{p23}. By changing
 $P(s)\to i P(s)$, we obtain  the global real  solution of the following  Painlev\'{e} \uppercase\expandafter{\romannumeral2}  equation
 \begin{equation}\label{p235}
 	P_{ss} =-2 P^3 +s P, \quad s \in \mathbb{R}.
\end{equation}


\end{appendices}
\vspace{10mm}

\noindent\textbf{Acknowledgements}

	This work is supported by  the National Natural Science
	Foundation of China (Grant No. 12271104,  51879045).\vspace{2mm}
	
	\noindent\textbf{Data Availability Statements}
	
	The data that supports the findings of this study are available within the article.\vspace{2mm}
	
	\noindent{\bf Conflict of Interest}
	
	The authors have no conflicts to disclose.

\end{document}